\documentclass[11pt,letterpaper]{article}
\usepackage[utf8]{inputenc}
\usepackage[english]{babel}
\usepackage{csquotes}
\usepackage{placeins}
\usepackage{nicematrix}

\usepackage[margin=0.75in]{geometry} 

\usepackage{setspace}
\usepackage{graphicx}

\usepackage{amsmath, amssymb, amsfonts}
\usepackage{amsthm, mathrsfs, mathtools}
\usepackage{algorithm}
\usepackage{algorithmic}
\numberwithin{equation}{section}
\usepackage{thmtools}

\usepackage[colorlinks=true, allcolors=blue]{hyperref}
\usepackage[capitalise]{cleveref}

\usepackage[backend=biber,style=alphabetic, minbibnames=99,sorting=nyt,sortcites=true,maxbibnames=99]{biblatex}
\usepackage{titling}

\usepackage{authblk}

\usepackage{enumitem} 
\usepackage{caption} 
\usepackage{subcaption} 
\usepackage{booktabs} 
\usepackage{microtype} 
\usepackage{dsfont} 
\usepackage{faktor} 
\usepackage{comment}

\usepackage[normalem]{ulem}

\theoremstyle{plain}
\newtheorem{theorem}{Theorem}[section]
\newtheorem{lemma}[theorem]{Lemma}
\newtheorem{corollary}[theorem]{Corollary}
\newtheorem{proposition}[theorem]{Proposition}
\newtheorem{conjecture}[theorem]{Conjecture}

\theoremstyle{definition}

\theoremstyle{remark}
\newtheorem{remark}[theorem]{Remark}     
\newtheorem*{remark*}{Remark}         
\newtheorem*{conjecture*}{Conjecture}

\newcommand{\bra}[1]{\ensuremath{\langle#1\rvert}}
\newcommand{\fbra}[1]{\ensuremath{\lparen#1\rvert}} 
\newcommand{\ket}[1]{\ensuremath{\lvert#1\rangle}}
\newcommand{\fket}[1]{\ensuremath{\lvert#1\rparen}} 
\newcommand{\braket}[2]{\langle #1  |#2\rangle}
\newcommand{\ketbra}[2]{|#1\rangle\langle #2|  }

\newcommand{\fketbra}[2]{|#1\rparen\lparen #2|  }

\newcommand{\fsandwich}[3]{\lparen #1|#2 |#3\rparen  }

\newcommand{\abs}[1]{\lvert#1\rvert}

\def\id{\mathds{1}}

\newcommand{\vacuumket}{\fket{\Omega}}

\newcommand{\cH}{\mathcal{H}}

\newcommand{\ii}{\mathrm{i}}

\newcommand{\Romannum}[1]{\uppercase\expandafter{\romannumeral #1\relax}}
\usepackage{xcolor}

\iftrue
\newcommand{\xiangling}[1]{{\color{orange}$\big[\![$~\raisebox{.7pt}{X}\!\!\:\raisebox{-.7pt}{X}: \textit{#1}~$]\!\big]$}}
\newcommand{\marco}[1]{{\color{magenta}$\big[\![$~\raisebox{.7pt}{M}\!\!\:\raisebox{-.7pt}{O}: \textit{#1}~$]\!\big]$}}
\newcommand{\fatemeh}[1]{{\color{teal}$\big[\![$~\raisebox{.7pt}{F}\!\!\:\raisebox{-.7pt}{{M}}: \textit{#1}~$]\!\big]$}}
\newcommand{\sadra}[1]{{\color{red!80!black}$\big[\![$~\raisebox{.7pt}{S}\!\!\:\raisebox{-.7pt}{Bo}: \textit{#1}~$]\!\big]$}}
\newcommand{\tommaso}[1]{{\color{green!80!black}$\big[\![$~\raisebox{.7pt}{T}\!\!\:\raisebox{-.7pt}{G}: \textit{#1}~$]\!\big]$}}
\newcommand{\lucas}[1]{{\color{blue}$\big[\![$~\raisebox{.7pt}{L}\!\!\:\raisebox{-.7pt}{T}: \textit{#1}~$]\!\big]$}}
\newcommand{\salman}[1]{{\color{violet}$\big[\![$~\raisebox{.7pt}{S}\!\!\:\raisebox{-.7pt}{Be}: \textit{#1}~$]\!\big]$}}

\fi

\hypersetup{
  pdftitle    = {Fermions are fundamentally more nonlocal than Bosons},
  pdfauthor   = {Fatemeh Moradi Kalarde; Sadra Boreiri; Xiangling Xu; Lucas Tendick; Salman Beigi; Paolo Perinotti; Tommaso Guaita; Marc-Olivier Renou},
  pdfsubject = {Quantum information theory; fermionic information theory; quantum nonlocality},
  pdfkeywords = {quantum information theory, fermionic information theory, fermions, bosons, distinguishable particles, quantum nonlocality, quantum correlations, quantum networks, self-testing, token-counting game, distributed computing, febits, particle statistics, network nonlocality.}
  pdflang     = {en},
  pdfcreator  = {LaTeX with hyperref},
  pdfdisplaydoctitle = true
}

\title{\bfseries{Fermions are fundamentally more nonlocal than Bosons}}

\author[1$\dagger$]{Fatemeh Moradi Kalarde}
\author[2]{Sadra Boreiri}
\author[1]{Xiangling Xu}
\author[1]{Lucas Tendick}
\author[3]{Salman Beigi}
\author[4,5]{Paolo Perinotti}
\author[6]{Tommaso Guaita}
\author[1]{Marc-Olivier Renou}

\affil[1]{Inria, CPHT, LIX, CNRS, Ecole polytechnique, Institut Polytechnique de Paris, Palaiseau, France}
\affil[2]{Department of Applied Physics, University of Geneva, Geneva, Switzerland}
\affil[3]{School of Mathematics, Institute for Research in Fundamental Sciences (IPM), Tehran,
Iran}
\affil[4]{Dipartimento di Fisica dell'Universit\`a di Pavia, via Bassi 6, 27100 Pavia}
\affil[5]{Istituto Nazionale di Fisica Nucleare, Gruppo IV, via Bassi 6, 27100 Pavia}
\affil[6]{Dahlem Center for Complex Quantum Systems, Freie Universität Berlin, Arnimallee 14, 14195 Berlin, Germany}
\affil[$\dagger$]{\scriptsize\texttt{fatemeh.moradi-kalarde@inria.fr}}

\date{\vspace{-3.00em}}

\allowdisplaybreaks

\begin{document}
\maketitle

\begin{abstract}
Bell's theorem shows that entangled quantum particles can exhibit correlations that classical particles cannot reproduce without an additional nonlocal resource, such as communication. In this sense, quantum particles are fundamentally more nonlocal than classical ones, and entanglement becomes unavoidable in physics. Here we prove the analogous result within quantum theory itself: indistinguishable fermions transmitted through a quantum network can generate correlations that distinguishable particles or indistinguishable bosons cannot reproduce without additional communication. In the same sense, fermions are fundamentally more nonlocal than bosons or distinguishable particles, motivating fermionic anticommutation and indistinguishability as unavoidable operational resources. Our result further implies that fermions can strictly surpass all qubit-based protocols for certain distributed computing tasks, demonstrating that a complete understanding of information processing requires going beyond qubits to fermionic information carriers---\emph{febits}.
\end{abstract}

\crefname{section}{Supplementary Material}{Supplementary Materials}
\Crefname{section}{Supplementary Material}{Supplementary Materials}

\par\addvspace{1.6\baselineskip}

\begin{minipage}{0.92\linewidth}
\emph{``One can speculate that, in principle, electrons might not
be fundamental particles but, rather, excitations in a (nonperturbative) system [of] bosons.
Of course, this is only a logical possibility which may or may not be true.''}

\par\medskip

{\raggedleft
Sergey Bravyi and Alexei Kitaev\\
\emph{Annals of Physics} \textbf{298}, 210--226 (2002)\par}
\end{minipage}%

\par\addvspace{1.0\baselineskip}

Physics describes nature in terms of particles and their interactions. In the 19th century, these were conceived as classical objects governed by gravity and electromagnetism. 
However, inconsistencies---notably the instability of the Rutherford model of the atom \cite{Rutherford1911Scattering}---led to the emergence of quantum physics.
The early successes of quantum theory~\cite{Planck1901, Einstein1905, Bohr1913A, Bohr1913B} already suggested that nature cannot be described by classical particles alone.
Yet in 1935, Einstein, Podolsky and Rosen argued that entanglement and other puzzling quantum phenomena indicated that quantum theory might be incomplete, and could admit a deeper classical description \cite{EPR1935}. 
Bell's theorem ruled out this possibility in 1964~\cite{Bell1964}: from only observed correlations and a locality principle, namely that information is carried by physical systems propagating at finite speed, it showed that entangled quantum particles can generate correlations that no local classical description can reproduce~\cite{BellNonlocalityReview}.
Any classical reproduction of these correlations therefore requires an additional nonlocal interactions, such as communication. 
In this sense, quantum particles are fundamentally more nonlocal than classical ones. 

Within quantum theory, particles themselves fall into distinct classes; see Supplemental Material~\ref{appendix:ParticleFormalism}.
A pair of particles may be either distinguishable, or indistinguishable with either bosonic or fermionic exchange statistics.
Strong physical arguments already suggest that indistinguishable fermions cannot be reduced to distinguishable particles or bosons, since they obey anticommutation relations that cannot be simply reproduced by commuting bosons
\cite{jordan_uber_1928}.
Yet, in contrast to Bell's theorem, no such separation has been derived solely from observed correlations with the above locality assumption.
The issue is sharpened by the discovery of superfast encodings \cite{bravyi_fermionic_2002}, which led Bravyi and Kitaev to speculate that fermions, such as electrons, might not be fundamental particles, but rather admit a deeper bosonic description.
Here we show that, for independently prepared systems, such a description is impossible: locally interacting fermionic particles can generate correlations that no bosonic or distinguishable-particle model can reproduce without additional nonlocal interactions or states.
Our proof relies on the same type of minimal assumptions as Bell's theorem, but applies them to a different setting: several distant fermionic systems are exchanged through a quantum network \cite{Tavakoli2022NetworkReview} whose global configuration is not known prior to the experiment.
Our main result 
can be informally stated as:
\par\bigskip

\begin{minipage}{0.94\textwidth}
\emph{
There exists a Gedankenexperiment in which several independently prepared fermions are locally exchanged through a communication network whose global configuration is unknown in advance.
The resulting correlations cannot be reproduced by bosons or distinguishable particles under the same independent preparation and local-exchange conditions.}
\end{minipage}

\bigskip
\noindent
The elementary operations required by our proposed Gedankenexperiment have already been demonstrated experimentally \cite{gonzalez-cuadra_fermionic_2023, bluvstein_quantum_2022}, for instance, with cold atoms, suggesting that such experiments may become accessible in the coming decades.

When viewed as information carriers, bosons and distinguishable particles are essentially equivalent to collections of qubits and can be described within standard quantum information theory\footnotemark[1].

Indistinguishable fermions are fundamentally different: when information is coded in the occupation of fermionic modes, their anticommutation relations give rise to information-theoretic features with no qubit analogue \cite{Verstraete2003Superselection,Bauls2009,DAriano_2014,Dariano2016, friis2016reasonable}. 
Our result thus establishes a fundamental separation between fermionic and qubit-based quantum information\footnotemark[2], in direct analogy with Bell's separation between quantum and classical physics.
From a computer science perspective, our result shows that qubits do not exhaust the full computational power of nature, as we exhibit a distributed computing problem \cite{PelegDistributedComputing, Naor1995} in which fermionic information processing outperforms all qubit-based protocols by requiring strictly less communication. 
More broadly, we demonstrate that a complete understanding of information processing in decentralized distributed settings cannot be achieved within standard qubit-based quantum information theory alone. It instead requires going beyond qubits to information encoded in the presence or absence of indistinguishable fermions in modes---\emph{febits}.

To present our main result, let us imagine a debate between two scientists, Fiona and Quentin. Quentin believes that nature admits fundamentally a description in terms of qubits, whereas Fiona argues that identical fermions are inevitable. 
Fiona exhibits a fermionic Gedankenexperiment and challenges Quentin to reproduce her observations using qubits. For Fiona's experiment, however, any Quentin's qubit-based strategy obeying the same locality constraints must fail this challenge. This formulation mirrors Bell's theorem: just as Bell disproved the conjecture of Einstein, Podolsky and Rosen that classical bits suffice to describe nature by exhibiting quantum correlations that no local hidden-variable model can reproduce, Fiona disproves the conjecture of Quentin that qubits suffice to describe nature by exhibiting fermionic correlations that no qubit-based model can reproduce.

\footnotetext[1]{As detailed in \cref{app:Q-QITinInternalDegOfFreedom}, not all information carried by fermions is fermionic in the information-theoretic sense. Information encoded in internal degrees of freedom, such as the spin of an electron (a typical type of fermion), is described by ordinary qubits within standard quantum information theory.}
\footnotetext[2]{We use the term \emph{qubit-based quantum information} here for simplicity. This does \emph{not} mean we restrict quantum information to only two-dimensional systems, but we rather allow to use any number of qubits (possibly infinite for modeling bosons and continuous variable systems) to represent higher dimensional systems.}

To construct her conclusive experiment (which will be introduced in detail later), Fiona starts by initially considering an experiment involving only two parties, $A_1$ and $A_2$, each initially holding a box, representing the concept of \emph{a mode} in physics, containing a fermion. 
They then exchange their boxes in order to reveal the anticommutation of fermions. 
More precisely, in second quantization\footnotemark[3]
the initial state is written as
\[
\ket{\psi_0} = f_1^\dagger f_2^\dagger \ket{\mathrm{vac}},
\]
where $f_k^\dagger$ is the \emph{creation operator} of a fermion in the $k$-th box, and $\ket{\mathrm{vac}}$ denotes the vacuum state in which all boxes are empty. For simplicity, we assume throughout that all fermions occupy the same internal state (e.g. the same spin state), so that only the occupation of the spatial modes is relevant. After $A_1$ and $A_2$ exchange their boxes, the state becomes
\[
\ket{\psi_1} = f_2^\dagger f_1^\dagger \ket{\mathrm{vac}} = -\ket{\psi_0},
\]
where we used the fermionic anticommutation relation $f_1^\dagger f_2^\dagger = - f_2^\dagger f_1^\dagger$. Fiona aims to reveal this $-1$ exchange phase experimentally, which has no counterpart for qubits or indistinguishable bosons, whose creation operators commute. Although the $-1$  phase is here global, and thus undetectable in quantum mechanics, Fiona could use a refined experiment that can reveal it by introducing a phase reference.
\begin{figure}[t]
    \centering
    \includegraphics[width=0.94\linewidth]{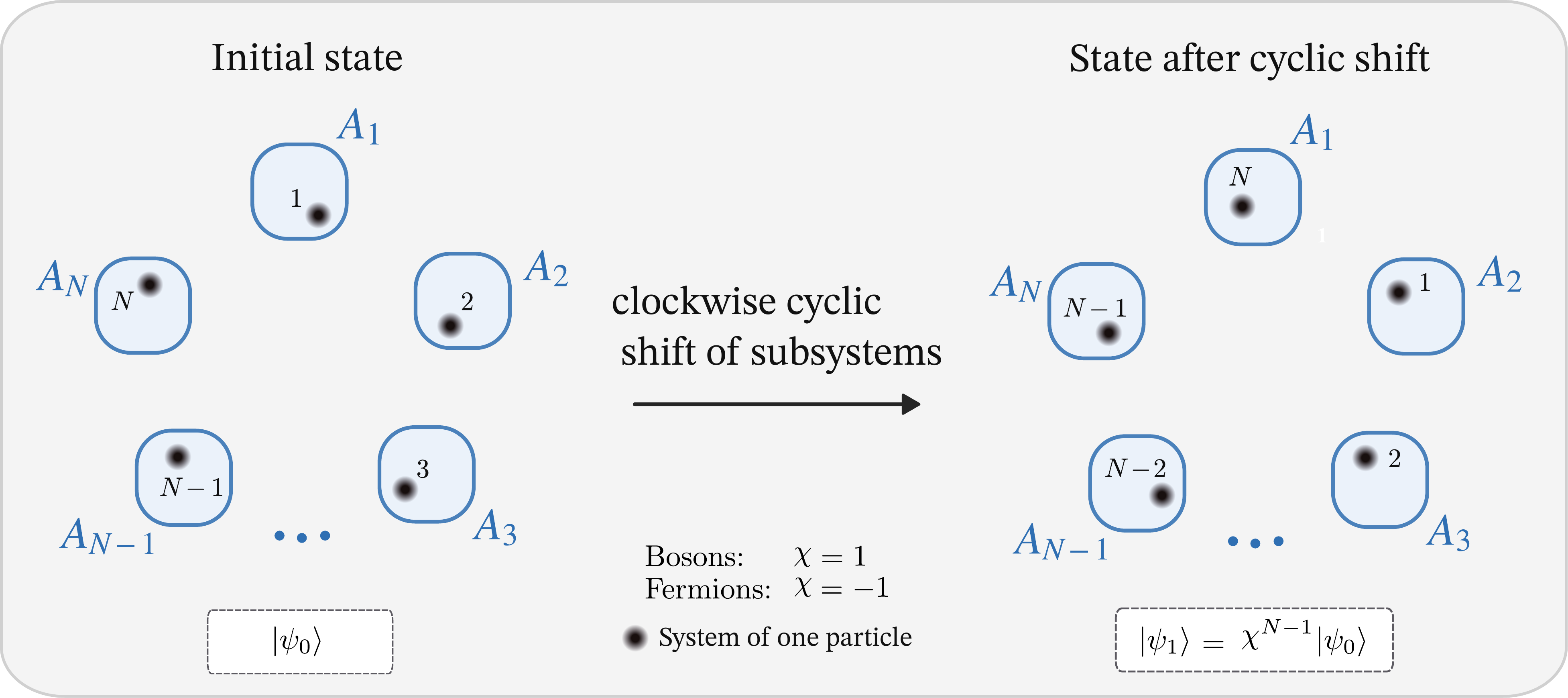}
    \caption{
\textbf{Cyclic exchange and fermionic phase.}
\textbf{(a)} \(N\) parties \(A_1,\ldots,A_N\) initially hold one particle (represented by the black disk) each, in a global state \( |\psi_0\rangle = a_1^\dagger \cdots a_N^\dagger |\mathrm{vac}\rangle \). $a_i^\dagger$ denotes the bosonic or fermionic creation operator of particle at mode $i$, depending on the particle type. \textbf{(b)} After a clockwise cyclic shift of the particles, the state becomes \( |\psi_1\rangle =
a_N^\dagger a_1^\dagger \cdots a_{N-1}^\dagger |\mathrm{vac}\rangle
= \chi^{N-1} |\psi_0\rangle \). For bosons, \(\chi=1\), while for fermions \(\chi=-1\), giving a parity-dependent global phase. 
} 
\label{fig:cyclic-fermionic-phase}
\end{figure}

\footnotetext[3]{
Second quantization provides the natural language for discussing fermionic nonlocality and entanglement. In first quantization, the wavefunction of two spatially separated independently prepared fermions must be antisymmetrized. However, this antisymmetrization alone does not imply the presence of shared resources between the preparation devices. In this sense, independently prepared distant fermions remain operationally independent. $f_1^\dagger f_2^\dagger\ket{\mathrm{vac}}$ is the mathematical description of a two-fermion state, not a preparation procedure in which one physically starts from $\ket{\mathrm{vac}}$ and sequentially applies the creation operators. The ordering of the creation operators is a convention and does not represent a temporal ordering of physical operations. See Supplemental Material~\ref{appendix:FirstQuantization}.
}

\footnotetext[4]{\label{indistinguishability-footnote}A common misconception is that
fermionic anticommutation phenomena only occur when their trajectories physically cross. Instead, indistinguishability is about the impossibility of assigning identities to the particles, which remains true even for far apart identical particles. As explained in Supplemental Material~\ref{appendix:FirstQuantization}, the distant \emph{modes} containing the particles are distinguishable, but the  particles they may contain are not.
}

Indeed, Ref.~\cite{roos2017revealing} shows that this fermionic exchange phase can be detected by first placing each fermion in a coherent superposition between two spatially separated boxes, for instance using a fermionic beam splitter. The parties then exchange one box from each pair while keeping the other fixed. These unexchanged boxes provide a phase reference, allowing the exchange phase acquired by the exchanged boxes to be observed through interference. Similar ideas are used in~\cite{Blasiak2025,Blasiak2019,Marletto2019} to show the nonlocal character of indistinguishable particles over classical physics, tracing back to the results in~\cite{PhysRevA.46.2229, PhysRevLett.68.1251}. Crucially, the indistinguishable character of the fermions used in this scenario and therefore their anticommutation properties does not depend on any spatial proximity, but on the fact that it is impossible to \emph{tell which fermion is which}\footnotemark[4].

However, by contrast to Bell's theorem, 
in these experiments 
the detection of the fermionic character of particles relies on very specific assumptions on the underlying implementation. In particular, Quentin may reproduce the observed correlations within a local qubit-based model by introducing the $-1$ phase through a local operation performed on one of the exchanged boxes (e.g., by $A_1$). Moreover, the protocol of~\cite{roos2017revealing} includes a post-selection step, which could potentially be further exploited by Quentin. Hence, if one truly seeks to rule out any qubit simulation that is only restricted to obey the same locality constraints of the fermionic experiments, one has to avoid the possibility of these simple simulation strategies.  

To overcome this limitation, Fiona considers a second experiment that generalizes the first one to multiple parties. She arranges $N$ parties $A_1,\ldots,A_N$ on a loop, each initially holding a box containing a fermion, with initial state
\[
\ket{\psi_0} = f_1^\dagger \cdots f_N^\dagger \ket{\mathrm{vac}}.
\]
Each party then sends their box to their left neighbor, resulting in the state
\[
\ket{\psi_1} = f_N^\dagger f_1^\dagger \cdots f_{N-1}^\dagger \ket{\mathrm{vac}} = (-1)^{N-1} \ket{\psi_0},
\]
where the phase arises from the anticommutation of fermionic operators. This situation is depicted in \cref{fig:cyclic-fermionic-phase}.

Again, the global phase can be turned into a relative phase by placing the fermions in a superposition between boxes that are kept or sent. Now, by cyclically permuting the modes, the parties collectively encode the parity of $N-1$ in the phase of the state, which has no counterpart for qubits, where a cyclic permutation generates no phase by itself.  
Nonetheless, Quentin could again reproduce this effect within a qubit-based model by having one party, say $A_1$, locally introduce a $-1$ phase conditioned on $N$ being even. However, this crucially requires the party $A_1$ to have access to global information about the network, namely the parity of the number $N$ of parties involved in the permutation loop. This simple insight is enough for Fiona to introduce her final refinement of the experiment that will ultimately make Quentin fail in his task. She considers multiple variants of the protocol by forming loops of size $N$ and $N-1$ (see Fig.~\ref{fig:rewiring-game}), so that no party can presume the size of the loop. Having described the intuition behind Fiona's Gedankenexperiment, we present its precise formulation in the following, where in the simplest case we choose $N=3$. 

To prove to Quentin that fermions are more nonlocal than bosons and distinguishable particles, Fiona finally proposes the following fermionic experiment. 
It involves three separated parties $A$, $B$, and $C$ in distant laboratories who locally prepare fermionic systems from the vacuum and exchange them through a communication network, generating specific correlations revealing the fermionic nonlocality.
Each party initially controls a fermion in their laboratory that is indistinguishable from the fermions in the other two laboratories. Each party places their fermion in a box and also picks an empty box. They then apply a unitary transformation to the two boxes in their possession, preparing a coherent (and balanced) superposition of the fermion being either in the first or the second box (that is, maximally entangled across the two fermionic modes contained in each box). Subsequently, one of the boxes (whose occupation is now unknown) is kept, while the other is sent outside the laboratory to a central and independent referee. The referee simply acts as a relay and redistributes the received systems among the parties according to a configuration $l$ that is unknown to the parties. 
As a result, each party ends up with two boxes: one that they initially had and one received from some other party in the network (or even themselves). The referee then randomly chooses an input setting $x_i \in \lbrace 0,1 \rbrace$ for each party $i$ and sends them the corresponding bit. 

Each party then locally performs an input-dependent projective measurement with four possible outcomes, $a_i \in \{0, 1_+, 1_-, 2 \}$. The outcomes 0 and 2 correspond, respectively, to the party observing zero and two particles in the boxes it now holds. The outcomes $1_+$ and $1_-$ correspond to an input-dependent basis spanning the two-dimensional subspace in which the party holds one particle; see Supplemental Material~\ref{appendix:ExperimentAndMainResult} for a precise definition of these bases.
The observations of the experiment are thus fully characterized by the correlations in the outputs observed in the Gedankenexperiment, that is in the probability distributions $p^{l}(a_A,a_B,a_C|x_A,x_B,x_C)$.

\begin{figure}[t]
    \centering
\includegraphics[width=0.94\textwidth]{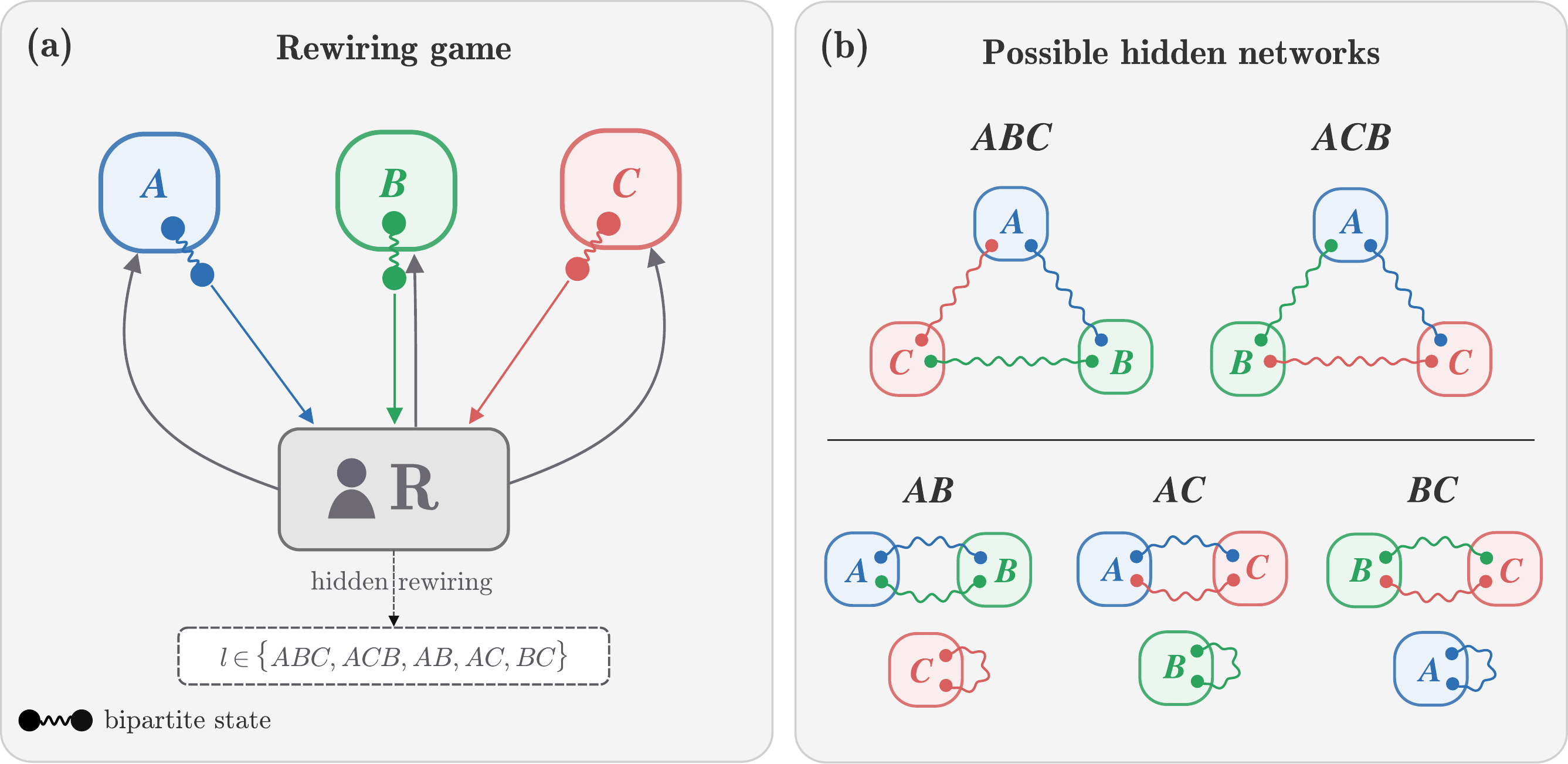}
    \caption{
    \textbf{Distributed rewiring game:}
    \textbf{(a)} Parties $A,B,C$ each locally prepare bipartite systems, keep one subsystem, and send the other to a referee $R$. The referee is a passive relay introduced to formalize the requirement that the local operations of the parties do not depend on the global network configuration. The referee secretly chooses a wiring $l\in\{ABC,ACB,AB,AC,BC\}$ and redistributes the sent subsystems accordingly as illustrated in \textbf{(b)}, while also assigning inputs $x_i \in \lbrace 0,1 \rbrace$. Following the protocol described in the main text, each party $i \in \lbrace A,B,C \rbrace$ performs an input-dependent measurement and obtains an outcome $a_i \in \lbrace 0,1_+, 1_-,2 \rbrace$. The observations of the experiment are characterized by the distribution $p^{l}(a_A,a_B,a_C|x_A,x_B,x_C)$.
    \textbf{(b)} The five possible hidden wirings: two three-party loops, \(ABC\) and \(ACB\), and three two-party loops, \(AB\), \(AC\), and \(BC\), with the excluded party receiving its own subsystem back. Crucially, the chosen wiring $l$ is hidden from the parties, so their local operations cannot depend on which network they are in. }
    \label{fig:rewiring-game}
\end{figure}

The network configuration $l$ chosen by the referee in each run is central in Fiona's proof and deserves further explanation. The referee is merely a conceptual device used to formalize the requirement that the local behavior of the parties is independent on the loop configuration. We therefore assume that it does not provide either Fiona or Quentin an advantage for their strategies.  
In particular, the referee ensures that, throughout the experiment, the global network configuration $l$ remains unknown to the individual parties and any possible simulation strategy that Quentin is employing cannot depend on it\footnotemark[5].
For instance, in some runs, the referee may establish the network configuration $ABC$, meaning that the box sent by $A$ goes to $B$, the box sent by $B$ goes to $C$, and the box sent by $C$ goes to $A$. The referee may also connect the parties in the order $ACB$. In some runs of the experiment, the referee may effectively exclude one of the parties from the experiment by only forming non-trivial network loops including any two of the parties. In that case, there are three possible loops: $AB$, $AC$, and $BC$. The box from the excluded party is then simply sent back to the same party. Notably, the parties do not have access at any point of the experiment even to the size $|l|\in\{2,3\}$ of the loop network in the configuration $l$. See \cref{fig:rewiring-game} for a schematic depiction of the experiment and further explanations. 

\footnotetext[5]{Note that the party sending the box could reveal their identity to the receiving party (for example by sending an additional subsystem that encodes this, or by using distinguishable modes, see also footnote~\ref{indistinguishability-footnote}), who could use this information in their measurement strategy. However, this still does not contain the full information about the loop configuration $l$ and its parity.
}

Finally, after introducing Fiona's fermionic strategy and the rules of the Gedankenexperiment, let us explain the intuition behind why Quentin cannot reproduce Fiona's correlations, that is that there does not exist any hidden model in which the parties locally manipulate arbitrary numbers of qubits and manage to reproduce the same $p^{l}(a_A,a_B,a_C \mid x_A,x_B,x_C)$ for all $l$ as in Fiona's experiment. A formal proof based on this intuition is provided in the Supplemental Materials~\ref{appendix:ExperimentAndMainResult} and \ref{appendix:contradiction}. Rather than analyzing the full distributions $\{ p^l \}_l$, we focus on the part of the distributions corresponding to the outputs $1_+$ and $1_-$. Equivalently, we consider the distribution conditioned on every party obtaining exactly one particle. In Fiona's strategy, the event $a_i = 1_\pm , \forall i \in \{A, B, C \}$ can occur in two ways: either every party kept their fermion, or every party sent their fermion to the referee. Crucially, these two possibilities occur coherently in superposition. As a result, the corresponding conditional states for loops of respective sizes $|l|=2$ and $|l|=3$ are
\[
\ket{\psi}_{|l|=2}
=
\frac{1}{\sqrt{2}}
\left(
\ket{01,01}
-
\ket{10,10}
\right),
\qquad
\ket{\psi}_{|l|=3}
=
\frac{1}{\sqrt{2}}
\left(
\ket{01,01,01}
+
\ket{10,10,10}
\right).
\]
Here, the local states \ket{01} and \ket{10} represent, respectively, the events in which a fermion is kept or sent. The global states are Greenberger--Horne--Zeilinger (GHZ) states~\cite{GHZPaperavailable} in the local basis $\{\ket{01},\ket{10}\}$, where the two GHZ components correspond, respectively, to the situations in which all parties keep their fermions and all parties send their fermions.
Importantly, the relative phase $\pm 1$ between the two terms depends on the parity of the loop size $|l|$, thereby encoding a global property of the communication network into the quantum state using only local interactions.
In fact, Fiona's PVM elements in this subspace (the projectors corresponding to outputs $1_\pm$) which depend on $x_A,x_B,x_C$ implement Bell tests designed to distinguish these GHZ states~\cite{baccari2020scalable}.
What can Quentin do to reproduce these correlations using qubits? A natural attempt is to imitate Fiona's protocol using bosonic particles to implement qubit systems instead of fermions. Each party prepares a bipartite state $(\ket{01}+\ket{10})/\sqrt{2}$, keeps the first subsystem, and sends the second to the referee. Performing a similar measurement and conditioning again on the event $a_i =1_\pm, \, \forall i \in \{A, B, C\}$, this strategy indeed produces similar GHZ states, but now always with a positive relative phase, independent of the loop size.
A seemingly simple fix would be that some parties locally introduce an additional phase. That is, each party $i \in \{A, B, C\}$ now prepares the state $(\ket{01}+e^{\ii\varphi_i}\ket{10})/\sqrt{2}$. However, Quentin faces a fundamental obstruction: the local states prepared by the parties cannot depend on the loop configuration $l$, since this configuration is unknown when the systems are emitted into the network. But there exists no fixed $l$-independent assignment of the values $\varphi_i$ that can consistently reproduce the parity-dependent phases observed in Fiona's experiment across all loops $l$, because each party is sometimes included in the loop and sometimes excluded from it.
Of course, this argument only rules out the most direct qubit simulation strategies. Quentin might still hope that a completely different protocol, in which each party can use an arbitrary (possibly infinite) number of qubits and perform arbitrary local operations, could succeed. However, in the Supplemental Material~\ref{appendix:contradiction} we prove that \emph{every} possible qubit-based strategy must essentially be equivalent to a variation of this naive attempt, and fail. Our proof relies on GHZ-state self-testing~\cite{Supic2020Review, baccari2020scalable}, and self-testing of network token-counting strategies introduced in~\cite{sekatski2023partial}. 
These together allow us to fully characterize all qubit realizations reproducing the observed correlations for each fixed loop $l$. 
We then show that no consistent qubit realization exists simultaneously for all loops without explicitly depending on $l$.

Bell's theorem can be understood as showing that correlations generated without communication by entangled quantum particles require communication to be reproduced classically, thereby justifying entanglement as a fundamental physical phenomenon and information theoretical resource. Here, we extend this notion of nonlocality to communication scenarios. We prove that correlations generated using a single round of local fermionic communication through a network cannot be reproduced by bosonic or distinguishable-particle systems without additional communication. More precisely, any such simulation where systems are initially independently prepared requires at least two communication rounds. In this sense, our result justifies the need for fermionic anticommutation.

Importantly, our experiment is realistic, as all its elementary ingredients have already been demonstrated experimentally, for instance with fermionic cold atoms, where tunnel coupling between spatial modes can be used to implement fermionic beam splitters and generate entanglement between fermionic modes~\cite{Kaufman2014,Spar2022}. Internal degrees of freedom such as spin may simply be frozen throughout the protocol. Our result nevertheless faces technical limitations. In particular, our proof currently lacks robustness to statistical and experimental noise, which is why we do not provide any Bell-like inequality satisfied by qubit, but violated by febits.

Our work also clarifies why such a separation had not previously been revealed. Existing fermion-to-qubit mappings, such as Bravyi--Kitaev superfast encodings, implicitly assume a centralized setting in which the communication network is globally known in advance and a global state (representing the fermionic vacuum) is shared among all the parties. 
By contrast, our work considers a genuinely distributed setting in which distant parties operate using only local information and a limited number of rounds of communication through the network. In general this is not enough to create the potentially entangled bosonic state encoding of the fermionic vacuum. For completeness, we show in the Supplemental Material~\ref{appendix:BK-simulation} that if the bosonic parties were instead able to initially pre-share an arbitrary global entangled state, then they would be able to locally reproduce the correlations of our fermionic Gedankenexperiment.

Our work leads to the conjecture (see Supplemental Material~\ref{appendix:conjecture}) that the separation established here is only the first manifestation of a much stronger phenomenon, in which reproducing a single round of local fermionic communication would require a number of bosonic communication rounds growing with the size of the network. It also suggests that paradigmatic distributed problems---such as coloring communication loops using only local exchanges---may fundamentally require fewer communication rounds with fermionic information carriers than with standard qubits.

More broadly, our work suggests that particle statistics themselves may define a hierarchy of information-theoretic resources. While fermions can always locally simulate bosons and distinguishable particles through composite even-parity excitations, it is natural to ask whether analogous separations may exist with more exotic particles such as anyons or paraparticles~\cite{Wang2025ParticleExchange}, which may be more powerful than fermions.

\section*{Acknowledgments}
We thank
Antonio Ac\'in, Pedro Barrios Hita, Nadia Belabas, Dagmar Bru\ss, Pierre Fraigniaud, Nicolas Gisin, David Gross, Timoth\'ee Hoffreumon, Hermann Kampermann, Matthias Kleinmann, Fran\c{c}ois Le Gall, Stefano Pironio, Jukka Suomela, Anton Trushechkin, and Mirjam Weilenmann for fruitful discussions. F.M.K., X.X., L.T. and M.-O. R. acknowledge funding by the ANR for the JCJC grants LINKS (No. ANR-23-CE47-0003), the T-ERC QNET (No. ANR24-ERCS-0008), the project QUANTINT, as well as the European Union's Horizon 2020 Research and Innovation Programme under QuantERA Grant Agreements No. 731473 and No. 101017733. 
T.G. acknowledges funding from the European Research Council (project DebuQC, grant agreement No. 101098279) and from  German Federal Ministry of
Research (project FermiQP).
This work was funded by the European Union under the Marie Sklodowska-Curie Actions (MSCA) through the QNETS project (grant agreement ID: 101208259). Views and opinions expressed are, however, those of the author(s) only and do not necessarily reflect those of the European Union or the European Education and Culture Executive Agency (EACEA). Neither the European Union nor EACEA can be held responsible for them.

\clearpage
\printbibliography

\clearpage
\appendix

\crefalias{section}{appendix}
\crefname{appendix}{Supplementary Material}{Supplementary Materials}
\Crefname{appendix}{Supplementary Material}{Supplementary Materials}

\setcounter{figure}{0}
\renewcommand{\thefigure}{S\arabic{figure}}

\begin{center}
    {\LARGE\bfseries Supplementary Materials\par}
\end{center}

\crefalias{section}{appendix}
\crefname{appendix}{Supplementary Material}{Supplementary Materials}
\Crefname{appendix}{Supplementary Material}{Supplementary Materials}

\tableofcontents

\section{Quantum Information Theory for Qubits, Bosons, and Fermions} 
\label{appendix:ParticleFormalism}

Information can be encoded into different types of physical degrees of freedom. In the most common implementations of quantum registers, information is encoded in internal degrees of freedom of distinct quantum systems. These systems can be bosonic or fermionic particles, where the relevant internal degrees of freedom may be, for example, spin or electronic states of atoms or ions, or polarization states of photons. They can also be macroscopic systems, such as electrical current states in superconducting circuits. In these cases, the systems themselves are distinguishable, for instance because they are placed in separate spatial modes. Alternatively, information may be encoded in the occupation of certain spatial modes by indistinguishable particles, where all internal degrees of freedom are assumed to be the same. These indistinguishable particles can be either bosons or fermions. In what follows, we describe in more detail the theoretical description of each of these three cases.

In the next sections, we show that when information is encoded in internal degrees of freedom, or when the presence or absence of bosonic particles in spatial modes is used as the encoding, the resulting information theory is essentially the same as the conventional framework of quantum information theory. We refer to this framework as \emph{qubit quantum information theory} (Q-QIT), in order to distinguish it from the less conventional theory that emerges when indistinguishable fermions are considered and information is encoded in their presence and absence (i.e, occupations). We refer to the latter as \emph{fermionic quantum information theory} (F-QIT).

An important notion in these theories is how systems are composed. We denote the composition rule abstractly by \(\boxtimes\) for a general theory. This notion is closely related to the concept of locally prepared states, (or sometimes called independently prepared systems). A global state \(\psi_{\mathrm{LP}}\) is locally prepared with respect to a partition \(A|B\) if and only if there exist local physical states \(\psi_{A}\) and \(\psi_{B}\) such that \cite{erba2024composition}
\begin{equation}
  \psi_{\mathrm{LP}} = \psi_{A} \boxtimes \psi_{B}.
\end{equation}
The operation \(\boxtimes\) corresponds to the standard tensor product in Q-QIT, while in F-QIT it corresponds to the antisymmetric wedge product, as explained in the next subsections.

\subsection{Qubit Quantum Information Theory (Q-QIT)}
In this subsection, we review the two standard realizations of Q-QIT. The first uses internal degrees of freedom of distinguishable quantum systems, while the second uses occupation numbers of identical bosons in distinguishable modes. Although these encodings are physically different, they lead to the same underlying information-theoretic structure.

\subsubsection{Q-QIT encoded in internal degrees of freedom of particles}\label{app:Q-QITinInternalDegOfFreedom}

Consider a collection of \(M\) particles, either all bosons, all fermions, or of mixed particle types, contained in \emph{distinct} modes that can be independently addressed. These could be, for instance, atoms or ions trapped at different positions in space, photons of different frequencies, or an array of superconducting circuits. Information can then be encoded in their internal degrees of freedom. Suppose that particle \(i\), for \(1\leq i\leq M\), has \(d_i\) accessible states, corresponding to the Hilbert space \(\mathcal{H}_{i}=\mathbb{C}^{d_i}\), or to an infinite-dimensional separable Hilbert space in the case \(d_i=\infty\). The composite system of the \(M\) particles is then described by the tensor-product Hilbert space
\begin{equation}
    \mathcal{H}=\bigotimes_{i=1}^M \mathcal{H}_{i}. \label{eq:tensor-product-distinguishable}
\end{equation}
This is precisely the setting of standard quantum information theory. Without loss of generality, we may assume that each \(d_i\) is a power of \(2\), so that each subsystem \(i\) can be regarded as a collection of qubits. For this reason, we also refer to this framework as \emph{qubit quantum information theory} (Q-QIT).

In this setting, it is natural to introduce a notion of locality based on the ability of each local party to address only a subset of particles, for instance because those particles are held in a specific region of space. More precisely, suppose that the \(M\) particles are divided into \(N\) disjoint subsets \(S_1,\dots,S_N\), and that the parties \(A_1,\dots,A_N\) each have access to one subset. Then a local operator for party \(A_k\) is an operator that can be expressed as the tensor product of an operator acting on the degrees of freedom accessible to \(A_k\) and the identity on the rest:
\begin{equation}
    O_{A_k}= O_{S_k}\otimes \id_{S_k^{\,c}},
\end{equation}
where \(O_{S_k}\) acts on the subset \(S_k\) of the tensor factors in~\eqref{eq:tensor-product-distinguishable}, and \(\id_{S_k^{\,c}}\) acts as the identity on the remaining tensor factors. Similarly, a locally preparable state is a tensor product across this partition of the system:
\begin{equation}
    \ket{\psi_{\mathrm{LP}}}
    =
    \ket{\psi_1}_{A_1}\otimes\cdots\otimes\ket{\psi_N}_{A_N},
\end{equation}
where \(\ket{\psi_k}_{A_k}\in\bigotimes_{i\in S_k}\mathcal{H}_{i}\).

\subsubsection{Q-QIT encoded in identical bosons}

Consider a collection of \(M\) distinct and independently addressable modes, which may be occupied by indistinguishable bosonic particles. These could be, for instance, distinct electromagnetic field modes in a cavity system, or bound states in the minima of an optical lattice potential or an atomic tweezer array occupied by bosonic atoms. Information can then be encoded in the occupation numbers of these modes. The Hilbert space describing this system is the bosonic Fock space~\cite{RevModPhys.84.621}.

In second quantization\footnote{see Section~\ref{appendix:FirstQuantization} for a discussion of the first-quantized picture, which is less convenient for the description of local operations and locally prepared states when information is encoded in occupations of the modes.}, the bosonic Fock space can be defined in terms of the creation and annihilation operators \(b_i^\dagger\), \(b_i\) associated with each mode \(i=1,\dots,M\), together with a vacuum state \(\ket{\Omega}\). These operators satisfy the canonical commutation relations (CCR)
\begin{equation}
    [b_i,b_j]=0, \qquad  [b^\dagger_i,b^\dagger_j]=0, \qquad  [b_i,b^\dagger_j]=\delta_{ij} \id,
\end{equation}
and the vacuum satisfies \(b_i\ket{\Omega}=0\) for all \(i=1,\dots,M\). More specifically, the bosonic Fock space \(\mathcal{F}_{\rm b}\) is defined as the linear span of the Fock states
\begin{equation}
    \ket{n_1,n_2,\dots,n_M}
    =
    \frac{1}{\sqrt{n_1!\dots n_M!}}\:
    (b_1^\dagger)^{n_1}\cdots (b_M^\dagger)^{n_M}\ket{\Omega},
    \label{eq:bosonic-fock-states}
\end{equation}
for all \((n_1,n_2,\dots,n_M)\in\mathbb{N}^M\). Linear operators on this Hilbert space can be regarded as functions of the creation and annihilation operators.

The notion of locality for bosonic systems is defined similarly to the qubit case~\cite{Eisert2003}. Adopting the same notation as above, suppose that the \(M\) modes are divided into \(N\) disjoint subsets \(S_1,\dots,S_N\), and that the parties \(A_1,\dots,A_N\) each have access to one subset. A local operator for party \(A_k\) is then an operator that can be expressed as a function only of the mode operators \(b_i^\dagger,b_i\) with \(i\in S_k\). That is,
\begin{equation}
    O_{A_k}=p(\{b_i^\dagger,b_i, \forall i \in S_k\})
\end{equation}
for some suitable function \(p(\cdot)\). Similarly, a locally preparable state is a state that can be created from the vacuum by a product of local operators on each subset \(S_k\). That is,
\begin{equation}
    \ket{\psi_{\mathrm{LP}}} = O_{A_1}\cdots O_{A_N} \ket{\Omega}.
\end{equation}

Notice that the space \(\mathcal{F}_{\rm b}\) is isomorphic to the tensor product
\begin{equation}
    \mathcal{H}=\bigotimes_{i=1}^M\mathcal{H}_i,
    \label{eq:tensor-product-bosons}
\end{equation}
where each local space \(\mathcal{H}_i\) is an infinite-dimensional Hilbert space spanned by the set of states \(\{\ket{n}: n \in \mathbb{N}\}\). The isomorphism is given simply by identifying the Fock basis~\eqref{eq:bosonic-fock-states} with a basis of~\eqref{eq:tensor-product-bosons} as
\begin{equation}
    \ket{n_1,n_2,\dots,n_M}
    \longmapsto
    \ket{n_1}\otimes \cdots \otimes \ket{n_M}.
    \label{eq:JW-isomorphism-bosons}
\end{equation}
This isomorphism respects the notion of locality: a local bosonic operator \(O_{A_k}\) is mapped to an operator of the form \(O_{S_k}\otimes \id_{S_k^{\,c}}\), where \(O_{S_k}\) acts on the subset \(S_k\) of the tensor factors of~\eqref{eq:tensor-product-bosons}. Similarly, a locally preparable bosonic state is mapped to
\begin{equation}
    \ket{\psi_1}_{A_1}\otimes\cdots\otimes\ket{\psi_N}_{A_N},
\end{equation}
where \(\ket{\psi_k}_{A_k}\in\bigotimes_{i\in S_k}\mathcal{H}_i\). This makes the bosonic Fock space completely analogous to the Hilbert space~\eqref{eq:tensor-product-distinguishable} of distinguishable systems with local dimensions \(d_i=\infty\).

Note that we use the same notation, \(\ket{\cdot}\), for states in the two cases, since they lead to the same information-theoretic description. Throughout the rest of this work, we may therefore use the terms distinguishable particles, bosons, and qubits interchangeably whenever only this information-theoretic structure is relevant.

\subsection{Fermionic Quantum Information Theory (F-QIT) encoded in identical fermions}

Consider a collection of \(M\) distinct and independently addressable modes, which may be occupied by a number of indistinguishable fermionic particles. These could be, for instance, bound states in the minima of an optical lattice potential or of a tweezer array occupied by fermionic atoms. We can encode information in their occupation numbers. The Hilbert space describing this system is the fermionic Fock space~\cite{bravyi_fermionic_2002, vidal2021quantum}.

The fermionic Fock space can be defined in terms of the creation and annihilation operators \(f_i^\dagger\), \(f_i\) associated with each mode \(i=1,\dots,M\), which satisfy the canonical anticommutation relations (CAR)
\begin{equation}
    \{f_i,f_j\}=0, \qquad  \{f^\dagger_i,f^\dagger_j\}=0, \qquad  \{f_i,f^\dagger_j\}=\delta_{ij} \id,
\end{equation}
and a vacuum state \(\fket{\Omega}\) such that \(f_i\fket{\Omega}=0\) for all \(i=1,\dots,M\). Throughout the Supplemental Materials, we use rounded kets, \(\fket{\cdot}\), for fermionic states in order to distinguish them from qubit-based quantum states, denoted by \(\ket{\cdot}\).
More specifically, the fermionic Fock space \(\mathcal{F}_{\rm f}\) is defined as the linear span of the Fock states
\begin{equation}
    \fket{n_1,n_2,\dots,n_M}
    =
    (f_1^\dagger)^{n_1}\cdots (f_M^\dagger)^{n_M}\fket{\Omega},
    \label{eq:fermionic-fock-states}
\end{equation}
where the CAR imply that \(n_i\) can only take values in \(\{0,1\}\), since \(f_i^2=0\). This is a finite-dimensional space, and the space of all linear operators on it coincides with the space of polynomials in the creation and annihilation operators.

In this setting, the notion of locality is based on the ability of a local party to address only a subset of modes, which are for instance localized in a specific region of space~\cite{bravyi_fermionic_2002}. More precisely, suppose that the \(M\) modes are divided into \(N\) disjoint subsets \(S_1,\dots,S_N\), and that the parties \(A_1,\dots,A_N\) each have access to one subset. Then a local operator for party \(A_k\) is an operator that can be expressed as a function only of the mode operators \(f_i^\dagger\), \(f_i\) with \(i\in S_k\), that is
\begin{equation}
    O_{A_k}=p(\{f_i^\dagger,f_i, \forall i \in S_k\}),
    \label{eq:local-op-fermions}
\end{equation}
for some polynomial \(p\).

Contrary to the bosonic case, not every such local operator corresponds to a physically allowed operation. For consistency with the no-signaling principle, physical operators and observables must obey the parity superselection rule.

For a general physical local operator \(O_{A_k}\), the parity superselection rule implies that \(O_{A_k}\) must have a definite local parity. That is, it may either preserve the local parity or change it, but it cannot produce superpositions of different parities. Let
\begin{equation}
    P_k=\exp\left(i\pi\sum_{i\in S_k} f_i^\dagger f_i\right)
\end{equation}
be the local parity operator. If \(O_{A_k}\) commutes with \(P_k\), then it is parity preserving and has the block-diagonal form
\begin{equation}
     O_{A_k}
    =
    \begin{pmatrix}
         O_e & 0 \\
         0 & O_o
    \end{pmatrix}.
\end{equation}
If \(O_{A_k}\) anticommutes with \(P_k\), then it is parity changing and has the block-off-diagonal form
\begin{equation}
     O_{A_k}
    =
    \begin{pmatrix}
         0 & O_{oe} \\
         O_{eo} & 0
    \end{pmatrix}.
    \label{eq:even-odd-block-operator}
\end{equation}
Here the block structure refers to an ordering of the basis~\eqref{eq:fermionic-fock-states} in which states with an even number of fermions in subsystem \(S_k\) come before states with an odd number of fermions. Equivalently, \(O_{A_k}\) can be an even or an odd polynomial in the local creation and annihilation operators, but not a superposition of the two.

For observables, the parity superselection rule imposes the stronger condition that physical local observables must preserve the local fermion-number parity~\cite{johansson2016comment}. We denote such observables by \(X_{A_k}\). Equivalently, \(X_{A_k}\) must be an even polynomial in the local creation and annihilation operators. This means that physical observables cannot connect sectors of the Hilbert space with even and odd local fermion number. Thus, in a basis ordered according to local parity, any admissible local observable has the block-diagonal form
\begin{align}
     X_{A_k}
    =
    X_e \oplus X_o
    =
    \begin{pmatrix}
         X_e & 0 \\
         0 & X_o
    \end{pmatrix}.
    \label{eq:even-odd-block-observable}
\end{align}
The same parity-preserving condition is imposed on all physically allowed unitary evolutions.

Similar to the previous section, a locally preparable state is a state that can be created from the vacuum by a product of local operators on each subset \(S_k\). That is,
\begin{equation}
    \fket{\psi_{\mathrm{LP}}} = O_{A_1}\cdots O_{A_N} \fket{\Omega},
    \label{eq:local-state-fermions}
\end{equation}
where each \(O_{A_k}\) is composed only of the mode operators \(f_i^\dagger\), \(f_i\) with \(i\in S_k\). Consistently with the parity superselection rule, locally preparable states cannot be superpositions of even and odd local fermion-number sectors. Equivalently, each \(O_{A_k}\) in~\eqref{eq:local-state-fermions} must have a definite parity: it may be even or odd, but it cannot be a superposition of even and odd terms. This ensures that the resulting local state is an eigenstate of the local parity operator \(P_k\).

The state~\eqref{eq:local-state-fermions} can be seen as a product of local states on each subsystem. To make this product structure explicit, it is useful to introduce a notation for the fermionic state product, which we denote by the wedge product \(\wedge\). Let
\begin{equation}
    \fket{\psi_{A_1}}= O_{A_1}\fket{\Omega}, \quad \dots, \quad
    \fket{\psi_{A_N}}= O_{A_N}\fket{\Omega}
\end{equation}
be the states on these local subsystems. Then we define their fermionic product as
\begin{equation}
    \fket{\psi_{A_1}}\wedge \dots \wedge  \fket{\psi_{A_N}}
    :=
    O_{A_1}\cdots O_{A_N}\fket{\Omega}.
    \label{eq:wedge-product-local-states}
\end{equation}
With this notation, the condition that a state is locally preparable can be written equivalently as
\begin{equation}
    \fket{\psi_{\mathrm{LP}}}
    =
    \fket{\psi_{A_1}}\wedge\cdots\wedge\fket{\psi_{A_N}}.
\end{equation}

The wedge-product notation can also be extended to operators. For two parties, let \(\fket{\psi_{A_1}},\fket{\psi'_{A_1}}\) be states of subsystem \(A_1\), and let \(\fket{\psi_{A_2}},\fket{\psi'_{A_2}}\) be states of subsystem \(A_2\). We define
\begin{equation}
   \fketbra{\psi_{A_1}}{\psi'_{A_1}}
   \wedge
   \fketbra{\psi_{A_2}}{\psi'_{A_2}}
   :=
   \big[
   \fket{\psi_{A_1}} \wedge \fket{\psi_{A_2}}
   \big]
   \left[
   \fket{\psi'_{A_1}} \wedge \fket{\psi'_{A_2}}
   \right]^\dagger .
\end{equation}
By linear extension, this defines the wedge product of operators. In this notation, the locality condition for operators can be written equivalently as
\begin{equation}
    O_{A_k}= O_{S_k} \wedge \id_{S_k^{\,c}},
\end{equation}
where \(O_{S_k}\) acts only on the modes in \(S_k\), and \(\id_{S_k^{\,c}}\) denotes the identity on the remaining modes. Thus, the wedge product provides an equivalent way of expressing the locality conditions introduced above: locally preparable states are wedge products of local states, and local operators are wedge products of an operator on the accessible modes with the identity on the complement.
Notice that, unlike the tensor product, the wedge product has the property that
\begin{align}
(O_A\wedge O_B)= (O_A\wedge \id_B)(\id_A\wedge O_B)= (-1)^{p(A)p(B)}(\id_A\wedge O_B)(O_A\wedge\id_B), 
\end{align}
where $p(X)$ is the parity of $X$ defined as $p(X)=0$ if $X$ is even and $p(X)=1$ if $X$ is odd. The same is true for vectors:
\begin{align}
\fket{\psi_A}\wedge \fket{\psi_B}= (\fket{\psi_A}\wedge \id_B)(1\wedge \fket{\psi_B})= (-1)^{p(\psi_A)p(\psi_B)}(\id_A\wedge \fket{\psi_B})(\fket{\psi_A}\wedge \id_B).
\end{align}

The fermionic Fock space \(\mathcal{F}_{\rm f}\) is naturally isomorphic to a qubit Hilbert space of the form~\eqref{eq:tensor-product-distinguishable} for \(d_i=2\). The isomorphism takes the same form as~\eqref{eq:JW-isomorphism-bosons} and is known as the Jordan--Wigner transformation~\cite{jordan_uber_1928}. However, this isomorphism does not preserve the locality of operators. More precisely, a local fermionic operator \(O_{A_k}\) will not always be mapped to an operator of the form \(O_{S_k}\otimes\id_{S_k^{\,c}}\) with respect to the above tensor-product structure. Alternative mappings have also been proposed which attempt to embed the fermionic Fock space in a qubit Hilbert space in such a way that operator locality is preserved; see, e.g.,~\cite{bravyi_fermionic_2002, ball_fermions_2005, verstraete_mapping_2005, setia_superfast_2019}. Nonetheless, these embeddings inevitably encounter the problem that locally preparable fermionic states~\eqref{eq:local-state-fermions} do not necessarily map to states of the form
\[
\ket{\psi_1}_{A_1}\otimes\cdots\otimes\ket{\psi_N}_{A_N}.
\]
This is the reason why, in a distributed setting where parties are constrained to act locally, a fermionic system cannot be immediately reproduced by a qubit one~\cite{Guaita2025localityofqubit}.

\section{A fermionic Gedankenexperiment that cannot be locally simulated by bosons or distinguishable particles}
\label{appendix:ExperimentAndMainResult}
Consider a distributed experiment involving three parties, $A$, $B$, and $C$, together with a referee who can redistribute subsystems between them. The experiment proceeds in three steps.

In step 1, each party locally prepares a bipartite state labelled by $\alpha,\beta,\gamma$ for parties $A$, $B$, and $C$, respectively, and sends one subsystem (half of the state) to the referee. Specifically, each party prepares a bipartite system $\theta_L\theta_R$ for $\theta \in \{\alpha,\beta,\gamma\}$ and retains subsystem $\theta_R$ and sends $\theta_L$ to the referee.

In step 2, the referee redistributes the subsystems among the parties. This redistribution is performed as follows: 
The referee first selects a non-trivial permutation of the three labels $A,B,C$, permutes the received subsystems accordingly, and then redistributes them to the parties. This means that, in the case of a cyclic permutation, they form a directed three-loop. In the other cases, two parties form a directed two-loop while one party is given back its half of the state $\theta_L$ by the referee.
The referee can be viewed as a relay whose role is to redistribute the subsystems.

The referee can arrange the three parties in two possible three-loops, which we represent by $ABC$ and $ACB$, corresponding to the two directions of the redistribution. For the two-party subsets, the redistribution is unique, yielding the configurations $AB$, $AC$, and $BC$. These five possible choices, which are also called wirings, are labelled by the index
\begin{equation}
    l \in \{ABC, ACB, AB, AC, BC\}.
\end{equation}
Crucially, the parties are not informed about the applied redistribution; that is, the global loop variable $l$ remains hidden. We refer to this operational constraint as \emph{locality}.

Throughout this manuscript, we employ two complementary labelling conventions for subsystems to simplify calculations. The first is defined from the perspective of the prepared states, with subsystems labelled by $\theta_R$ and $\theta_L$ for $\theta \in \{\alpha,\beta,\gamma\}$. The second is defined from the perspective of the parties after the redistribution occurred, with subsystems denoted by $i_R$ and $i_L$ for $i \in \{A,B,C\}$. Translating between these two conventions requires knowledge of the underlying wiring $l$, and becomes clear only after the redistribution of subsystems among the parties has been carried out.

As an illustrative example, consider the configuration $ABC$. According to the redistribution associated to this case described above, the correspondence between the labellings of the subsystems is given by
\begin{equation}
\begin{aligned}
\label{eq:labellingMap}
&A_L = \alpha_R, \quad B_R = \alpha_L, \\
&B_L = \beta_R, \quad C_R = \beta_L, \\
&C_L = \gamma_R, \quad A_R = \gamma_L,
\end{aligned}
\end{equation}
as depicted in \cref{fig:labels}.

\begin{figure}[h]
    \centering
    \includegraphics[width=0.6\linewidth]{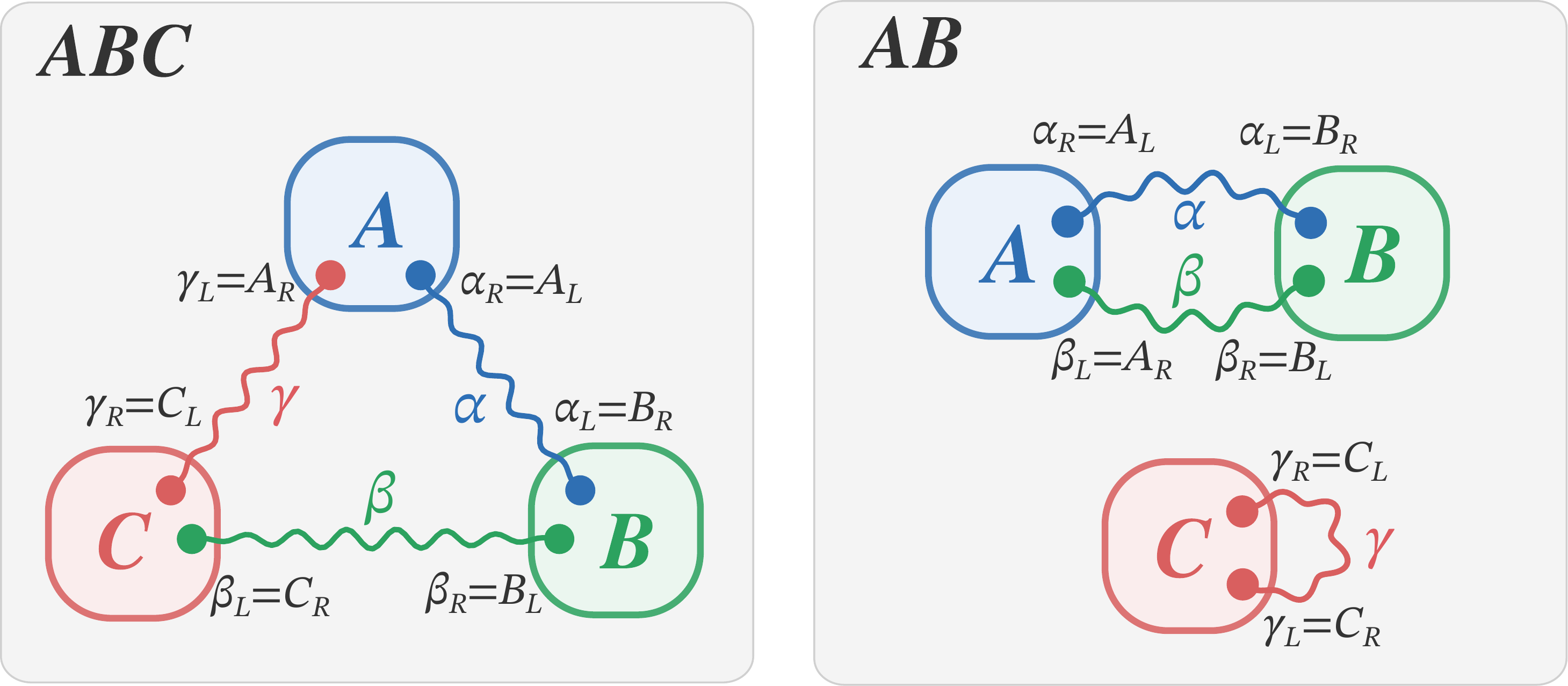}
    \caption{Party-based labelling depends on the underlying wiring; shown here for the loop $ABC$ and $AB$.}
    \label{fig:labels}
\end{figure}

Note that, for any loop configuration, the global state can be written as a locally prepared state,
\begin{equation}
  \psi_{\mathrm{global}}= \psi_{\alpha} \boxtimes \psi_{\beta} \boxtimes \psi_{\gamma},
\end{equation}
where $\boxtimes$ denotes a theory-agnostic composition rule. In Qubit Quantum Information Theory, this corresponds to the standard tensor product, while in the fermionic setting it corresponds to the wedge product (see Supplemental Material~\ref{appendix:ParticleFormalism}).

Steps~1 and~2 can be interpreted jointly as a blind communication process: each party effectively shares a state with a neighbouring party, although neither the identity of that party nor the topology of the network is known to the parties.

In Step 3, the referee sends each party \(i\) an input bit \(x_i\in\{0,1\}\). 

In Step 4, each party $i$ produces an output $a_i \in \{0,1_-, 1_+,2\}$. Importantly, each party generates these outputs by performing a local operation on the subsystems it holds at this stage. We note that each party only knows its own value of $x_i$ and not those of the others. We furthermore note that during the whole experiment, the referee acts independently of the rest of the experiment and also unbiased, i.e., the actions of the referee do not favor Fiona's or Quentin's strategy.

In what follows, we describe a specific fermionic strategy for the game introduced above, alongside a qubit strategy. We focus on the different observations that arise from these two settings to provide an intuition to the proof of our main result.

In Step 1, each party $i \in \{A,B,C\}$ locally prepares a bipartite maximally entangled state. In the fermionic setting, this state is prepared from a single fermion in a superposition between a mode being kept and one being sent to the referee. This state is given by

\begin{align}
\label{eq:FermionicTrustedState}
    \fket{\psi^+_\theta}_{\theta_R \theta_L}
    = \frac{1}{\sqrt{2}} \bigl( \fket{01}_{\theta_R \theta_L} + \fket{10}_{\theta_R \theta_L} \bigr)
    := \frac{1}{\sqrt{2}} \bigl( f_{\theta_L}^\dagger + f_{\theta_R}^\dagger \bigr)\vacuumket_{\theta_R \theta_L},
\end{align}
where $\vacuumket_{\theta_R \theta_L}$ denotes the vacuum over modes $\theta_R$ and $\theta_L$. In the qubit setting, we allow for a relative phase $\varphi_\theta$ between the two components and consider the state
\begin{align}\label{eq:qubit-state-psi-theta}
    \ket{\psi^+_\theta}_{\theta_R \theta_L}
    = \frac{1}{\sqrt{2}} \bigl( \ket{01}_{\theta_R \theta_L} + e^{\ii \varphi_\theta}\ket{10}_{\theta_R \theta_L} \bigr).
\end{align}

In Step 2 and 3, the referee redistributes the subsystems as described above and sends the parties the input bits.

In Step 4, the parties are required to produce outputs. To produce the output $a_i$, each party performs a measurement with the following PVM elements.
\begin{equation}
\begin{aligned}
\label{eq:FermionicProjectors}
    &\Pi_i^0=\fket{\Omega}_{i_R i_L},\\
    &\Pi_i^2=f^\dagger_{i_R} f^\dagger_{i_L}  \fket{\Omega}_{i_R i_L},\\
    &\Pi_{i,x_i}^{1_+}=\left( u^{+}_{i,x_i} f^\dagger_{i_L} + v^{+}_{i,x_i} f^\dagger_{i_R} \right)\vacuumket_{i_R, i_L}, \\
    &\Pi_{i,x_i}^{1_-}= \left( u^{-}_{i,x_i} f^\dagger_{i_L} + v^{-}_{i,x_i} f^\dagger_{i_R} \right)\vacuumket_{i_R, i_L}.
\end{aligned}
\end{equation}

In the qubit setting, the corresponding PVM elements are
\begin{equation}
\begin{aligned}
\label{eq:QubitProjectors}
    &\Pi_i^0=\ket{00}_{i_R i_L},\\
    &\Pi_i^2=\ket{11}_{i_R i_L},\\
    &\Pi_{i, x_i}^{1_+}=u^{+}_{i,x_i} \ket{01}_{i_R i_L} + v^{+}_{i,x_i} \ket{10}_{i_R i_L},\\
    &\Pi_{i, x_i}^{1_-}= u^{-}_{i,x_i} \ket{01}_{i_R i_L} + v^{-}_{i,x_i} \ket{10}_{i_R i_L}.
\end{aligned}
\end{equation}
where the coefficients $u^{\pm}_{i,x_i}, v^{\pm}_{i,x_i} \in \mathbb{C}$ satisfy
\begin{equation}
\label{eq:CoefficientsOrthonormality}
\begin{aligned}
    |u^{\pm}_{i,x_i}|^2 + |v^{\pm}_{i,x_i}|^2 = 1, \quad
    u^{-}_{i,x_i} = \bigl(v^{+}_{i,x_i}\bigr)^*, \quad
    v^{-}_{i,x_i} = -\bigl(u^{+}_{i,x_i}\bigr)^*.
\end{aligned}
\end{equation}

We note that the measurement observable in the fermionic setting are parity-preserving and thus are compatible with the parity superselection rule.

We denote by $P_f^l$, for $l \in \{ABC, ACB, AB, AC, BC\}$, the probability distributions arising from this experiment in the fermionic setting, and by $P_b^l$ the corresponding distributions in the bosonic (qubit) setting.

For each loop $l$, let $V_l$ denote the subset of parties involved in the non-trivial loop, and let $E_l$ denote the corresponding subset of states included in the loop. We introduce the shorthand
\[
\bar{a}_{V_l} := (a_i)_{i \in V_l}, \quad
\bar{x}_{V_l} := (x_i)_{i \in V_l}.
\]

\begin{proposition}\label{prop:FermionExperimentCorrelationValue}
After coarse graining the outcomes \(1_\pm\) into the outcome \(1\), the probability distributions of the experiment described above are independent of the inputs. For \(l \in \{ABC,ACB\}\), they are given by
\begin{equation}
\begin{aligned}
&P_{f/b}^{l}(a_A=2,a_B=1,a_C=0)= \frac{1}{8},
\qquad
P_{f/b}^{l}(a_A=2,a_B=0,a_C=1)= \frac{1}{8},\\
&P_{f/b}^{l}(a_A=1,a_B=2,a_C=0)= \frac{1}{8},
\qquad
P_{f/b}^{l}(a_A=0,a_B=2,a_C=1)= \frac{1}{8},\\
&P_{f/b}^{l}(a_A=0,a_B=1,a_C=2)= \frac{1}{8},
\qquad
P_{f/b}^{l}(a_A=1,a_B=0,a_C=2)= \frac{1}{8},\\
&P_{f/b}^{l}(a_A=1,a_B=1,a_C=1)= \frac{1}{4},\\
&P_{f/b}^{l}(a_A,a_B,a_C)=0,
\qquad \text{if } a_A+a_B+a_C \neq 3 .
\end{aligned}
\end{equation}
Moreover, the conditional distributions satisfy
\begin{equation} \label{eq:Dist_3}
\begin{aligned}
P_f^{l}\!\left(
\bar{a}_{V_l}
\mid 
\bar a_{V_l} \in \{1_\pm \}^{|V_l|},
\bar{x}_{V_l}
\right)
&=
\frac{1}{2}
\left|
\prod_{i \in V_l } u_{i,x_i}^{s_i}
+
\prod_{i \in V_l} v_{i,x_i}^{s_i}
\right|^2,\\
P_b^{l}\!\left(
\bar{a}_{V_l}
\mid 
\bar a_{V_l} \in \{1_\pm \}^{|V_l|},
\bar{x}_{V_l}
\right)
&=
\frac{1}{2}
\left|
\prod_{i \in V_l } u_{i,x_i}^{s_i}
+
e^{\ii \sum_{\theta \in E_l } \varphi_\theta}
\prod_{i \in V_l} v_{i,x_i}^{s_i}
\right|^2 .
\end{aligned}
\end{equation}
For \(l\in\{AB,AC,BC\}\), the corresponding coarse-grained distributions are given by
\begin{equation}
\begin{aligned}
&P_{f/b}^{l}(a_i=1,a_j=1)=\frac{1}{2},\\
&P_{f/b}^{l}(a_i=2,a_j=0)=\frac{1}{4},
\qquad
P_{f/b}^{l}(a_i=0,a_j=2)=\frac{1}{4},\\
&P_{f/b}^{l}(a_i,a_j)=0,
\qquad \text{if } a_i+a_j\neq 2,
\end{aligned}
\end{equation}
where \(V_l=\{i,j\}\).

And the corresponding conditional distributions are
\begin{equation}\label{eq:Dist_2}
\begin{aligned}
P_f^{l}\!\left(
\bar{a}_{V_l}
\mid 
\bar a_{V_l} \in \{1_\pm \}^{|V_l|},
\bar{x}_{V_l}
\right)
&=
\frac{1}{2}
\left|
\prod_{i\in V_l} u_{i,x_i}^{s_i}
-
\prod_{i \in V_l} v_{i,x_i}^{s_i}
\right|^2, \\
P_b^{l}\!\left(
\bar{a}_{V_l}
\mid 
\bar a_{V_l} \in \{1_\pm \}^{|V_l|},
\bar{x}_{V_l}
\right)
&=
\frac{1}{2}
\left|
\prod_{i\in V_l} u_{i,x_i}^{s_i}
+
e^{\ii \sum_{\theta \in E_l } \varphi_\theta}
\prod_{i \in V_l} v_{i,x_i}^{s_i}
\right|^2 .
\end{aligned}
\end{equation}
\end{proposition}

\begin{proof}
    This follows by a direct calculation using the formalism introduced in Supplemental Material~\ref{appendix:ParticleFormalism}; see Supplemental Material~\ref{appendix:dist} for details.
\end{proof}
Proposition~\ref{prop:FermionExperimentCorrelationValue} extends to an arbitrary number of parties $N$.
In the $N$-party generalization of the strategy described above, the resulting
distributions generalize according to the loop parity: for odd loops, they
take the form in~\cref{eq:Dist_3}, while for even loops, they take the form
in~\cref{eq:Dist_2}.

Note that for each fixed loop $l$, the corresponding fermionic correlations can be reproduced by an appropriate qubit strategy within the same loop configuration. If the loop size is odd, the strategy is the same as in the fermionic game; if it is even, a specific party can compensate for the phase difference in~\cref{eq:Dist_2} by manually incorporating it into the state they prepare.

Indeed, in the trusted qubit strategy above, for a fixed loop configuration $l$
(take, for instance, $l=ABC$), the global state is
\begin{equation}
\frac{1}{2\sqrt{2}}
\bigl( \ket{01}_{A_L B_R} + e^{\ii \varphi_\alpha}\ket{10}_{A_L B_R} \bigr)
\otimes
\bigl( \ket{01}_{B_L C_R} + e^{\ii \varphi_\beta}\ket{10}_{B_L C_R} \bigr)
\otimes
\bigl( \ket{01}_{C_L A_R} + e^{\ii \varphi_\gamma}\ket{10}_{C_L A_R} \bigr).
\end{equation}
Conditioned on the event that every party receives exactly one token,
\[
    a_A, a_B, a_C \in \{1_\pm \},
\]
this state reduces to
\begin{equation}
\frac{1}{2\sqrt{2}}
\bigl(
\ket{01}_{A_L B_R}\ket{01}_{B_L C_R}\ket{01}_{C_L A_R}
+
e^{\ii(\varphi_\alpha+\varphi_\beta+\varphi_\gamma)}
\ket{10}_{A_L B_R}\ket{10}_{B_L C_R}\ket{10}_{C_L A_R}
\bigr),
\end{equation}
which is a GHZ state~\cite{greenberger1989going}. More generally, for any loop configuration $l \in \{ABC, ACB, AB, AC, BC\}$, the conditional state shared by the parties belonging to the non-trivial loop is a GHZ state whose relative phase is given by
\[
    \phi_l := \sum_{\theta \in E_l} \varphi_\theta .
\]
By choosing suitable measurement coefficients $u^{\pm}_{i,x_i}$ and
$v^{\pm}_{i,x_i}$, the resulting correlations reproduce the GHZ self-testing
scenario of Ref.~\cite{baccari2020scalable}. In particular, combining the
partial self-testing result for the token-counting game~\cite{sekatski2023partial} with the self-testing of the GHZ experiment~\cite{baccari2020scalable} gives a self-testing statement for the fixed loop: any strategy reproducing the qubit distribution associated with loop $l$ must be equivalent, up to local unitaries, to the corresponding trusted strategy. This is formalized in the following theorem for a fixed loop of $N$ parties.

\begin{theorem}[Full qubit quantum self-testing for a fixed loop of size $N \geq 3$]\footnote{or standard quantum self-testing, with no assumption on the dimension.}
\label{thm:fullselftesting}
For the trusted strategy, consider the qubit game described above generalized to
$N$ parties, where $N \geq 3$. Fix a loop configuration $l$. Label the parties sequentially as
$1,\dots,N$ according to their order along the loop, and denote by $\ket{\psi_i}$
the state prepared by party $i$; see~\cref{fig:QubitSelftesting}. In the trusted
scenario described above, $\ket{\psi_i}$ is given by~\cref{eq:qubit-state-psi-theta}:
\begin{equation}
\ket{\psi_i}_{i_L,(i+1)_R}
=
\frac{1}{\sqrt{2}}
\left(
\ket{01}_{i_L,(i+1)_R}
+
e^{\ii\varphi_i}
\ket{10}_{i_L,(i+1)_R}
\right),
\qquad i\in\{1, \dots, N\},
\end{equation}
where indices are understood modulo \(N\)
Assume that the phases $\varphi_i$ satisfy
\begin{equation}
\label{eq:phase_constraint}
    \sum_{i=1}^{N} \varphi_i = \phi_l,
    \qquad \phi_l \in \{0,\pi\}.
\end{equation}
Moreover, let the coefficients of the trusted measurements in~\cref{eq:QubitProjectors} be given as follows. For party $i=1$,
    \begin{equation}
        \begin{aligned}
            & u_{i,0}^{+}=-v_{i, 0}^{-}=\frac{1}{\sqrt{2(2-\sqrt{2}})} &\qquad &v_{i, 0}^{+}=u_{i,0}^{-}=\frac{\sqrt{2}-1}{\sqrt{2(2-\sqrt{2}})}  \\
            & u_{i,1}^{+}=-v_{i, 1}^{-}= \frac{\sqrt{2}-1}{\sqrt{2(2-\sqrt{2}})}   &\qquad 
            &v_{i, 1}^{+}=u_{i,1}^{-}=\frac{1}{\sqrt{2(2-\sqrt{2}})},
             \label{eq:meas_coeff_forA}
        \end{aligned}
    \end{equation}
and for $i \in \{ 2, ..., N \}$,
    \begin{equation}
        \begin{aligned}
            & u_{i,0}^{+}= - v_{i,0}^{-}=1 &\qquad &v_{i,0}^{+}=u_{i,0}^{-}=0  \\
            & u_{i,1}^{+}= - v_{i,1}^{-}=\frac{1}{\sqrt{2}}  &\qquad &v_{i,1}^{+}=u_{i,1}^{-}=\frac{1}{\sqrt{2}}.
        \label{eq:meas_coeff_forother}
        \end{aligned}
    \end{equation}
Then, the resulting output distribution, for this fixed loop $l$, must be generated by the above strategy, up to local unitary transformations. More precisely, any strategy with the same output distribution as above, takes the following form:
\begin{itemize}
\item The right and left components $i_R, i_L$ of the Hilbert space of each party $i$ decompose into system and junk parts, with the system part being two-dimensional. More precisely, the left component consists of subsystems $S_{i_L}, J_{i_L}$ and the right component consists of subsystems $S_{i_R}, J_{i_R}$. Moreover, $S_{i_R}$ and $S_{i_L}$ are two-dimensional Hilbert spaces, each
isomorphic to $\mathbb{C}^2$.

\item For any $i \in \{1, \cdots, N \}$, there exist local unitary operators $U_{i_R}$ and $U_{i_L}$, acting on $S_{i_R} \otimes J_{i_R}$ and $S_{i_L} \otimes J_{i_L}$, respectively, such that
\begin{equation}
\label{eq:FullSelf-testedState}
U_{i_L} U_{(i+1)_R} \ket{\hat{\psi}_i}_{i_L (i+1)_R} = \frac{1}{\sqrt{2}} \left( \ket{01}_{S_{i_L} S_{(i+1)_L}} + e^{\ii \hat{\varphi}_i}\ket{10}_{S_{i_L} S_{(i+1)_L}}\right) \otimes \ket{\zeta_i}_{J_{i_L} J_{(i+1)_L}} 
\end{equation}
where $\sum_{i=1}^{N} \hat{\varphi}_i = \phi_l$, and $\ket{\zeta_i}_{J_{i_L} J_{(i+1)_L}}$ is a normalized state on the Hilbert space $J_{i_L} \otimes J_{(i+1)_L}$.
\item 
Moreover, define $U_i:= U_{i_R} \otimes U_{i_L}$ and let $\hat{M}_{x_i}^{(i)}$ for $x_i\in \{0,1\}$ be the conditional observable of party $i$, i.e.,
\begin{equation}
    \hat M_{x_i}^{(i)}= \hat \Pi_{i, x_i}^{1_+} - \hat \Pi_{i, x_i}^{1_-},
\end{equation}
and define
\begin{equation}
\begin{aligned}
\mathcal{X}_{1} &= \frac{\hat{M}^{(1)}_0 + \hat{M}^{(1)}_1}{\sqrt{2}} &\qquad \mathcal{Z}_{1} &= \frac{\hat{M}^{(1)}_0 - \hat{M}^{(1)}_1}{\sqrt{2}} \\
\mathcal{X}_{i} &= \hat{M}^{(i)}_1  &\qquad \mathcal{Z}_{i} &= \hat{M}^{(i)}_0 \qquad \forall i \in \{2,..., N \}.
\end{aligned}
\end{equation}
Then, the measurement is self-tested as
\begin{equation}
\begin{aligned}
\label{eq:FullSelf-testedObservable}
&\hat \Pi_{i }^{0} = \ketbra{00}{00}_{S_{i_R} S_{i_L}} \otimes \id_{J_i} \\
&\hat \Pi_{i }^{2} = \ketbra{11}{11}_{S_{i_R} S_{i_L}} \otimes \id_{J_i}  \\
& U_i  \mathcal{Z}_i U_i^\dagger = Z_{S_i} \otimes \id_{J_i} \\
&U_i  \mathcal{X}_i U_i^\dagger = X_{S_i} \otimes \id_{J_i},
\end{aligned}
\end{equation}
where
\begin{equation}
\begin{aligned}
&Z_{S_i}= \left(\ketbra{01}{01}- \ketbra{10}{10} \right)_{S_{i_R} S_{i_L}} \\
&X_{S_i}= \left(\ketbra{01}{10}+ \ketbra{10}{01} \right)_{S_{i_R} S_{i_L}}.
\end{aligned}
\end{equation}
\end{itemize}
\end{theorem}

\begin{figure}[h]
    \centering
    \includegraphics[width=0.4\linewidth]{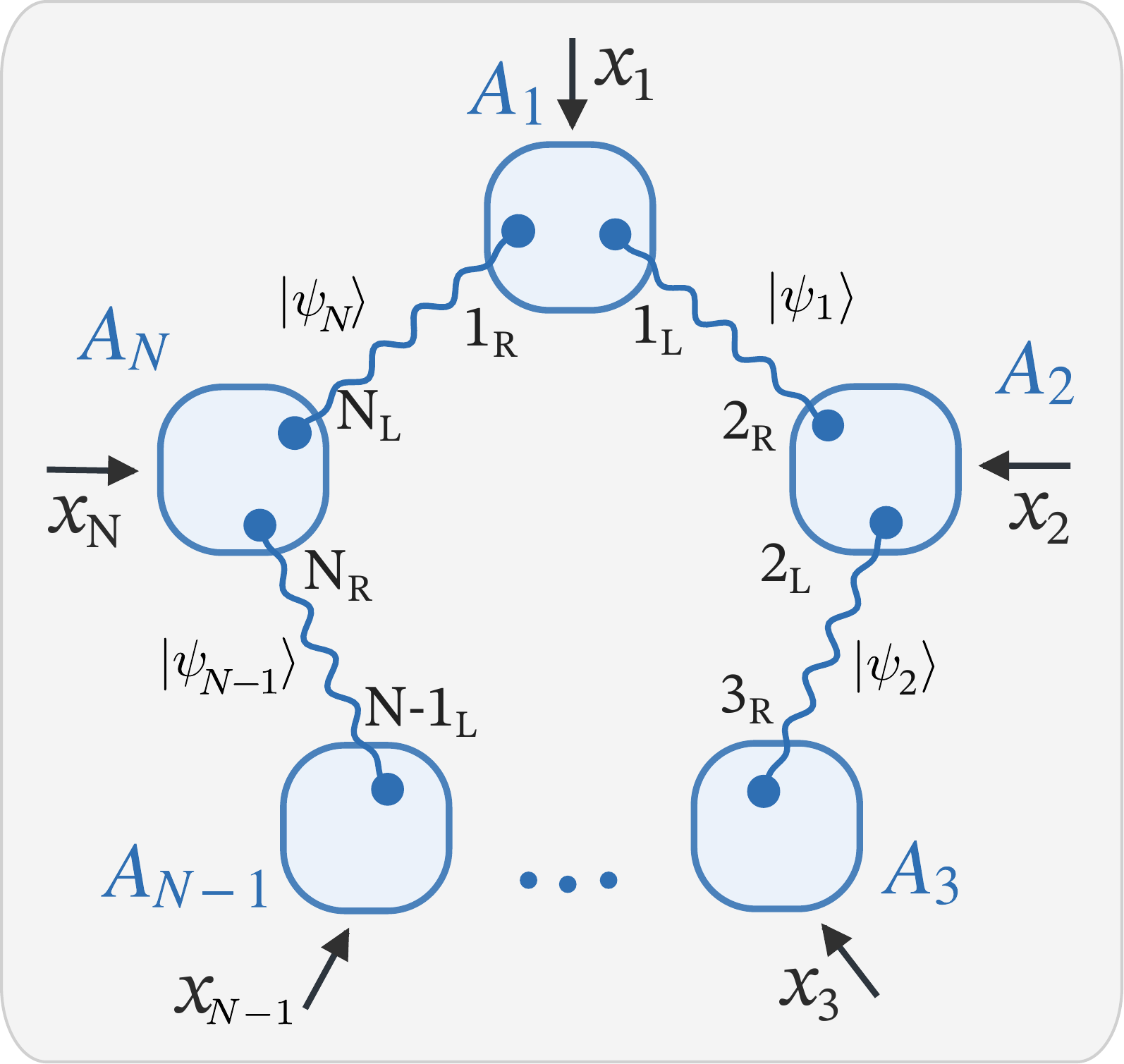}
    \caption{\(N\) parties connected by bipartite states in a loop, together with the corresponding Hilbert-space labels.}
    \label{fig:QubitSelftesting}
\end{figure}

\begin{proof}
It is given in Supplemental Material~\ref{appendix:FullSelfTesting}.
\end{proof}

\textit{Remark.}
There is a weaker version of this theorem for \(N=2\): the states and measurements can be self-tested as above, but under an additional assumption on the structure of the strategy. This result is stated in \cref{thm:fullselftesting_N2} of \cref{appendix:fullselftesting_N2}.

We now show that, within the rewiring experiment described above, it is impossible to simultaneously reproduce the distributions $P_f^l$, for all loop configurations $l$, using a local bosonic or distinguishable-particle strategy. By a local strategy, we mean that both the state preparation and the measurements are performed locally by each party and are crucially independent of the loop configuration implemented by the referee.

Indeed, based on the self-testing result stated in \cref{thm:fullselftesting}, we derive nontrivial constraints on any qubit strategy that simultaneously reproduces the fermionic distributions $P_f^l$ for all loop configurations $l$. Then, we argue that these constraints are incompatible with the requirement that the strategies corresponding to different loops arise from a single underlying local model, which leads to a contradiction. This is the main result of this work and is formalised in the following theorem.

\begin{theorem}[No local qubits simulation of the fermionic game]
\label{thm:MainTheoremNoSimultaneuousSimulation}
The family of probability distributions $\{P_f^l\}_{l \in \{ABC, ACB, AB, AC, BC\}}$ arising from the fermionic rewiring experiment described above (see \cref{fig:rewiring-game}), with the measurement coefficients specified in \cref{eq:meas_coeff_forA,eq:meas_coeff_forother}, cannot be simultaneously reproduced by any local qubits strategy within the same scenario.
\end{theorem}

\begin{proof}
To provide intuition for the origin of the indistinguishable-fermion advantage, consider the situation in which each local party prepares a maximally entangled single-fermion state, as illustrated in \cref{fig:SuperPositionState}(a). Conditioning on the events in which every party ends up with one particle gives a coherent superposition of two terms: one in which all particles are kept locally, and one in which all particles are transmitted. Crucially, the second term involves \(N-1\) fermionic exchanges. As a result, a GHZ state is created whose relative phase depends on the parity of the loop.

This provides a way of encoding the loop parity into the global state, in a manner that is not accessible with standard qubits under the same communication structure. In the next step, the measurements reveal this difference at the level of the observed distributions. Indeed, the states \(\mathrm{GHZ}^+\) and \(\mathrm{GHZ}^-\) can be perfectly distinguished by suitable local measurements.

Importantly, reproducing a parity-dependent state with local qubit systems requires \(N/2\) rounds of communication. Our result shows that, with the same single communication step as in the fermionic protocol, standard qubits cannot reproduce the fermionic statistics. The formal proof is provided in~\cref{appendix:contradiction}.

The conjecture that any protocol with fewer than \(N/2\) communication rounds is insufficient is discussed in~\cref{appendix:conjecture}.
\end{proof}

\begin{figure}[H]
    \centering
    \includegraphics[width=0.99\linewidth]{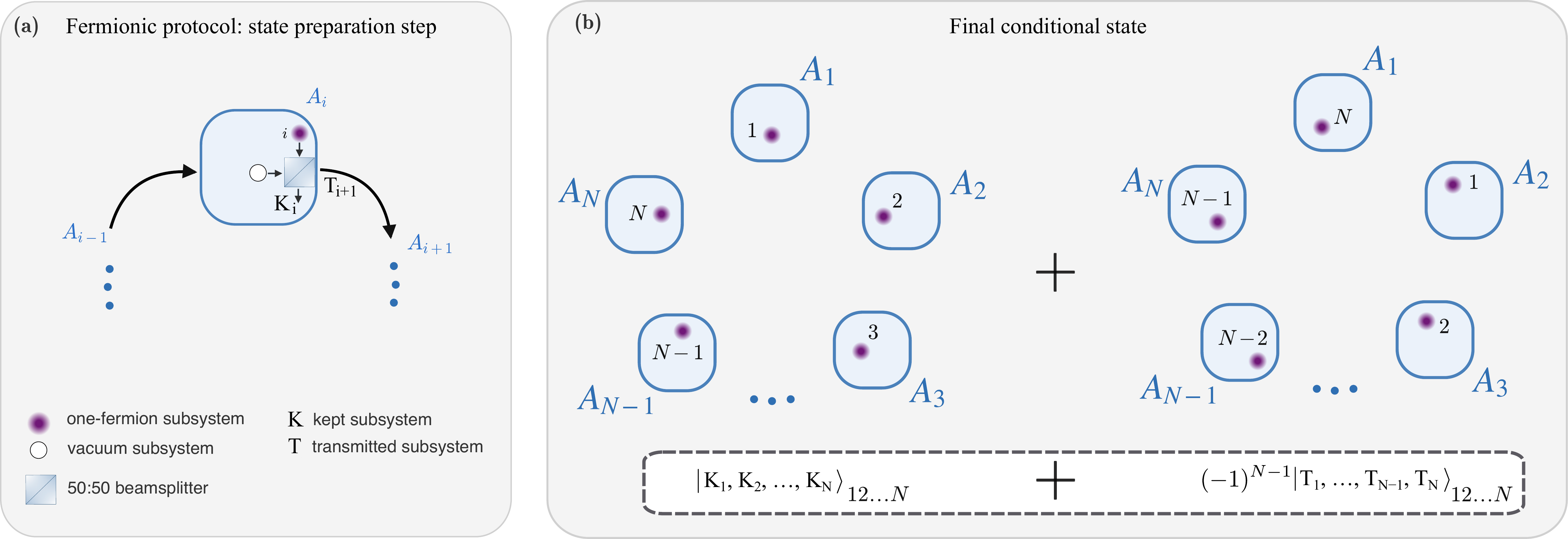}
    \caption{ \textbf{Fermionic exchange phase encoded in a global state in a detectable way.} \textbf{(a)} state preparation in the fermionic protocol.
    \textbf{(b)} The post-selected state is a GHZ state with a relative sign determined by the parity of the loop.} 
    \label{fig:SuperPositionState}
\end{figure}

\section{About the  concrete experimental realisation of the fermionic Gedankenexperiment}
\label{app:experimental-implementation}
For an experimental realisation of the fermionic Gedankenexperiment (which we prove cannot be locally reproduced by bosonic implementations), one needs three parties to be able to locally address two fermionic modes each. During the experiment, some of these modes need to be rearranged, such that each party has access to a different subset of modes. Given the current state-of-the-art, the  most promising platform on which such a scenario could be implemented are trapped fermionic atoms in reconfigurable optical tweezers~\cite{gonzalez-cuadra_fermionic_2023}. Each tweezer can be modeled as a fermionic mode that may or  may not be occupied by indistinguishable atoms. Recent technological advances have shown that it is possible to rearrange these tweezers in space without destroying the quantum state they contain~\cite{bluvstein_quantum_2022}. 

Once this system is available, the whole experiment can be implemented by free fermionic operations (or fermionic linear optics)~\cite{PhysRevA.65.032325} and occupation number measurements. In what follows we explain more in detail how this can be done. An implementation of the entire distributed protocol can be divided into the following subsequent steps:
\begin{enumerate}
    \item An initial Fock basis state is prepared, where each party holds the state $\fket{01}_{i_R i_L}$, \textit{i.e.} an occupied and an unoccupied mode;
    \item Each party acts with a local unitary operation $U^{(1)}_i$ on their two modes creating locally an entangled state;
    \item The modes are rearranged according to the referee's prescriptions;
    \item The parties apply another local unitary $U^{(2)}_{i,x_i}$, which depends on the input bit $x_i$, on the two modes they now hold;
    \item Each party measures their modes in the occupation basis.
\end{enumerate}
The last two steps in particular, as we will see below, are a way to implement the measurement protocol described in the main presentation of the fermionic experiment. We now discuss each step more in detail.

Step 1. can be achieved simply by loading fermionic atoms into the experiment and measuring locally in the occupation basis to confirm that the required Fock basis state has been prepared.

In Step 2. each party needs to apply a local fermionic unitary transformation that maps the state $\fket{01}_{i_Ri_L}=f_{i_L}^\dagger \vacuumket_{i_R i_L}$ into the required initial state
\begin{equation}
    \fket{\psi^+}_{i_R i_L}
    = \frac{1}{\sqrt{2}} \bigl( \fket{01}_{i_R i_L} + \fket{10}_{i_R i_L} \bigr)
    = \frac{1}{\sqrt{2}} \bigl( f_{i_L}^\dagger + f_{i_R}^\dagger \bigr)\vacuumket_{i_R i_L}.
\end{equation}
This can be achieved by the following local fermionic unitary, corresponding to a 50:50 fermionic beam splitter
\begin{equation}
    U^{(1)}_i=\exp\left[\frac{\pi}{4}(f_{i_R}^\dagger f_{i_L} -f_{i_L}^\dagger f_{i_R} )\right].
\end{equation}
This indeed has the action $U^{(1)}_i f_{i_L}^\dagger {U^{(1)}_i}^\dagger = \bigl( f_{i_L}^\dagger + f_{i_R}^\dagger \bigr)/\sqrt{2}$.

Step 3. is implemented as discussed above by moving the optical tweezers into a new spatial configuration. 

In Step 4. each party applies, depending on their input bit $x_i$, the local free fermionic unitary, corresponding to a tunable beam splitter, 
\begin{equation}
    U^{(2)}_{i,x_i}=\exp\left[\alpha_{i,x_i}(f_{i_R}^\dagger f_{i_L} -f_{i_L}^\dagger f_{i_R} )\right].
\end{equation}
The angles $\alpha_{i,x_i}$ are chosen in the following way: in the case of party A, 
\begin{equation}
    \alpha_{A,0}=\frac{3\pi}{8}, \quad \quad \alpha_{A,1}=\frac{\pi}{8},
\end{equation}
in the case of the other parties, $i\in\{B,C\}$, 
\begin{equation}
    \alpha_{i,0}=\frac{\pi}{2}, \quad \quad \alpha_{i,1}=\frac{\pi}{4}.
\end{equation}

In Step 5., the parties measure in the occupation basis. This is a measurement with four possible outcomes, corresponding to projecting on the states $\fket{00}_{i_R i_L}$, $\fket{01}_{i_R i_L}$, $\fket{10}_{i_R i_L}$ and $\fket{11}_{i_R i_L}$. Thus, the combined effect of Steps 4. and 5. can be seen as a four outcome projective measurement on the rotated basis given by the states ${U^{(2)}_{i,x_i}}^\dagger\fket{00}_{i_R i_L}$, ${U^{(2)}_{i,x_i}}^\dagger \fket{01}_{i_R i_L}$, ${U^{(2)}_{i,x_i}}^\dagger \fket{10}_{i_R i_L}$ and ${U^{(2)}_{i,x_i}}^\dagger \fket{11}_{i_R i_L}$. An expression for these states can be easily derived by applying the relation
\begin{equation}
    {U^{(2)}_{i,x_i}}^\dagger \begin{pmatrix} f_{i_R}^\dagger \\ f_{i_L}^\dagger \end{pmatrix} U^{(2)}_{i,x_i} = \begin{pmatrix} \cos\alpha_{i,x_i} & \sin\alpha_{i,x_i} \\ -\sin\alpha_{i,x_i} &\cos\alpha_{i,x_i}  \end{pmatrix}\begin{pmatrix} f_{i_R}^\dagger \\ f_{i_L}^\dagger \end{pmatrix}.
\end{equation}
This in particular leads to ${U^{(2)}_{i,x_i}}^\dagger\fket{00}_{i_R i_L}=\fket{00}_{i_R i_L}$, corresponding to observing zero fermions in the local modes, and ${U^{(2)}_{i,x_i}}^\dagger\fket{11}_{i_R i_L}=\fket{11}_{i_R i_L}$, corresponding to observing two fermions in the local modes. The states ${U^{(2)}_{i,x_i}}^\dagger \fket{01}_{i_R i_L}$, ${U^{(2)}_{i,x_i}}^\dagger \fket{10}_{i_R i_L}$ on the other hand can be seen to coincide with the measurement bases of the one fermion sector given in Eq.~\eqref{eq:FermionicProjectors}. This means that one can fully recover the measurement protocol discussed in the main description of the fermionic experiment by attributing the following labels to the outcomes of the occupation number measurement of Step 5.:
\begin{align}
    &00 \longmapsto a=0 \nonumber\\
    &01 \longmapsto a=1_+ \nonumber\\
    &10 \longmapsto a=1_- \nonumber\\
    &11 \longmapsto a=2. \nonumber
\end{align}

\section{Conjecture: stronger fermionic advantages}
\label{appendix:conjecture}

Our main result demonstrates a fundamental advantage of fermionic information carriers over bosonic or distinguishable ones in the distributed rewiring scenario of Fig.~\ref{fig:rewiring-game}. 
More precisely, perfect correlations (without experimental or statistical noise) generated using one round of local fermionic communication require at least two rounds of communication---for instance communication between parties at distance $L=2$ on the loop network---to be reproduced using bosonic or distinguishable-particle systems.
We believe that this separation is only the first manifestation of a broader hierarchy of fermionic nonlocality.

First, our result only holds in an idealized Gedankenexperiment in which the observed correlations are exactly the one predicted in F-FIT with an ideal experimental setup. This would require both a noisyless experimental setup, and infinite experimental statistics to avoid any statistical noise. 
This assumption is technically required by our proof techniques, but there is currently no conceptual reason why noisy correlations close to the one of our Gedankenexperiment would suddenly become feasible with a bosonic simulation.

Second, Fig.~\ref{fig:cyclic-fermionic-phase} suggests a much stronger communication separation.
While fermions generate the phase $(-1)^{N-1}$ using only nearest-neighbor communication on the loop, reproducing the same global parity information with qubits appears to require communication over distances scaling as $N/2$.
We therefore conjecture that no qubit-based strategy restricted to communication distances strictly smaller than $N/2$ can reproduce all correlations obtainable using one round of fermionic communication.

These considerations motivate the following conjecture.

\begin{conjecture}
Consider the distributed rewiring game of Fig.~\ref{fig:rewiring-game}, generalized to $N$ parties $A_1,\ldots,A_N$. Then there exists a fermionic strategy whose correlations cannot be reproduced, or even approximated, by any qubit-based strategy:
\begin{enumerate}
    \item using systems of arbitrary, possibly infinite, local dimension, including simulations in the commuting-operator model;
    \item restricted to communication distances strictly smaller than $N/2$;
\end{enumerate}
\end{conjecture}

\FloatBarrier
\section{Fermionic and bosonic distributions}
\label{appendix:dist}

In this Supplemental Material, we prove \cref{prop:FermionExperimentCorrelationValue} by explicitly deriving the corresponding distributions. We proceed in three steps: first analyzing the three-party loop in the fermionic case, then the two-party loop, and finally comparing with the bosonic (or distinguishable) case.

\subsection*{Fermionic particle case: three-party loops}

In the fermionic setting, after step 1 of the strategy described in Supplemental Material~\ref{appendix:ExperimentAndMainResult}, the state is given by
\begin{equation}
    \fket{\psi}= \fket{\psi^+_{\alpha}} \wedge \fket{\psi^+_{\beta}} \wedge \fket{\psi^+_{\gamma}}.
\end{equation}

In the loop $ABC$, using the labeling according to \cref{fig:labels}, the state can be written as
\begin{equation}
    \fket{\psi}= \frac{1}{2\sqrt{2}}(f^{\dagger}_{A_L}+ f^{\dagger}_{B_R}) (f^{\dagger}_{B_L}+ f^{\dagger}_{C_R}) (f^{\dagger}_{C_L}+ f^{\dagger}_{A_R}) \fket{\Omega},
\end{equation}
where $\fket{\Omega}$ is the global vacuum. 

Fixing the ordering $A_R A_L B_R B_L C_R C_L$, we rewrite
\begin{equation}\label{eq:statethree}
\begin{aligned}
\fket{\psi}
&=
\frac{1}{2\sqrt{2}}(f^{\dagger}_{A_L}+ f^{\dagger}_{B_R})
(f^{\dagger}_{B_L}+ f^{\dagger}_{C_R})
f^{\dagger}_{C_L} \fket{\Omega}
\\
&\quad +
\frac{1}{2\sqrt{2}}f^{\dagger}_{A_R}
(f^{\dagger}_{A_L}+ f^{\dagger}_{B_R})
(f^{\dagger}_{B_L}+ f^{\dagger}_{C_R})
\fket{\Omega}.
\end{aligned}
\end{equation}

\paragraph{Token-counting distribution.}

We first compute the particle-number (coarse grained) statistics. We denote the coarse-grained outcomes corresponding to $1_\pm$ by $1$. Note that the resulting distribution is independent of the inputs. As an example,
\begin{equation}
\begin{aligned}
P_f^{ABC}(a_A=2,a_B= 1,a_C=0)
&=
\fbra{\psi} \Pi_A^2 \wedge \Pi_B^1 \wedge \Pi_C^0 \fket{\psi}
\\
&=
\fbra{\psi}\Big[
f_{A_R}^\dagger f_{A_L}^\dagger
\fketbra{\Omega}{\Omega}_{A_R A_L}
f_{A_L} f_{A_R}
\\
&\qquad \wedge
\Big(
f^\dagger_{B_R} \fketbra{\Omega}{\Omega}_{B_R B_L} f_{B_R}
+
f^\dagger_{B_L} \fketbra{\Omega}{\Omega}_{B_R B_L} f_{B_L}
\Big)
\\
&\qquad \qquad \wedge
\fketbra{\Omega}{\Omega}_{C_R C_L} \Big]
\fket{\psi}.
\end{aligned}
\end{equation}

This reduces to
\begin{equation}
\begin{aligned}
P_f^{ABC}(2,1,0)
&=
\abs{\fbra{\Omega}f_{B_R} f_{A_L} f_{A_R} \fket{\psi}}^2
+
\abs{\fbra{\Omega}f_{B_L} f_{A_L} f_{A_R} \fket{\psi}}^2.
\end{aligned}
\end{equation}

From the structure of $\fket{\psi}$, the first term vanishes, while
\begin{equation}
\begin{aligned}
\abs{\fbra{\Omega}f_{B_L} f_{A_L} f_{A_R} \fket{\psi}}^2= \frac{1}{8}.
\end{aligned}
\end{equation}

Hence,
\begin{equation}
    P_f^{ABC}(a_A=2,a_B=1,a_C=0)= \frac{1}{8}.
\end{equation}

By symmetry,
\begin{equation}
\begin{aligned}
    &P_f^{ABC}(a_A=2,a_B=0,a_C=1)= \frac{1}{8}, 
    &P_f^{ABC}(a_A=1,a_B=2,a_C=0)= \frac{1}{8},\\
    &P_f^{ABC}(a_A=0,a_B=2,a_C=1)= \frac{1}{8},
    &P_f^{ABC}(a_A=0,a_B=1,a_C=2)= \frac{1}{8},\\
    &P_f^{ABC}(a_A=1,a_B=0,a_C=2)= \frac{1}{8}.
\end{aligned}
\end{equation}

Next, we compute
\begin{equation}
\begin{aligned}
P_f^{ABC}(a_A=1,a_B=1,a_C=1)
&=
\fbra{\psi} \Pi_A^1 \wedge \Pi_B^1 \wedge \Pi_C^1 \fket{\psi}.
\end{aligned}
\end{equation}

Expanding the projectors,
\begin{equation}
\begin{aligned}
&=
\fbra{\psi} \Big[ 
\Big(
f^\dagger_{A_R} \fketbra{\Omega}{\Omega} f_{A_R}
+
f^\dagger_{A_L} \fketbra{\Omega}{\Omega} f_{A_L}
\Big)
\\
&\qquad \quad \qquad \qquad \wedge
\Big(
f^\dagger_{B_R} \fketbra{\Omega}{\Omega} f_{B_R}
+
f^\dagger_{B_L} \fketbra{\Omega}{\Omega} f_{B_L}
\Big)
\\
&\qquad \qquad \qquad \qquad \qquad \wedge
\Big(
f^\dagger_{C_R} \fketbra{\Omega}{\Omega} f_{C_R}
+
f^\dagger_{C_L} \fketbra{\Omega}{\Omega} f_{C_L}
\Big) \Big]
\fket{\psi}.
\end{aligned}
\end{equation}

All mixed terms vanish. For instance,
\begin{equation}
\begin{aligned}
 \fbra{\Omega}  f_{C_L} f_{B_R} f_{A_R} \fket{\psi}  =0.
\end{aligned}
\end{equation}

Thus,
\begin{equation}
\begin{aligned}
P_f^{ABC}(1,1,1)
&=
\lvert \fbra{\Omega}  f_{C_R} f_{B_R} f_{A_R}  \fket{\psi} \rvert^2
+
\lvert \fbra{\Omega}  f_{C_L} f_{B_L} f_{A_L} \fket{\psi} \rvert^2
\
&= \frac{1}{4}.
\end{aligned}
\end{equation}
From the above calculations it is clear that the $P(a_A, a_B, a_C)=0$ when $a_A+a_B+a_C\neq 3$. This distribution is a token counting distribution according to \cite{PhysRevA.105.022408}, arising from a simple token counting model with one uniform token per source on a triangle network.

\paragraph{Conditional distribution.}

We now compute
\begin{equation}
\begin{aligned}
P_f^{ABC}\!\left(a_A,a_B,a_C \mid (a_A,a_B,a_C)\in\{1_\pm\}^3,x_A,x_B,x_C\right)
=
\frac{
P_f^{ABC}\! \left(a_A,a_B,a_C \, \cap \, (a_A,a_B,a_C)\in\{1_\pm\}^3 \mid x_A,x_B,x_C \right)
}{
P_f^{ABC}\!\left((a_A,a_B,a_C)\in\{1_\pm\}^3 \mid x_A,x_B,x_C\right)
}.
\end{aligned}
\end{equation}

The numerator $\forall a_i \in \{1_\pm \}$, is 
\begin{equation}
\begin{aligned}
&P_f^{ABC}(1_{s_A}, 1_{s_B}, 1_{s_C}  \mid x_A,x_B,x_C) \\
&=\left|
\fsandwich{\Omega}{
(u_{C,x_C}^{s_C*}f_{C_R}+v_{C,x_C}^{s_C*}f_{C_L})
(u_{B,x_B}^{s_B*}f_{B_R}+v_{B,x_B}^{s_B*}f_{B_L})
(u_{A,x_A}^{s_A*}f_{A_R}+v_{A,x_A}^{s_A*}f_{A_L})}{\psi}
\right|^2.
\end{aligned}
\label{eq:numerator}
\end{equation}
Expanding all terms, we obtain
\begin{equation}
    \begin{aligned}
    P_f^{ABC}&( 1_{s_A}, 1_{s_B}, 1_{s_C} \mid x_A,x_B,x_C) \\
    &=\left|\frac1{2\sqrt2}\left\{\fsandwich{\Omega}{(u_{C,x_C}^{s_C*}f_{C_R}+v_{C,x_C}^{s_C*}f_{C_L})
    (u_{B,x_B}^{s_B*}f_{B_R}+v_{B,x_B}^{s_B*}f_{B_L})
    (u_{A,x_A}^{s_A*}f_{A_R}+v_{A,x_A}^{s_A*}f_{A_L})f^\dag_{A_R}f^\dag_{B_R}f^\dag_{C_R}}{\Omega}\right.\right.\\
    &\left.\left.+\fsandwich{\Omega}{(u_{C,x_C}^{s_C*}f_{C_R}+v_{C,x_C}^{s_C*}f_{C_L})
    (u_{B,x_B}^{s_B*}f_{B_R}+v_{B,x_B}^{s_B*}f_{B_L})
    (u_{A,x_A}^{s_A*}f_{A_R}+v_{A,x_A}^{s_A*}f_{A_L})f^\dag_{A_L}f^\dag_{B_L}f^\dag_{C_L}}{\Omega}\right\}\right|^2
    \end{aligned}
\end{equation}

Now, only two terms survive and we obtain
\begin{equation}
\begin{aligned}
P_f^{ABC}(1_{s_A}, 1_{s_B}, 1_{s_C} &\mid x_A,x_B,x_C)\\
&= \left|
u_{A,x_A}^{a_A*}u_{B,x_B}^{a_B*}u_{C,x_C}^{a_C*}
\fbra{\Omega} f_{C_R} f_{B_R} f_{A_R}\fket{\psi}
+
v_{A,x_A}^{a_A*}v_{B,x_B}^{a_B*}v_{C,x_C}^{a_C*}
\fbra{\Omega} f_{C_L} f_{B_L} f_{A_L}\fket{\psi}
\right|^2\\
   &= \frac{1}{8}
   \left|
   \prod_i u_{i,x_i}^{s_i}
   +
   \prod_i v_{i,x_i}^{s_i}
   \right|^2.
\end{aligned}
\end{equation}
After normalization, this yields \cref{eq:Dist_3}.


\subsection*{Fermionic particle case: two-party loops}
We now consider a two-party loop, for instance $AB$, in which case the party $C$ is excluded from the loop. The global state is given by
\begin{equation}
    \fket{\psi}= \fket{\psi^+_{\alpha}} \wedge \fket{\psi^+_{\beta}} \wedge \fket{\psi^+_{\gamma}},
\end{equation}
which can be written in terms of the party labels (see \cref{fig:labels}) as
\begin{equation}
    \fket{\psi}= \frac{1}{2\sqrt{2}}(f^{\dagger}_{A_L}+ f^{\dagger}_{B_R}) (f^{\dagger}_{B_L}+ f^{\dagger}_{A_R}) (f^{\dagger}_{C_L}+ f^{\dagger}_{C_R}) \fket{\Omega},
\end{equation}
where $\fket{\Omega}$ denotes the global vacuum.

Fixing the ordering $A_R A_L B_R B_L C_R C_L$, the state simplifies to
\begin{equation}
\label{eq:statetwo}
\begin{aligned}
    \fket{\psi}=
    \frac{1}{2\sqrt{2}}(f^{\dagger}_{A_L}+ f^{\dagger}_{B_R}) f^{\dagger}_{B_L} (f^{\dagger}_{C_L}+ f^{\dagger}_{C_R})\fket{\Omega}
    -
    \frac{1}{2\sqrt{2}}f^{\dagger}_{A_R} (f^{\dagger}_{A_L}+ f^{\dagger}_{B_R})(f^{\dagger}_{C_L}+ f^{\dagger}_{C_R}) \fket{\Omega}.
\end{aligned}   
\end{equation}
We note that, in contrast to the three-party case in \cref{eq:statethree}, the state in \cref{eq:statetwo} contains a relative minus sign, which will play a crucial role in the following derivation of the conditional distribution.
\paragraph{Token-counting distribution.}

Proceeding as in the three-party case, one obtains
\begin{equation}
\begin{aligned}
    &P_f^{AB}(a_A=2, a_B=0, a_C=1)= \frac{1}{4},\\
    &P_f^{AB}(a_A=0, a_B=2, a_C=1)= \frac{1}{4},\\
    &P_f^{AB}(a_A=1, a_B=1, a_C=1)= \frac{1}{2}.
\end{aligned}
\end{equation}
similarly, here $P_f^{AB}(a_A, a_B)=0$ when $a_A + a_B \neq 2$.

\paragraph{Conditional distribution.}

We now compute the conditional probabilities. Starting from the definition, $\forall a_i \in \{ 1_\pm\}$
\begin{equation}
\begin{aligned}
P_f^{AB}(a_A, a_B, a_C &\mid x_A, x_B, x_C )\\
= \fsandwich{\psi}{
\Big[&
\fketbra{m_{A, x_A}^{a_A}}{m_{A, x_A}^{a_A}} \wedge
\fketbra{m_{B, x_B}^{a_B}}{m_{B, x_B}^{a_B}} \wedge 
\fketbra{m_{C, x_C}^{a_C}}{m_{C, x_C}^{a_C}}
\Big] (\Pi_A^1 \wedge \Pi_B^1 \wedge \Pi_C^1)
}{\psi}
\\
=\Big| \fsandwich{\Omega}{&
\Big( u_{C,x_C}^{a_C*} f_{C_R} + v_{C,x_C}^{a_C*} f_{C_L} \Big)
\Big( u_{B,x_B}^{a_B*} f_{B_R} + v_{B,x_B}^{a_B*} f_{B_L} \Big) 
\Big( u_{A,x_A}^{a_A*} f_{A_R} + v_{A,x_A}^{a_A*} f_{A_L}\Big)
}{\psi} \Big|^2.
\end{aligned}
\end{equation}
Expanding the product, the non-vanishing terms give
\begin{equation}
\begin{aligned}
P_f^{AB}(a_A, a_B, a_C &\mid  x_A, x_B, x_C )\\
& =
\Big|
u_{A,x_A}^{a_A*}u_{B,x_B}^{a_B*}u_{C,x_C}^{a_C*}
\fbra{\Omega} f_{C_R} f_{B_R} f_{A_R}\fket{\psi}
+
v_{A,x_A}^{a_A*}v_{B,x_B}^{a_B*}v_{C,x_C}^{a_C*}
\fbra{\Omega} f_{C_L}f_{B_L}f_{A_L}\fket{\psi}
\\
 &\qquad +
u_{A,x_A}^{a_A*}u_{B,x_B}^{a_B*}v_{C,x_C}^{a_C*}
\fbra{\Omega} f_{C_L}f_{B_R}f_{A_R}\fket{\psi}
+
v_{A,x_A}^{a_A*}v_{B,x_B}^{a_B*}u_{C,x_C}^{a_C*}
\fbra{\Omega} f_{C_R}f_{B_L}f_{A_L}\fket{\psi}
\Big|^2.
\end{aligned}
\end{equation}
Evaluating the amplitudes using \cref{eq:statetwo}, we obtain $\forall a_i \in \{ 1_\pm\}$
\begin{equation}
\begin{aligned}
P_f^{AB}(a_A, a_B, a_C &\mid x_A, x_B, x_C )\\
&=\frac{1}{8}
\Big|
- u_{A,x_A}^{a_A*}u_{B,x_B}^{a_B*}u_{C,x_C}^{a_C*}
+ v_{A,x_A}^{a_A*}v_{B,x_B}^{a_B*}v_{C,x_C}^{a_C*}
- u_{A,x_A}^{a_A*}u_{B,x_B}^{a_B*}v_{C,x_C}^{a_C*}
+ v_{A,x_A}^{a_A*}v_{B,x_B}^{a_B*}u_{C,x_C}^{a_C*}
\Big|^2\\
& =
\frac{1}{8}
\Big|
- u_{A,x_A}^{a_A*}u_{B,x_B}^{a_B*}
\big(
u_{C,x_C}^{a_C*}+v_{C,x_C}^{a_C*}
\big)
+ v_{A,x_A}^{a_A*}v_{B,x_B}^{a_B*}
\big(
v_{C,x_C}^{a_C*}+u_{C,x_C}^{a_C*}
\big)
\Big|^2\\
& =
\frac{1}{8}
\left|
u_{C,x_C}^{a_C}+v_{C,x_C}^{a_C}
\right|^2
\cdot
\left|
\prod_{i \in \{A,B\}} u_{i,x_i}^{s_i}
-
\prod_{i \in \{A,B\}} v_{i,x_i}^{s_i}
\right|^2.
\end{aligned}
\end{equation}
Finally, summing over all outcomes $a_C$, and using \cref{eq:CoefficientsOrthonormality} , we obtain
\begin{equation}
   P_f^{AB}(a_A, a_B \mid x_A, x_B )
   =
   \frac{1}{4}
   \left|
   \prod_{i \in \{A,B\}} u_{i,x_i}^{s_i}
   -
   \prod_{i \in \{A,B\}} v_{i,x_i}^{s_i}
   \right|^2.
\end{equation}
Upon normalization, this yields \cref{eq:Dist_2}.

\subsection*{Bosonic (or distinguishable) particle case:}

Starting from the more general states of \cref{eq:qubit-state-psi-theta} and repeating the above calculations for the bosonic setting (or, equivalently, for distinguishable particles), one finds that the token-counting distributions coincide with those obtained in the fermionic case.

The difference arises in the conditional distributions, due to the absence of fermionic antisymmetry. For example in loop $ABC$, the resulting distribution is given by
\begin{equation}
P_b^{ABC}(a_A, a_B, a_C \mid (a_A, a_B, a_C) \in \{1_+, 1_- \}^3, x_A, x_B, x_C) =
\frac{1}{2}
\left|
\prod_{i\in V_l} u_{i, x_i}^{s_i}
+e^{\ii \sum_{\theta \in E_l } \varphi_\theta}
\prod_{i \in V_l} v_{i, x_i}^{s_i}
\right|^2.
\end{equation}
The same expression holds for the \(ACB\) loop.

In contrast, for a two-party loop, like $AB$ the relative phase between the two components of the conditional state differs from the fermionic case, leading to
\begin{equation}
    P_b^{AB}( a_A, a_B \mid (a_A, a_B) \in \{1_+, 1_- \}^2, x_A, x_B)
    =
    \frac{1}{2}
    \left|
    \prod_{i\in V_l} u_{i,x_i}^{s_i}
    + e^{\ii \sum_{\theta \in E_l } \varphi_\theta}
    \prod_{i\in V_l} v_{i,x_i}^{s_i}
    \right|^2.
\end{equation}
A similar expression holds for \(l \in \{AC,BC\}\).

\begin{remark}
This behavior admits a natural generalization to loops involving an arbitrary
number of parties. Loops of odd length have the same structure as the three-party
case, whereas loops of even length behave analogously to the two-party case. The
origin of this parity dependence is the fermionic reordering sign: already in the
two-party calculation, rewriting the state in the fixed ordering produces the
relative minus sign in \cref{eq:statetwo}, in contrast with the three-party case
\cref{eq:statethree}. This sign then propagates to the final conditional distribution. Such a reordering sign is absent in the bosonic, or
distinguishable-particle, setting.
\end{remark}

\begin{remark}
Looking at a single loop, one observes that when the length of the loop is odd, the fermionic distribution generated by the above game can be simulated by exactly the same strategy in a standard quantum (bosonic or distinguishable-particle) setting, with all qubit-state phases set to zero.
However, for even loops, the strategy must be modified in order to reproduce the fermionic distribution. One way to achieve this is to keep the same measurements but let one party (or more generally an odd number of parties) prepare a $\ket{\psi^-}$ state instead of $\ket{\psi^+}$. As can be seen from the bosonic calculations above, this compensates for the minus sign appearing in the fermionic distribution through a relative phase. 
Equivalently, one may keep all states equal to $\ket{\psi^+}$, but modify the measurements: for instance, one party (or, more generally, any odd number of parties) can flip the sign of one of the measurement coefficients for one of the conditional measurements, e.g., by changing the sign of the $v$-coefficients in the measurements corresponding to input $x=0$.
\end{remark}

\section{Full qubit quantum self-testing in a fixed loop of arbitrary size}
\label{appendix:FullSelfTesting}

This Supplemental Material is devoted to the proof of \cref{thm:fullselftesting}. 
In particular, we extend the partial self-testing result of \cite{sekatski2023partial} to a full self-testing argument by exploiting the additional information contained in the input-dependent part of the distribution. 

Let us briefly recall the trusted scenario of the qubit game in a fixed loop. Consider a loop $l$ of $N \geq 3$ parties, labeled sequentially as $1,\dots,N$. Each party $i$ prepares a bipartite maximally entangled state with relative phase $\varphi_i$,
\[
\ket{\psi_i}_{i_L,(i+1)_R}
= \frac{1}{\sqrt{2}}\bigl(\ket{01} + e^{\ii \varphi_i}\ket{10}\bigr),
\]
subject to the constraint that $\sum_{i=1}^{N} \varphi_i = \phi_l \in \{0,\pi\}$, and shares it with its left neighbor; see \cref{fig:QubitSelftesting}. The parties then receive the inputs and perform the measurement, with the PVM elements:
\begin{equation}
\begin{aligned}
    &\Pi_i^0=\ket{00}_{i_R i_L},\\
    &\Pi_i^2=\ket{11}_{i_R i_L},\\
    &\Pi_{i, x_i}^{1_+}=u^{+}_{i,x_i} \ket{01}_{i_R i_L} + v^{+}_{i,x_i} \ket{10}_{i_R i_L},\\
    &\Pi_{i, x_i}^{1_-}= u^{-}_{i,x_i} \ket{01}_{i_R i_L} + v^{-}_{i,x_i} \ket{10}_{i_R i_L}.
\end{aligned}
\end{equation}
Importantly, if the measurement outcomes of $\{ 1_\pm\}$ are coarse grained into a single outcome $1$, the measurement becomes independent of the input bits and corresponds to a token counting measurement~\cite{PhysRevA.105.022408}:
\begin{equation}
\begin{aligned}
\label{eq:QubitTokenCount}
    &\Pi_i^0=\ketbra{00}{00}_{i_R i_L},\\
    &\Pi_i^2= \ketbra{11}{11}_{i_R i_L},\\
    &\Pi_i^1=\ketbra{01}{01}_{i_R i_L} + \ketbra{10}{10}_{i_R i_L} .
\end{aligned}
\end{equation}
Operationally, this corresponds to measuring the local particle number in the fermionic case, that is, counting how many qubits in the state $\ket{1}$ a party has obtained. 

The observed coarse grained distribution is independent of the input choice and coincides with the token-counting scenario of~\cite{sekatski2023partial}. We note that the present situation is slightly different from that of~\cite{sekatski2023partial}: in their setting there are no inputs, whereas here the inputs are given to the parties before they perform their measurements.

Nevertheless, for any fixed value of the inputs, one can apply the partial self-testing result of~\cite{sekatski2023partial}.
It follows from the observed coarse grained statistics that the Hilbert space of each party $i$ decomposes into left and right subsystems, each of which further decomposes into system and junk components. We denote the corresponding spaces by $S_{i_R}$ and $J_{i_R}$ on the right, and by $S_{i_L}$ and $J_{i_L}$ on the left, where the $S_{i_R}, S_{i_L}$ are two-dimensional qubit systems. Moreover, the states are partially self-tested to take the form
\begin{equation}
\begin{aligned}
\label{eq:PartialStateTC}
\ket{\hat{\psi}_i}_{i, i+1}&= \frac{1}{\sqrt{2}} \big ( \ket{01}_{S_{i_L}S_{(i+1)_R}} \otimes \ket{\zeta_i^{c}}_{J_{i_L}J_{(i+1)_R}} + \ket{10}_{S_{i_L}S_{(i+1)_R}} \otimes \ket{\zeta_i^{a}}_{J_{i_L}J_{(i+1)_R}} \big ),
\end{aligned}
\end{equation}
up to local isometries, where $\ket{\zeta_i^{c}}$ and $\ket{\zeta_i^{a}}$ denote the clockwise and anticlockwise junk states, respectively, corresponding to whether the token (encoded in the system component) propagates clockwise or anticlockwise. Here, we use the hat notation and replace $\ket{\psi_i}$ with $\ket{\hat \psi_i}$ to denote the states derived by the partial token-counting self-testing of~\cite{sekatski2023partial}. Adopting the same notation, the coarse grained measurement elements are self-tested as
\begin{equation}
\begin{aligned}
\label{eq:SelfTestQubitTokenCount}
 \hat{\Pi}_i^{0} &= \ketbra{00}{00}_{S_{i_R}S_{i_L}} \otimes \id_{J_{i_R}, J_{i_L}}, \\
\hat{\Pi}_i^{2} &= \ketbra{11}{11}{S_{i_R}S_{i_L}} \otimes \id_{J_{i_R} J_{i_L}}, \\
\hat{\Pi}_i^{1} &= (\ketbra{01}{01}+ \ketbra{10}{10} ){S_{i_R}S_{i_L}} \otimes \id_{J_{i_R}, J_{i_L}} .
\end{aligned}
\end{equation}
So far, we have obtained the partially self-tested states and self-tested coarse grained measurements for an arbitrary but fixed value of the inputs. We now argue why these self-tested states and coarse grained measurements can are the same for all input values.

First, since the parties have access to the inputs only after the state preparation step, the states cannot depend on the inputs. 

Second, the self-tested measurements can, in principle, depend on the inputs, since the inputs are given before the measurement stage. However, from \cref{eq:SelfTestQubitTokenCount}, after coarse graining, there are no free parameter left in the measurement operators that can possibly depend on inputs. The only possible input dependence could therefore appear in the way the system and junk subspaces are identified, equivalently in the choice of the local self-testing isometries.

Importantly, the coarse grained token-counting distribution is independent of the inputs. Since, the coarse grained trusted strategy, consists of the input-independent states and the input-independent coarse grained measurements. Hence, the local isometry for each input value can be chosen independent of the input values. This allows us to choose common local isometries for all input values, and consequently a common system--junk decomposition for all input values.

We now extend this partial self-testing to get a full self-testing by restricting to the event that all the outputs are in $\{ 1_\pm\}$. Note that this event is in fact equivalent to all parties obtaining exactly one particle (or only one qubit with state $\ket{1}$). Conditioned on this event, the trusted global state reduces to 
\begin{equation}
\ket{\psi}^{\mathrm{cond}}
=
\frac{1}{\sqrt{2}}
\big(
\ket{01}_{1_L 1_R}\cdots\ket{01}_{N_L N_R}
+
e^{\ii \phi_l}\ket{10}_{1_L 1_R}\cdots\ket{10}_{N_L N_R}
\big),
\end{equation}
which corresponds to a $\mathrm{GHZ}^+$ or $\mathrm{GHZ}^-$ state, depending on $\phi_l$ being 0 or $\pi$, shared among the $N$ parties.
The PVM elements of the trusted measurement in case of this event are given by
\begin{equation}
\begin{aligned}
    \ket{m^{+}_{i,x_i}}
    &:= \Pi_{i, x_i}^{1_+}= u^{+}_{i,x_i} \ket{01}_{i_R i_L} + v^{+}_{i,x_i} \ket{10}_{i_R i_L}\\
    \ket{m^{-}_{i,x_i}}
    &:=\Pi_{i, x_i}^{1_-}= u^{-}_{i,x_i} \ket{01}_{i_R i_L} + v^{-}_{i,x_i} \ket{10}_{i_R i_L}.
\end{aligned}
\end{equation}
It is convenient to define the following observable in the subspace conditioned on the event that $\forall i, \, a_i \in \{1_\pm \}$ as
\[
M_{x_i}^{(i)} = \ket{m_{i,x_i}^{+}}\bra{m_{i,x_i}^{+}} - \ket{m_{i,x_i}^{-}}\bra{m_{i,x_i}^{-}},
\]
which can explicitly be written as,
\begin{equation}
\begin{aligned}
M_{x_i}^{(i)}
&=
\begin{pNiceMatrix}[first-row,last-col]
\ket{01} & \ket{10} &  \\
|u_{i,x_i}^{+}|^2 - |u_{i,x_i}^{-}|^2 & u_{i,x_i}^{+} v_{i,x_i}^{+*} - u_{i,x_i}^{-} v_{i,x_i}^{-*} & \bra{01} \\
v_{i,x_i}^{+} u_{i,x_i}^{+*} - v_{i,x_i}^{-} u_{i,x_i}^{-*} &|v_{i,x_i}^{+}|^2 - |v_{i,x_i}^{-}|^2 & \bra{10}
\end{pNiceMatrix}
\end{aligned}
\end{equation}

With the choice of coefficients specified in \cref{eq:meas_coeff_forA,eq:meas_coeff_forother}, these measurements correspond to special Pauli observables. In particular, for party $i \in \{ 2, \cdots, N\} $:
\begin{equation}
\begin{aligned}
M_0^{(i)}=\ketbra{m_{i,0}^{+}}{m_{i,0}^{+}} - \ketbra{m_{i,0}^{-}}{m_{i,0}^{-}}
&=\begin{pNiceMatrix}[first-row,last-col]
\ket{01} & \ket{10} & \\
 1 & 0 & \bra{01} \\
 0 & -1 & \bra{10}
\end{pNiceMatrix}
:= Z, \\
M_1^{(i)}=\ketbra{m_{i,1}^{+}}{m_{i,1}^{+}} - \ketbra{m_{i,1}^{-}}{m_{i,1}^{-}}
&=
\begin{pNiceMatrix}[first-row,last-col]
 \ket{01} & \ket{10} & \\
 0 & 1 & \bra{01}\\
 1 & 0 & \bra{10}
\end{pNiceMatrix}
:= X.
\end{aligned}
\end{equation}
and for party $i=1$, we obtain
\begin{equation}
\begin{aligned}
M_0^{(i)}=\ketbra{m_{i,0}^{+}}{m_{i,0}^{+}} - \ketbra{m_{i,0}^{-}}{m_{i,0}^{-}}
&=
\begin{pNiceMatrix}[first-row,last-col]
\ket{01} & \ket{10} & \\
 1 & 1 & \bra{01} \\
 1 & -1 & \bra{10}
\end{pNiceMatrix}= \frac{X+Z}{\sqrt{2}}, \\[10pt]
M_1^{(i)}=\ketbra{m_{i,1}^{+}}{m_{i,1}^{+}} - \ketbra{m_{i,1}^{-}}{m_{i,1}^{-}}
&=\frac{1}{\sqrt{2}}
\begin{pNiceMatrix}[first-row,last-col]
\ket{01} & \ket{10} & \\
 1 & 1 & \bra{01} \\
 1 & -1 & \bra{10}
\end{pNiceMatrix}= \frac{X-Z}{\sqrt{2}}.
\end{aligned}
\end{equation}

These correspond precisely to the $\mathrm{GHZ}^{\pm}$ self-testing scenario described in~\cite{baccari2020scalable}, implying that both the underlying state and the measurements are uniquely determined, up to local isometries.

To make this explicit, we introduce two Bell inequalities based on the statistics restricted to the event $\forall i, \,  a_i \in \{ 1_\pm\}$:
\begin{align}
    &\text{Bell inequality \Romannum{1}:  }\quad \Big\langle B_1 \equiv (N-1)(M_0^{(1)}+M_1^{(1)})M_1^{(2)} \cdots M_1^{(N)}+ \sum_{i=2}^N  (M_0^{(1)}-M_1^{(1)})M_0^{(i)} \Big\rangle \leq \beta_C,
\label{eq:Bell_1}    
\\
   &\text{Bell inequality \Romannum{2}:  }\quad\Big\langle B_2 \equiv -(N-1)(M_0^{(1)}+M_1^{(1)})M_1^{(2)} \cdots M_1^{(N)}+ \sum_{i=2}^N  (M_0^{(1)}-M_1^{(1)})M_0^{(i)} \Big\rangle \leq \beta_C,
\label{eq:Bell_2}
\end{align}
with classical and quantum bounds given by $\beta_C = 2(N-1)$ and $\beta_Q = 2(N-1)\sqrt{2}$, respectively, and the averages in \cref{eq:Bell_1,eq:Bell_2} are taken with respect to the conditional state. With the trusted strategy described above, the conditional distribution maximally violates Bell inequality \Romannum{1} (attaining $\beta_Q$) when $\phi_l = 0$, corresponding to a $\mathrm{GHZ}^+$ state. Similarly, Bell inequality \Romannum{2} is maximally violated when $\phi_l = \pi$, corresponding to a $\mathrm{GHZ}^-$ state. 
Note that this association of Bell inequalities~\Romannum{1} and~\Romannum{2}
with $\mathrm{GHZ}^+$ and $\mathrm{GHZ}^-$, respectively, is meaningful only for
the fixed measurement choices specified above.

Now consider the untrusted conditional measurements, $\hat{M}_{x_i}^{(i)}$, that is applied by  party $i$ conditioned on
obtaining exactly one token. This occurs when party $i-1$ sends its token to $i$ while $i$ sends its token to $i+1$, or when both $i-1$ and $i$ retain their tokens. In this case, partial self-testing, \cref{eq:PartialStateTC}, implies that the global state takes the form
\begin{equation}
\begin{aligned}
\ket{\hat{\psi}_i}^\mathrm{cond} &= \frac{1}{\sqrt{2}} \Big ( \ket{01}_{S_{(i-1)_L}S_{i_R}} \ket{01}_{S_{i_L},S_{(i+1)_R}} \otimes \ket{\zeta_{i-1}^{c}}_{J_{(i-1)_L}J_{i_R}} \otimes \ket{\zeta_i^{c}}_{J_{i_L}J_{(i+1)_R}}  \\
&\qquad \qquad +\ket{10}_{S_{(i-1)_L}S_{i_R}}  \ket{10}_{S_{i_L}S_{(i+1)_R}} \otimes \ket{\zeta_{i-1}^{a}}_{J_{(i-1)_L}J_{i_R}} \otimes \ket{\zeta_i^{a}}_{J_{i_L}J_{(i+1)_R}} \Big )  \\
& \quad \bigotimes_{j\notin \{i, i-1\}}\frac{1}{\sqrt{2}} \Big ( \ket{01}_{S_{j_L}S_{(j+1)_R}} \otimes \ket{\zeta_j^{c}}_{J_{j_L}J_{(j+1)_R}} + \ket{10}_{S_{j_L}S_{(j+1)_R}} \otimes \ket{\zeta_j^{a}}_{J_{j_L}J_{(j+1)_R}} \Big ).
\end{aligned}
\end{equation}
For completeness, we also define the global conditional state corresponding to the event that all parties obtain one token. By the partial self-testing result, this state is given by
\begin{equation}
\begin{aligned}
\label{eq:ConditionalState}
\ket{\hat{\psi}}^\mathrm{cond} = \frac{1}{\sqrt{2}}\Big( \ket{01, \cdots, 01}_{S_{1_R} S_{1_L} \cdots S_{N_R} S_{N_L}} \otimes \ket{\zeta_1^{a}}_{J_{1_L} J_{2_R}} \otimes \cdots \otimes \ket{\zeta_N^{a}}_{J_{N_L} J_{1_R}}  \\ +\ket{10, \cdots, 10}_{S_{1_R} S_{1_L} \cdots S_{N_R} S_{N_L}}\otimes \ket{\zeta_1^{c}}_{J_{1_L} J_{2_R}}\otimes \cdots \otimes \ket{\zeta_N^{c}}_{J_{N_L} J_{1_R}}  \Big).
\end{aligned}
\end{equation}
Note that, from the perspective of party $i$, the states $\ket{\hat{\psi}}^{\mathrm{cond}}$ and $\ket{\hat{\psi}_i}^{\mathrm{cond}}$ are equivalent. This is the state on which the conditional measurement of party $i$, $\hat{M}_{x_i}^{(i)}$, acts.

Maximal violation of either Bell inequalities identifies the stabilizers of the conditional global state, $\ket{\hat{\psi}}^{\rm cond}$. Define
\begin{equation}
\begin{aligned}
\label{eq:observableToXZ}
\mathcal{X}_{1} &= \frac{\hat{M}^{(1)}_0 + \hat{M}^{(1)}_1}{\sqrt{2}} &\qquad \mathcal{Z}_{1} &= \frac{\hat{M}^{(1)}_0 - \hat{M}^{(1)}_1}{\sqrt{2}} \\
\mathcal{X}_{i} &= \hat{M}^{(i)}_1  &\qquad \mathcal{Z}_{i} &= \hat{M}^{(i)}_0 \qquad \forall i \in \{2,..., N \}.
\end{aligned}
\end{equation}
The following lemma formalizes the resulting stabilizer structure. 

\begin{lemma} \label{lem:FirstBell}
If Bell inequality \Romannum{1}, \cref{eq:Bell_1}, is maximally violated, $\mathcal{X}_1 \cdots \mathcal{X}_N$ is a stabilizer of the conditional state:
\begin{equation}
\label{eq:stab1}
     \mathcal{X}_1 \cdots \mathcal{X}_N \ket{\hat{\psi}}^\mathrm{cond}= \ket{\hat{\psi}}^\mathrm{cond},
\end{equation}
and $\forall i \in \{1, \cdots, N\}$:
\begin{equation}
\begin{aligned}
\mathcal{X}_i^{2} \ket{\hat{\psi}}^\mathrm{cond}= \mathcal{Z}_i^{2} \ket{\hat{\psi}}^\mathrm{cond}=\ket{\hat{\psi}}^\mathrm{cond}, \quad
\{ \mathcal{X}_i, \mathcal{Z}_i \} \ket{\hat{\psi}}^\mathrm{cond}= 0,
\end{aligned}
\end{equation}
\end{lemma}
\begin{proof}
It is given in Supplemental Material~\ref{appendix:FirstBell}.
\end{proof}

\begin{lemma} \label{lem:SecondBell}
If Bell inequality \Romannum{2}, \cref{eq:Bell_2}, is maximally violated, then $-\mathcal{X}_1 \cdots \mathcal{X}_N $ is a stabilizer of the conditional state:
\begin{equation}
\label{eq:stab2}
    -\mathcal{X}_1 \cdots \mathcal{X}_N   \ket{\hat{\psi}}^\mathrm{cond}= \ket{\hat{\psi}}^\mathrm{cond},
\end{equation}
and $\forall i \in \{1, \cdots, N\}$:
\begin{equation}
\begin{aligned}
\mathcal{X}_i^{2} \ket{\hat{\psi}}^\mathrm{cond}= \mathcal{Z}_i^{2} \ket{\hat{\psi}}^\mathrm{cond}=\ket{\hat{\psi}}^\mathrm{cond}, \quad
\{ \mathcal{X}_i, \mathcal{Z}_i \} \ket{\hat{\psi}}^\mathrm{cond}= 0,
\end{aligned}
\end{equation}
\end{lemma}
\begin{proof}
It is given in Supplemental Material~\ref{appendix:SecondBell}.
\end{proof}

Before proceeding, we define the subspace on which the measurements, $\hat{M}_{x_i}^{(i)}$, operates. Specifically, the measurement is performed on the support of the reduced state obtained by tracing out all subsystems except those of party $i$:
\begin{equation}\label{eq:def-H-i-cond}
H_i^{\mathrm{cond}}
:= \operatorname{supp}\!\left(
\operatorname{Tr}_{\{S_j, J_j\}_{j \neq i}}
\big[
\,|\hat{\psi}_i\rangle^{\mathrm{cond}}\langle \hat{\psi}_i|^{\mathrm{cond}}
\big]
\right)=\operatorname{supp}\!\left(
\operatorname{Tr}_{\{S_j, J_j\}_{j \neq i}}
\big[
\,|\hat{\psi}\rangle^{\mathrm{cond}}\langle \hat{\psi}|^{\mathrm{cond}}
\big]
\right).
\end{equation}
This evaluates to
\begin{equation}
\begin{aligned}
H_i^{\mathrm{cond}}
= \operatorname{supp}\!\Big(
&\ketbra{10}{10}_{S_{i_R} S_{i_L}} \otimes 
\operatorname{Tr}_{J_{(i-1)_L}} \ketbra{\zeta_{i-1}^c}{\zeta_{i-1}^c}
\otimes 
\operatorname{Tr}_{J_{(i+1)_R}} \ketbra{\zeta_i^c}{\zeta_i^c} \\
&+ 
\ketbra{01}{01}_{S_{i_R} S_{i_L}} \otimes 
\operatorname{Tr}_{J_{(i-1)_L}} \ketbra{\zeta_{i-1}^a}{\zeta_{i-1}^a}
\otimes 
\operatorname{Tr}_{J_{(i+1)_R}} \ketbra{\zeta_i^a}{\zeta_i^a}
\Big),
\end{aligned}
\end{equation}
and since the two terms are orthogonal, this simplifies to
\begin{equation}
\begin{aligned}
H_i^{\mathrm{cond}}
= {}&
\operatorname{supp}\!\Big(
\ketbra{10}{10}_{S_{i_R} S_{i_L}} \otimes 
\operatorname{Tr}_{J_{(i-1)_L}} \ketbra{\zeta_{i-1}^c}{\zeta_{i-1}^c}
\otimes 
\operatorname{Tr}_{J_{(i+1)_R}} \ketbra{\zeta_i^c}{\zeta_i^c}
\Big) \\
&\oplus
\operatorname{supp}\!\Big(
\ketbra{01}{01}_{S_{i_R} S_{i_L}} \otimes 
\operatorname{Tr}_{J_{(i-1)_L}} \ketbra{\zeta_{i-1}^a}{\zeta_{i-1}^a}
\otimes 
\operatorname{Tr}_{J_{(i+1)_R}} \ketbra{\zeta_i^a}{\zeta_i^a}
\Big) \\
= {}&
\mathrm{span}\{\ket{10}\}_{S_{i_R} S_{i_L}} \otimes J^c_{i_R} \otimes J^c_{i_L}
\;\oplus\;
\mathrm{span}\{\ket{01}\}_{S_{i_R} S_{i_L}} \otimes J^a_{i_R} \otimes J^a_{i_L},
\end{aligned}
\end{equation}
where
\begin{equation}
\begin{aligned}
J^c_{i_R} &:= \operatorname{supp}\!\left(
\operatorname{Tr}_{J_{(i-1)_L}} \ketbra{\zeta^c_{i-1}}{\zeta^c_{i-1}}
\right), \\
J^c_{i_L} &:= \operatorname{supp}\!\left(
\operatorname{Tr}_{J_{(i+1)_R}} \ketbra{\zeta^c_i}{\zeta^c_i}
\right), \\
J^a_{i_R} &:= \operatorname{supp}\!\left(
\operatorname{Tr}_{J_{(i-1)_L}} \ketbra{\zeta^a_{i-1}}{\zeta^a_{i-1}}
\right), \\
J^a_{i_L} &:= \operatorname{supp}\!\left(
\operatorname{Tr}_{J_{(i+1)_R}} \ketbra{\zeta^a_i}{\zeta^a_i}
\right).
\end{aligned}
\end{equation}

With this identification of conditional subspaces, \cref{lem:FirstBell,lem:SecondBell} imply that $\mathcal{X}_i^2 = \mathcal{Z}_i^2 = \id$ and $\{\mathcal{X}_i, \mathcal{Z}_i\} = 0$ on $H_i^{\mathrm{cond}}$.

\begin{lemma}
\label{lem:Zsep_SysJunk}
For all $i \in \{1,\dots,N\}$, the operator $\mathcal{Z}_i$ acts trivially on the junk subsystem within the conditional subspace $H_i^{\mathrm{cond}}$. More precisely, it has the form
\begin{equation}
\label{eq:Zsep}
\begin{aligned}
\mathcal Z_i\big|_{H_i^{\mathrm{cond}}} = Z_{S_i}\otimes \id_{J_i}\big|_{H_i^{\mathrm{cond}}},
\end{aligned}
\end{equation}
where $Z = (\ket{01}\bra{01} - \ket{10}\bra{10})_{S_{i_R} S_{i_L}}$.
\end{lemma}

\begin{proof}
It is given in Supplemental Material~\ref{appendix:Zsep}.
\end{proof}

\begin{lemma}
\label{lem:X_HtoH}
For any $i \in \{1,\dots,N\}$, the operator $\mathcal{X}_i$ maps $H_i^{\mathrm{cond}}$ into itself.
\end{lemma}

\begin{proof}
Assuming that either Bell inequality \Romannum{1} or \Romannum{2} is maximally violated, \cref{lem:FirstBell,lem:SecondBell} 
imply that \cref{eq:stab1} or \cref{eq:stab2} holds. Multiplying both sides by $\prod_{j \neq i} \mathcal{X}_j$ and using $\mathcal{X}_j^2 \ket{\psi}^{\mathrm{cond}} = \ket{\psi}^{\mathrm{cond}}$ for all $j$, we obtain
\begin{equation}
\mathcal{X}_i \ket{\psi}^{\mathrm{cond}} 
= \pm \left( \prod_{j \neq i} \mathcal{X}_j \right) \ket{\psi}^{\mathrm{cond}}.
\end{equation}
Then, taking the partial trace of both sides with respect to all parties except party $i$ and using the definition of $H_i^{\mathrm{cond}}$ from~\cref{eq:def-H-i-cond}, we find that $\mathcal X_i(H_i^{\mathrm{cond}})=H_i^{\mathrm{cond}}.$
\end{proof}

\begin{lemma} 
\label{lem:Xsep_SysJunk}
For any $i \in \{1,\dots,N\}$, the operator $\mathcal{X}_i$ restricted to the subspace $H_i^{\mathrm{cond}}$ has the form
\begin{equation}
\label{eq:Xsep}
\begin{aligned}
\mathcal{X}_i\big|_{H_i^{\mathrm{cond}}} &= \left( \frac{X + \ii Y}{2} \otimes T_i \;+\; \frac{X - \ii Y}{2} \otimes T_i^\dagger \right) \big|_{H_i^{\mathrm{cond}}},
\end{aligned}
\end{equation}
where 
\[
X = (\ket{10}\bra{01} + \ket{01}\bra{10})_{S_{i_R} S_{i_L}}, 
\qquad 
Y = (\ii \ket{10}\bra{01} - \ii \ket{01}\bra{10})_{S_{i_R} S_{i_L}},
\]
and $T_i$ is a unitary operator mapping $J^c_{i_L} \otimes J^c_{i_R}$ to $J^a_{i_L} \otimes J^a_{i_R}$.
\end{lemma}

\begin{proof}
The proof is given in Supplemental Material~\ref{appendix:Xsep_SJ_3}.
\end{proof}

\begin{corollary} 
\label{cor:COR1}
\cref{lem:Xsep_SysJunk}, together with maximal violation of one of the Bell inequalities, implies
\begin{equation}
\label{eq:COR1} 
\begin{aligned}
\left( \bigotimes_{i=1}^N T_i \right)\;\ket{\zeta^c_{1} \cdots \zeta^c_{N}} 
&= e^{\ii \phi_l} \ket{\zeta^a_{1} \cdots \zeta^a_{N}}, \\
\left( \bigotimes_{i=1}^N T_i^\dagger \right) \;\ket{\zeta^a_{1} \cdots \zeta^a_{N}} 
&= e^{-\ii \phi_l} \ket{\zeta^c_{1} \cdots \zeta^c_{N}}.
\end{aligned}
\end{equation}
\end{corollary}

\begin{proof}
\cref{lem:FirstBell,lem:SecondBell} imply that \cref{eq:stab1,eq:stab2} hold. Substituting the decomposition~\cref{eq:Xsep} into the stabilizer relations and using the structure of the conditional state, \cref{eq:ConditionalState}, yields the result.
\end{proof}

\begin{lemma} 
\label{lem:JunkRLSep}
For each $i \in \{1,\dots,N\}$, the operator $T_i: J^c_{i_L} \otimes J^c_{i_R} \rightarrow J^a_{i_L} \otimes J^a_{i_R}$ from \cref{lem:Xsep_SysJunk} factorizes across the left and right junk subsystems. More precisely, there exist unitary operators $T_i^{(L)}:J^c_{i_L}\to J^a_{i_L}$ and $T_i^{(R)}:J^c_{i_R}\to J^a_{i_R}$ such that
\begin{equation}
\begin{aligned}
T_i \ket{\zeta^c_{i-1}} \ket{\zeta^c_i}
&= (T_i^{(L)} \otimes T_i^{(R)}) \ket{\zeta^c_{i-1}} \ket{\zeta^c_i}.
\end{aligned}
\end{equation}
\end{lemma}

\begin{proof}
This is a consequence of \cref{cor:COR1} and applying \cref{lem:ProductStructureLemmaInfDim} of Supplemental Material~\ref{appendix:Tsep_RL_3}.
\end{proof}

\begin{corollary}
\label{cor:COR2}   
\cref{lem:JunkRLSep,cor:COR1} imply that, for all $i \in \{1,\dots,N\}$,
\begin{equation}
\label{eq:COR2}
T_i^{(L)} T_{i+1}^{(R)} \ket{\zeta^c_i}
= e^{\ii \hat{\varphi}_i}\ket{\zeta^a_i},
\end{equation}
with phases satisfying $\sum_{i=1}^N \hat{\varphi}_i = \phi_l$.
\end{corollary}

\begin{proof}
Substituting the decomposition from \cref{lem:JunkRLSep} into the global relation~\cref{eq:COR1}, and using that the operators act on distinct Hilbert spaces, we obtain the local relations above. Matching the global phases on both sides then yields $\sum_{i=1}^N \hat{\varphi}_i = \phi_l$.
\end{proof}

We now define the local unitary operators $U_{i_R}$ and $U_{i_L}$ based on the action of the second-round measurement on the junk subsystems, namely the operators $T_i^{(R)}$ and $T_i^{(L)}$:
\begin{equation}
\begin{aligned}
U_{i_R} &:= \ketbra{0}{0}_{S_{i_R}} \otimes \id_{J_{i_R}} 
+ \ketbra{1}{1}_{S_{i_R}} \otimes T_i^{(R)}, \\
U_{i_L} &:= \ketbra{1}{1}_{S_{i_L}} \otimes \id_{J_{i_L}} 
+ \ketbra{0}{0}_{S_{i_L}} \otimes T_i^{(L)}.
\end{aligned}
\end{equation}
Applying these unitary operators to the partially self-tested state $\ket{\hat{\psi}_i}$, by construction and using \cref{cor:COR2}, we obtain
\begin{equation}
U_{i_L} U_{(i+1)_R} \ket{\hat{\psi}_i}
= \frac{1}{\sqrt{2}}
\bigl(
e^{\ii \hat{\varphi}_i} \ket{01}_{S_{i_L}, S_{(i+1)_R}}
+ \ket{10}_{S_{i_L}, S_{(i+1)_R}}
\bigr)
\otimes
\ket{\zeta_i^{a}}_{J_{i_L} J_{(i+1)_R}}.
\end{equation}
Next,
defining $U_i := U_{i_R} U_{i_L}$ and considering the action of these unitaries on the measurement operators yields
\begin{equation}
\begin{aligned}
U_i \mathcal{Z}_i U_i^\dagger
&= U_i \big( \ketbra{01}{01} \otimes \id - \ketbra{10}{10} \otimes \id \big) U_i^\dagger, \\
&= Z \otimes \id,
\end{aligned}
\end{equation}
and
\begin{equation}
\begin{aligned}
U_i \mathcal{X}_i U_i^\dagger
&= U_i \big( \ketbra{01}{10} \otimes T_i + \ketbra{10}{01} \otimes T_i^\dagger \big) U_i^\dagger, \\
&= \ketbra{01}{10} \otimes T_i T_i^\dagger 
+ \ketbra{10}{01} \otimes T_i T_i^\dagger, \\
&= X \otimes \id.
\end{aligned}
\end{equation}

\subsection{\texorpdfstring{\cref{thm:fullselftesting}}{Theorem~\ref{thm:fullselftesting}} for \texorpdfstring{$N=2$}{N=2} with stronger assumptions}
\label{appendix:fullselftesting_N2}

The following statement is a weaker analogue of
\cref{thm:fullselftesting} for a two-party loop. Unlike the case of 
\(N\geq 3\), we do not derive the full structure from the distribution alone; instead, we assume a partial structure for the states and measurements and show that they together with the distribution are sufficient to recover the trusted two-qubit strategy up to local unitaries.

The trusted qubit strategy is the same as the case for $N \geq 3$ and as described in Supplemental Material~\ref{appendix:ExperimentAndMainResult}: Consider a fixed two-party loop \(l\) of parties $1, 2$. The states are give by~\cref{eq:qubit-state-psi-theta} or
\begin{equation}
\ket{\psi_i}_{i_L,(i+1)_R}
=
\frac{1}{\sqrt{2}}
\left(
\ket{01}_{i_L,(i+1)_R}
+
e^{\ii\varphi_i}
\ket{10}_{i_L,(i+1)_R}
\right),
\qquad i\in\{1,2\},
\end{equation}
where indices are understood modulo \(2\). We assume that the phases satisfy
\begin{equation}
\label{eq:phase_constraint_N2}
\varphi_1+\varphi_2=\phi_l,
\qquad
\phi_l\in\{0,\pi\}.
\end{equation}
The trusted measurements are given by \cref{eq:QubitProjectors} with the coefficients as in \cref{eq:meas_coeff_forA} for party \(1\) and \cref{eq:meas_coeff_forother} for party \(2\). 

Note that restricting to the event when $\forall i, \, a_i \in \{1_\pm\}$, the conditional global state is a two-qubit Bell state, with relative phase \(\phi_l\). Moreover, defining the conditional observable
\begin{equation}
    M_{x_i}^{(i)}= \Pi_{i, x_i}^{1_+} - \Pi_{i, x_i}^{1_-}
\end{equation}
and
\begin{equation}
\begin{aligned}
\mathcal{X}_{1}
&=
\frac{M^{(1)}_0+M^{(1)}_1}{\sqrt{2}},
&
\mathcal{Z}_{1}
&=
\frac{M^{(1)}_0-M^{(1)}_1}{\sqrt{2}},
\\
\mathcal{X}_{2}
&=
M^{(2)}_1,
&
\mathcal{Z}_{2}
&=
M^{(2)}_0,
\end{aligned}
\end{equation}
we find out that they satisfy
\[
\mathcal{Z}_i=Z_{i_R i_L},
\qquad
\mathcal{X}_i=X_{i_R i_L},
\qquad i\in\{1,2\},
\]
where
\begin{equation}
Z_{i_R i_L}
=
\left(
\ketbra{01}{01}
-
\ketbra{10}{10}
\right)_{{i_R}{i_L}},
\qquad
X_{i_R i_L}
=
\left(
\ketbra{01}{10}
+
\ketbra{10}{01}
\right)_{{i_R}{i_L}} .
\end{equation}
Therefore, the conditional distribution maximally violates Bell
inequality \Romannum{1} when \(\phi_l=0\), and Bell inequality \Romannum{2} when \(\phi_l=\pi\).

\begin{theorem}[full self-testing for a fixed two-party loop under extra structural assumptions]
\label{thm:fullselftesting_N2}

Consider an arbitrary untrusted strategy producing the same conditional
distribution for the two-party loop as above. We assume that, after post-selection
on the event that both parties obtain outputs $1_\pm$, the source states have the partially self-tested form
\begin{equation}
\label{eq:partial_state_N2}
\ket{\hat{\psi}_i}_{i_L,(i+1)_R}
=
\frac{1}{\sqrt{2}}
\left(
\ket{01}_{S_{i_L}S_{(i+1)_R}}
\otimes
\ket{\zeta_i^{c}}_{J_{i_L}J_{(i+1)_R}}
+
\ket{10}_{S_{i_L}S_{(i+1)_R}}
\otimes
\ket{\zeta_i^{a}}_{J_{i_L}J_{(i+1)_R}}
\right),
\end{equation}
up to local isometries. 
Let \(\hat{M}_{x_i}^{(i)}\), for \(x_i\in\{0,1\}\), denote the corresponding
untrusted conditional observables, and define
\begin{equation}
\label{eq:observableToXZ_N2}
\begin{aligned}
\mathcal{X}_{1}
&=
\frac{\hat{M}^{(1)}_0+\hat{M}^{(1)}_1}{\sqrt{2}},
&
\mathcal{Z}_{1}
&=
\frac{\hat{M}^{(1)}_0-\hat{M}^{(1)}_1}{\sqrt{2}},
\\
\mathcal{X}_{2}
&=
\hat{M}^{(2)}_1,
&
\mathcal{Z}_{2}
&=
\hat{M}^{(2)}_0 .
\end{aligned}
\end{equation}
Assume that the operators \(\mathcal{Z}_i\) and \(\mathcal{X}_i\) on the conditional subspace
\begin{align}
H_i^{\rm cond}
=
\mathrm{span}\{\ket{10}\}_{S_{i_R}S_{i_L}}\otimes J^c_{i_R}\otimes J^c_{i_L}
\;\oplus\;
\mathrm{span}\{\ket{01}\}_{S_{i_R}S_{i_L}}\otimes J^a_{i_R}\otimes J^a_{i_L},
\end{align}
have the following
structure:
\begin{equation}
\label{eq:Z_structure_N2}
\mathcal{Z}_i
=
Z_{S_i}\otimes \id_{J_i},
\end{equation}
and
\begin{equation}
\label{eq:X_structure_N2}
\mathcal{X}_i
=
\frac{X+\ii Y}{2}\otimes T_i
+
\frac{X-\ii Y}{2}\otimes T_i^\dagger,
\end{equation}
where
\begin{equation}
\begin{aligned}
X_{S_i}
&=
\left(
\ketbra{01}{10}
+
\ketbra{10}{01}
\right)_{S_{i_R}S_{i_L}},
\\
Y_{S_i}
&=
\left(
\ii\ketbra{10}{01}
-
\ii\ketbra{01}{10}
\right)_{S_{i_R}S_{i_L}},
\end{aligned}
\end{equation}
and
\[
T_i:
J^c_{i_L}\otimes J^c_{i_R}
\longrightarrow
J^a_{i_L}\otimes J^a_{i_R}
\]
is unitary.
Finally, assume that each \(T_i\) factorizes across the left and right junk
supports on the relevant state:
\begin{equation}
\label{eq:T_factor_N2}
T_i
\left(
\ket{\zeta^c_{i-1}}_{J_{(i-1)_L}J_{i_R}}
\ket{\zeta^c_i}_{J_{i_L}J_{(i+1)_R}}
\right)
=
\left(
T_i^{(L)}\otimes T_i^{(R)}
\right)
\left(
\ket{\zeta^c_{i-1}}_{J_{(i-1)_L}J_{i_R}}
\ket{\zeta^c_i}_{J_{i_L}J_{(i+1)_R}}
\right),
\end{equation}
where \(T_i^{(L)}\) and \(T_i^{(R)}\) are unitary maps on the corresponding
left and right junk supports.

Then the untrusted two-party-loop strategy is equivalent to the trusted qubit
strategy up to local unitary transformations. More precisely, there exist local
unitaries \(U_{i_R}\) and \(U_{i_L}\), acting on
\(S_{i_R}\otimes J_{i_R}\) and \(S_{i_L}\otimes J_{i_L}\), respectively, such
that
\begin{equation}
\label{eq:N2_state_selftest}
U_{i_L}U_{(i+1)_R}
\ket{\hat{\psi}_i}_{i_L,(i+1)_R}
=
\frac{1}{\sqrt{2}}
\left(
\ket{01}_{S_{i_L}S_{(i+1)_R}}
+
e^{\ii\hat{\varphi}_i}
\ket{10}_{S_{i_L}S_{(i+1)_R}}
\right)
\otimes
\ket{\zeta_i}_{J_{i_L}J_{(i+1)_R}},
\end{equation}
with
\begin{equation}
\hat{\varphi}_1+\hat{\varphi}_2=\phi_l .
\end{equation}
Moreover, defining \(U_i:=U_{i_R}U_{i_L}\), the conditional observables satisfy
\begin{equation}
\label{eq:N2_measurement_selftest}
\begin{aligned}
U_i\mathcal{Z}_iU_i^\dagger
&=
Z_{S_i}\otimes \id_{J_i},\\
U_i\mathcal{X}_iU_i^\dagger
&=
X_{S_i}\otimes \id_{J_i}.
\end{aligned}
\end{equation}
\end{theorem}

\begin{proof}
Maximal violation of Bell inequality \Romannum{1} or \Romannum{2} (\cref{lem:FirstBell,lem:SecondBell} and for $N=2$) implies that 
\begin{equation}
e^{\ii \phi_l} \mathcal{X}_1 \mathcal{X}_2 \ket{\hat{\psi}}^{\mathrm{cond}} = \ket{\hat{\psi}}^{\mathrm{cond}}.
\end{equation}
Now by plugging in the structure that is assumed for the conditional observables \cref{eq:X_structure_N2}, and the structure for the states \cref{eq:partial_state_N2} in the above equation one can repeat the argument of \cref{cor:COR1} to obtain the corresponding relations~\cref{eq:COR1} for the two-loop:
\begin{equation}
\begin{aligned}
T_1 T_2 \;\ket{\zeta^c_{1} \zeta^c_{2}} 
&= e^{\ii \phi_l} \ket{\zeta^a_{1} \zeta^a_{2}}, \\
T_1^\dagger T_2^\dagger \;\ket{\zeta^a_{1} \zeta^a_{2}} 
&= e^{-\ii \phi_l} \ket{\zeta^c_{1} \zeta^c_{2}}.
\end{aligned}
\end{equation}
Furthermore, by using the structure of $T$ operators in the assumption \cref{eq:T_factor_N2}, an argument analogous to that of \cref{cor:COR2} applies, yielding \cref{eq:COR2} for the two-party loop:
\begin{equation}
T_i^{(L)} T_{i+1}^{(R)} \ket{\zeta^c_i}
= e^{\ii \hat{\varphi}_i}\ket{\zeta^a_i},
\end{equation}
where $\sum_{i=1}^N \hat{\varphi}_i = \phi_l$.
\end{proof}

\section{No local bosonic simulation for the fermionic distributions}
\label{appendix:contradiction}

Here we prove the main result of this work, \cref{thm:MainTheoremNoSimultaneuousSimulation}. The proof proceeds by contradiction. Suppose that there exists a local bosonic (or distinguishable-particle) strategy for the same rewiring experiment involving three parties $A$, $B$, and $C$, which reproduces all fermionic distributions
\[
P_f^{ABC}, \quad P_f^{ACB}, \quad P_f^{AB}, \quad P_f^{AC}, \quad P_f^{BC}.
\]
Let $\ket{\psi_\alpha}$, $\ket{\psi_\beta}$, and $\ket{\psi_\gamma}$ denote the states prepared by parties $A$, $B$, and $C$, respectively.

We first apply \cref{thm:fullselftesting} to the loops $ABC$ and $ACB$ to infer properties of the states and measurements in these configurations. 

By locality, the prepared states are independent of the loop configuration and
therefore remain fixed across all loops. In contrast, the measurement operators
associated with a given party may depend partially on the loop. For instance,
the states prepared by the parties may carry labels identifying their sender, so
that a receiving party can, in principle, infer the origin of an incoming
subsystem and adapt its measurement accordingly. However, since the parties are
local and have no access to the global loop configuration, this is the only
available source of loop-dependent information. Thus, any admissible variation
in the measurement strategy can only arise from differences in the incoming
subsystem.

To make this dependence explicit, we refine the notation by adding to each
measurement operator a superscript indicating the incoming state on the right of
the party.

Consider the loop $ABC$. The qubit strategy in this loop must reproduce the distribution $P_f^{ABC}$. This distribution corresponds to the trusted qubit strategy described in Supplemental Material~\ref{appendix:ExperimentAndMainResult} where the phases satisfy $\varphi_\alpha + \varphi_\beta + \varphi_\gamma = 0$. Then \cref{thm:fullselftesting} implies that any qubit strategy reproducing this distribution admits local unitary operators $U_{i_R}$ and $U_{i_L}$, for all $i \in \{A,B,C\}$, such that
\begin{align}
\label{eq:loopABC_alpha}
U_{A_L}^{(\gamma)} U_{B_R}^{(\alpha)} \ket{\hat{\psi}_\alpha} &=  \frac{1}{\sqrt{2}}
\bigl(
e^{\ii \hat{\varphi}^{ABC}_\alpha} \ket{01}
+ \ket{10}
\bigr)_{S_{\alpha_R} S_{\alpha_L}}
\otimes
\ket{\zeta_\alpha^{a}}_{J_{\alpha}} \\
\label{eq:loopABC_beta}
U_{B_L}^{(\alpha)} U_{C_R}^{(\beta)} \ket{\hat{\psi}_\beta} &=  \frac{1}{\sqrt{2}}
\bigl(
e^{\ii \hat{\varphi}^{ABC}_\beta} \ket{01}
+ \ket{10}
\bigr)_{S_{\beta_R} S_{\beta_L}}
\otimes
\ket{\zeta_\beta^{a}}_{J_{\beta}} \\
\label{eq:loopABC_gamma}
U_{C_L}^{(\beta)} U_{A_R}^{(\gamma)} \ket{\hat{\psi}_\gamma} &=  \frac{1}{\sqrt{2}}
\bigl(
e^{\ii \hat{\varphi}^{ABC}_\gamma} \ket{01}
+ \ket{10}
\bigr)_{S_{\gamma_R} S_{\gamma_L}}
\otimes
\ket{\zeta_\gamma^{a}}_{J_{\gamma}},
\end{align}
where $\hat{\varphi}^{ABC}_\alpha + \hat{\varphi}^{ABC}_\beta + \hat{\varphi}^{ABC}_\gamma = \phi_{ABC} = 0$. 
On the right-hand side of the above equations, the subsystems are labeled according to the state labeling for convenience. For the precise mapping between state and party labelings, see \cref{eq:labellingMap} or \cref{fig:labels}.

Similarly for loop $ACB$,
\begin{align}
\label{eq:loopACB_alpha}
U_{A_L}^{(\beta)} U_{C_R}^{(\alpha)} \ket{\hat{\psi}_\alpha}
&= \frac{1}{\sqrt{2}}
\bigl(
e^{\ii \hat{\varphi}^{ACB}_\alpha} \ket{01}
+ \ket{10}
\bigr)_{S_{\alpha_R} S_{\alpha_L}}
\otimes
\ket{\zeta_\alpha^{a}}_{J_{\alpha}} \\
\label{eq:loopACB_gamma}
U_{C_L}^{(\alpha)} U_{B_R}^{(\gamma)} \ket{\hat{\psi}_\gamma}
&= \frac{1}{\sqrt{2}}
\bigl(
e^{\ii \hat{\varphi}^{ACB}_\gamma} \ket{01}
+ \ket{10}
\bigr)_{S_{\gamma_R} S_{\gamma_L}}
\otimes
\ket{\zeta_\gamma^{a}}_{J_{\gamma}} \\
\label{eq:loopACB_beta}
U_{B_L}^{(\gamma)} U_{A_R}^{(\beta)} \ket{\hat{\psi}_\beta} 
&= \frac{1}{\sqrt{2}}
\bigl(
e^{\ii \hat{\varphi}^{ACB}_\beta} \ket{01}
+ \ket{10}
\bigr)_{S_{\beta_R} S_{\beta_L}}
\otimes
\ket{\zeta_\beta^{a}}_{J_{\beta}},
\end{align}
with $ \hat{\varphi}^{ACB}_\alpha 
+ \hat{\varphi}^{ACB}_\gamma 
+ \hat{\varphi}^{ACB}_\beta 
= \phi_{ACB} = 0$.

It is crucial to note that the junk states appearing in the expression for the loop $ACB$ are the same as those appearing for the loop $ABC$, as a consequence of the locality condition.

We now turn to the two-party loops. The full self-testing result of \cref{thm:fullselftesting} does not apply directly in this case. Instead, we show that, as a consequence of the locality constraint, the stronger assumptions required for the $N=2$ analogue of \cref{thm:fullselftesting} are satisfied for all two-party loops involving party $A$; see \cref{thm:fullselftesting_N2}.

More precisely, as noted above, the states prepared by the parties are the same in all cases, since the network configuration is not known at the state-preparation stage. Moreover, by locality, the operators appearing in the two-party loops coincide with the corresponding operators appearing in the three-party loops. For instance, party $A$ implements the same operator in the loops $AC$ and $ABC$, while party $C$ implements the same operator in the loops $AC$ and $ACB$, because these configurations are locally indistinguishable from the perspective of the respective party; see \cref{fig:3loops,fig:2loopswithA}.

\begin{figure}
    \centering
    \includegraphics[width=0.6\linewidth]{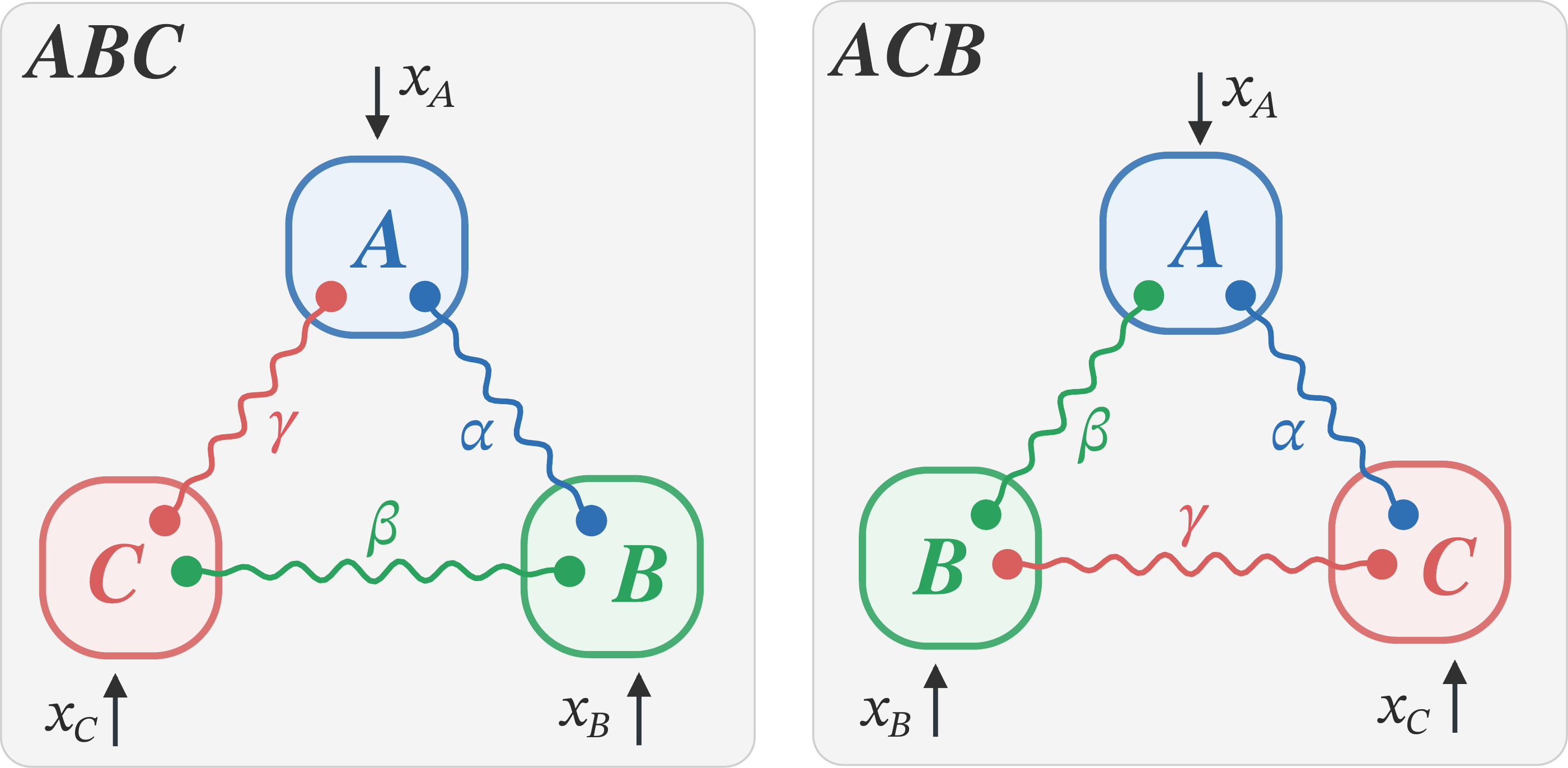}
    \caption{All the configurations of three-party loops.}
    \label{fig:3loops}
\end{figure}

\begin{figure}
    \centering
    \includegraphics[width=0.6\linewidth]{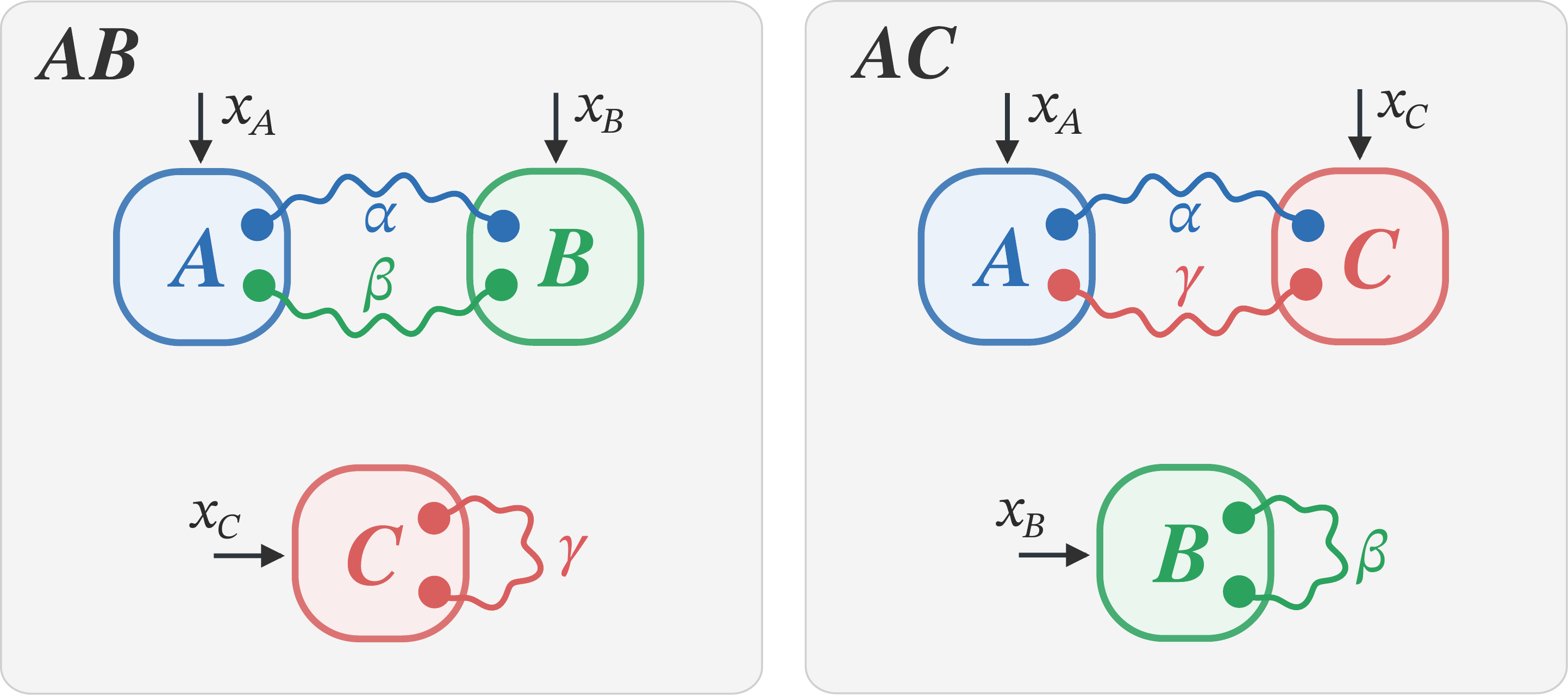}
    \caption{All the configurations of two-party loops including $A$.}
    \label{fig:2loopswithA}
\end{figure}
Thus, the structural assumptions of \cref{thm:fullselftesting_N2} are satisfied for all two-party loops. We further note that the assumptions on the trusted measurement coefficients are satisfied only for the two-party loops involving $A$, namely $AC$ and $AB$.

First consider the loop $AC$. The qubit strategy in this loop must reproduce the distribution $P_f^{AC}$. This distribution corresponds to the trusted qubit strategy described in Supplemental Material~\ref{appendix:ExperimentAndMainResult} where the phases satisfy $\varphi_\alpha + \varphi_\gamma = \pi$.
Therefore, for the $AC$ loop, we obtain
\begin{align}
\label{eq:loopAC_alpha}
U_{A_L}^{(\gamma)} U_{C_R}^{(\alpha)} \ket{\hat{\psi}_\alpha} &=  \frac{1}{\sqrt{2}}
\bigl(
e^{\ii \hat{\varphi}^{AC}_\alpha} \ket{01}
+ \ket{10}
\bigr)_{S_{\alpha_R} S_{\alpha_L}}
\otimes
\ket{\zeta_\alpha^{a}}_{J_{\alpha}} \\
\label{eq:loopAC_gamma}
U_{C_L}^{(\alpha)} U_{A_R}^{(\gamma)} \ket{\hat{\psi}_\gamma} &=  \frac{1}{\sqrt{2}}
\bigl(
e^{\ii \hat{\varphi}^{AC}_\gamma} \ket{01}
+ \ket{10}
\bigr)_{S_{\gamma_R} S_{\gamma_L}}
\otimes
\ket{\zeta_\gamma^{a}}_{J_{\gamma}},
\end{align}
where $\hat{\varphi}^{AC}_\alpha + \hat{\varphi}^{AC}_\gamma= \phi_{AC}=\pi$.

Similarly for loop $AB$, we obtain
\begin{align}
\label{eq:loopAB_alpha}
U_{A_L}^{(\beta)} U_{B_R}^{(\alpha)} \ket{\hat{\psi}_\alpha} &=  \frac{1}{\sqrt{2}}
\bigl(
e^{\ii \hat{\varphi}^{AB}_\alpha} \ket{01}
+ \ket{10}
\bigr)_{S_{\alpha_R} S_{\alpha_L}}
\otimes
\ket{\zeta_\alpha^{a}}_{J_{\alpha}} \\
\label{eq:loopAB_beta}
U_{B_L}^{(\alpha)} U_{A_R}^{(\beta)} \ket{\hat{\psi}_\beta} &=  \frac{1}{\sqrt{2}}
\bigl(
e^{\ii \hat{\varphi}^{AB}_\beta} \ket{01}
+ \ket{10}
\bigr)_{S_{\beta_R} S_{\beta_L}}
\otimes
\ket{\zeta_\beta^{a}}_{J_{\beta}},
\end{align}
where $\hat{\varphi}^{AB}_\alpha + \hat{\varphi}^{AB}_\beta= \phi_{AB}=\pi$.

The locality assumption ensures that the junk states and local unitaries appearing in these equations are the same across the different loop configurations.

Note that a similar argument to that used for loops $AC$ and $AB$ does not apply to loop $BC$, since the trusted measurement coefficients in this loop do not satisfy \cref{eq:meas_coeff_forA,eq:meas_coeff_forother}. Nevertheless, the distribution in this loop can be determined from the inferred properties of the measurement operators in the other loops, which coincide with those in loop $BC$ by the locality constraint. Indeed, we will show that the resulting distribution differs from the fermionic distribution in loop $BC$.

From \cref{eq:loopABC_alpha,eq:loopAB_alpha}, it follows that the actions of the operators $U_{A_L}^{(\gamma)}$ and $U_{A_L}^{(\beta)}$ on the state $\ket{\hat{\psi}_\alpha}$, when projected with $\bra{0}_{S_{\alpha_R}}$, coincide up to a phase factor:
\begin{align}
   \bra{0}_{S_{\alpha_R}} U_{A_L}^{(\gamma)}  \ket{\hat{\psi}_\alpha} &=  \frac{1}{\sqrt{2}}
e^{\ii \hat{\varphi}^{ABC}_\alpha}\big(U_{B_R}^{(\alpha)}\big)^\dagger \ket{1}_{ S_{\alpha_L}}
\otimes
\ket{\zeta_\alpha^{a}}_{J_{\alpha}},\\
\bra{0}_{S_{\alpha_R}} U_{A_L}^{(\beta)}  \ket{\hat{\psi}_\alpha} &=  \frac{1}{\sqrt{2}}
e^{\ii \hat{\varphi}^{AB}_\alpha}\big(U_{B_R}^{(\alpha)}\big)^\dagger \ket{1}_{ S_{\alpha_L}}
\otimes
\ket{\zeta_\alpha^{a}}_{J_{\alpha}}.
\end{align}
This motivates to define a phase $\phi_{A_L}$ as
\begin{align}
\label{eq:OperatorPhaseFactorRelation1}
\phi_{A_L} := \hat{\varphi}^{AB}_\alpha - \hat{\varphi}^{ABC}_\alpha.
\end{align}
Similarly, from \cref{eq:loopACB_beta,eq:loopAB_beta}, it follows that the actions of $U_{B_L}^{(\gamma)}$ and $U_{B_L}^{(\alpha)}$ on $\ket{\hat{\psi}_\beta}$, when projected with $\bra{0}_{S_{\beta_R}}$, coincide up to the phase factor
\begin{align}
\label{eq:OperatorPhaseFactorRelation2}
\phi_{B_L} := \hat{\varphi}^{ACB}_\beta - \hat{\varphi}^{AB}_\beta.
\end{align}
The same equations also give
\begin{align}
\label{eq:usedinloopBC1}
\bra{1}_{S_{\beta_R}} U_{B_L}^{(\gamma)}\ket{\hat{\psi}_\beta}= \bra{1}_{S_{\beta_R}} U_{B_L}^{(\alpha)}\ket{\hat{\psi}_\beta}.
\end{align}
Moreover, from \cref{eq:loopABC_gamma,eq:loopAC_gamma}, it follows that the actions of $U_{C_L}^{(\alpha)}$ and $U_{C_L}^{(\beta)}$ on $\ket{\hat{\psi}_\gamma}$, when projected onto $\bra{0}_{S_{\gamma_R}}$, coincide up to the phase factor
\begin{equation}
\label{eq:OperatorPhaseFactorRelation3}
\phi_{C_L} := \hat{\varphi}^{AC}_\gamma - \hat{\varphi}^{ABC}_\gamma.
\end{equation}
Considering \cref{eq:loopABC_beta,eq:loopAB_beta}, it follows that the actions of $U_{A_R}^{(\beta)}$ and $U_{C_R}^{(\beta)}$ on $\ket{\hat{\psi}_\beta}$, when projected onto $\bra{1}_{S_{\beta_L}}$, coincide up to the phase factor
\begin{equation}
\label{eq:OperatorPhaseFactorRelation4}
\phi_{A_R} := \hat{\varphi}^{AB}_\beta - \hat{\varphi}^{ABC}_\beta.
\end{equation}
Similarly, from \cref{eq:loopACB_gamma,eq:loopAC_gamma}, it follows that the actions of $U_{B_R}^{(\gamma)}$ and $U_{A_R}^{(\gamma)}$ on $\ket{\hat{\psi}_\gamma}$, when projected onto $\bra{1}_{S_{\gamma_L}}$, coincide up to the phase factor
\begin{equation}
\label{eq:OperatorPhaseFactorRelation5}
\phi_{B_R} := \hat{\varphi}^{ACB}_\gamma -\hat{\varphi}^{AC}_\gamma .
\end{equation}
These equations also yield
\begin{align}
\label{eq:usedinloopBC2}
\bra{0}_{S_{\gamma_L}} U_{B_R}^{(\gamma)}\ket{\hat{\psi}_\gamma}= \bra{0}_{S_{\gamma_L}} U_{A_R}^{(\gamma)}\ket{\hat{\psi}_\gamma}.
\end{align}
Finally, from \cref{eq:loopABC_alpha,eq:loopAC_alpha}, it follows that the actions of $U_{C_R}^{(\alpha)}$ and $U_{B_R}^{(\alpha)}$ on $\ket{\hat{\psi}_\alpha}$, when projected onto $\bra{1}_{S_{\alpha_L}}$, coincide up to the phase factor
\begin{equation}
\label{eq:OperatorPhaseFactorRelation6}
\phi_{C_R} := \hat{\varphi}^{AC}_\alpha - \hat{\varphi}^{ABC}_\alpha.
\end{equation}
We note that, using \cref{eq:loopACB_alpha,eq:loopAB_alpha} and repeating the same argument, we realize that this phase factor also satisfies
\begin{equation}
\label{eq:OperatorPhaseFactorRelation7}
\phi_{C_R} = \hat{\varphi}^{AB}_\alpha - \hat{\varphi}^{ACB}_\alpha.
\end{equation}

We now rewrite the phase constraints in each loop using the relations derived above.

For the loop $ABC$, the phase constraint reads
\begin{align}
\hat{\varphi}^{ABC}_\alpha + \hat{\varphi}^{ABC}_\beta + \hat{\varphi}^{ABC}_\gamma
= \phi_{ABC} = 0.
\end{align}
For the loop $AC$, the phase condition $\hat{\varphi}^{AC}_\alpha + \hat{\varphi}^{AC}_\gamma = \phi_{AC} = \pi$, together with \cref{eq:OperatorPhaseFactorRelation3,eq:OperatorPhaseFactorRelation6}, yields
\begin{align}
\hat{\varphi}^{ABC}_\alpha + \hat{\varphi}^{ABC}_\gamma + \phi_{C_R} + \phi_{C_L}
= \pi.
\end{align}
For the loop $AB$, the phase condition $\hat{\varphi}^{AB}_\alpha + \hat{\varphi}^{AB}_\beta = \phi_{AB} = \pi$, together with \cref{eq:OperatorPhaseFactorRelation1,eq:OperatorPhaseFactorRelation4}, gives
\begin{align}
\hat{\varphi}^{ABC}_\alpha + \hat{\varphi}^{ABC}_\beta + \phi_{A_R} + \phi_{A_L}
= \pi.
\end{align}
For the loop $ACB$, we express each phase using the relations derived above. Then, \cref{eq:OperatorPhaseFactorRelation2,eq:OperatorPhaseFactorRelation4} imply
\[
\hat{\varphi}^{ACB}_\beta = \hat{\varphi}^{ABC}_\beta + \phi_{B_L} + \phi_{A_R},
\]
and \cref{eq:OperatorPhaseFactorRelation3,eq:OperatorPhaseFactorRelation5} yield
\[
\hat{\varphi}^{ACB}_\gamma = \hat{\varphi}^{ABC}_\gamma + \phi_{C_L} + \phi_{B_R}.
\]
Similarly, using \cref{eq:OperatorPhaseFactorRelation7,eq:OperatorPhaseFactorRelation1}, we obtain
\[
\hat{\varphi}^{ACB}_\alpha = \hat{\varphi}^{ABC}_\alpha + \phi_{C_R} + \phi_{A_L}.
\]
Substituting these expressions into the phase constraint $\hat{\varphi}^{ACB}_\alpha + \hat{\varphi}^{ACB}_\beta + \hat{\varphi}^{ACB}_\gamma = \phi_{ACB} = 0$, we obtain
\begin{align}
\hat{\varphi}^{ABC}_\alpha + \hat{\varphi}^{ABC}_\beta + \hat{\varphi}^{ABC}_\gamma
+ \phi_{A_R} + \phi_{A_L} + \phi_{B_R} + \phi_{B_L} + \phi_{C_R} + \phi_{C_L}
= 0.
\end{align}

\begin{figure}
    \centering
    \includegraphics[width=0.3\linewidth]{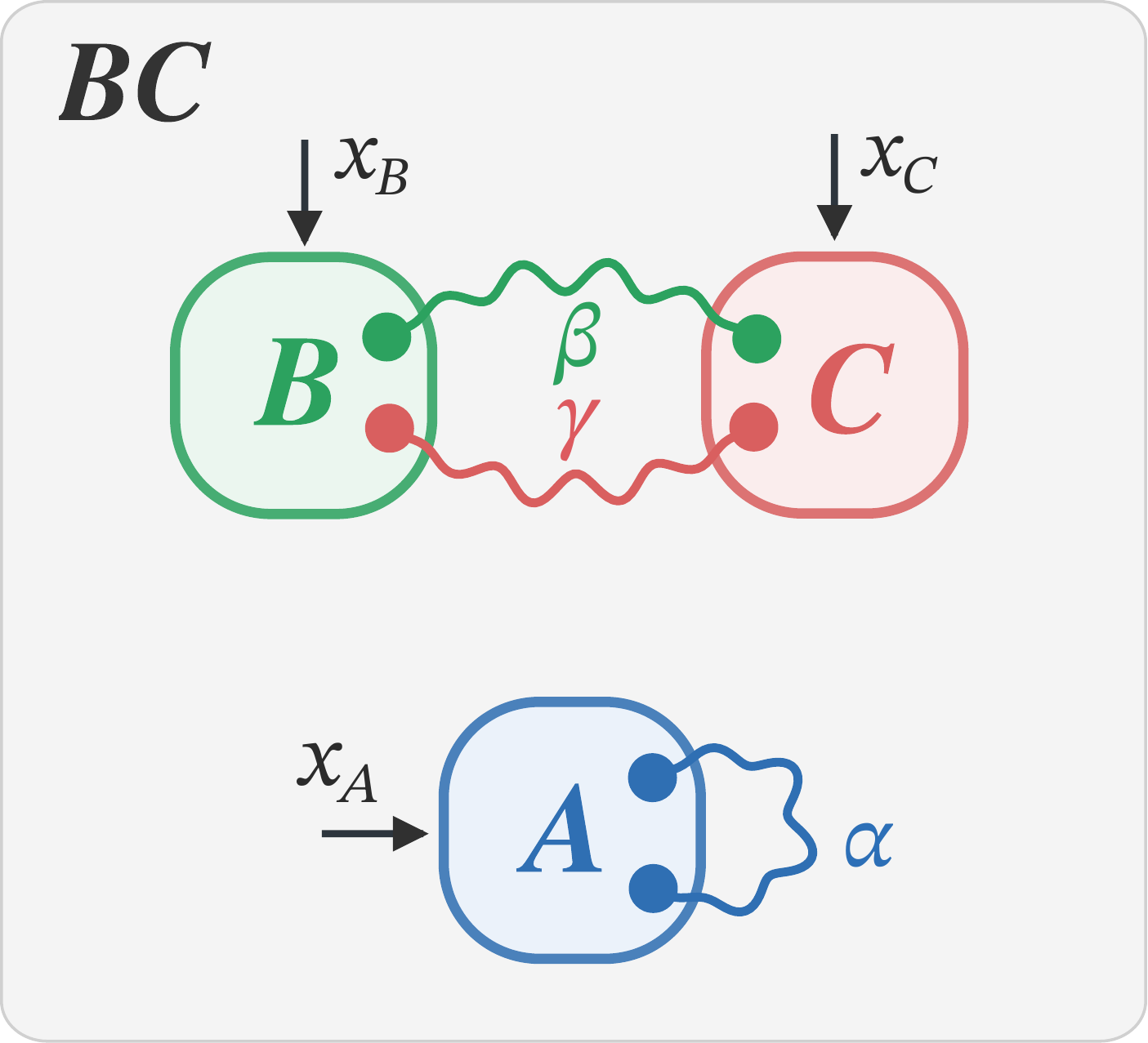}
    \caption{Loop $BC$.}
    \label{fig:bcloop}
\end{figure}
We now derive the corresponding phase constraint for the loop $BC$ using the relations obtained above. By adding the phase constraints of the three-party loops and subtracting those of the two-party loops, we obtain
\begin{align}
\label{eq:AccumulatedPhaseInBC}
\hat{\varphi}^{ABC}_\beta + \hat{\varphi}^{ABC}_\gamma + \phi_{B_R} + \phi_{B_L} = 0.
\end{align}
This expression corresponds to the total accumulated phase in the loop $BC$, as we show below. 

As illustrated in \cref{fig:bcloop}, the operators associated with this loop are $U_{B_L}^{(\gamma)} U_{B_R}^{(\gamma)}$ and $U_{C_R}^{(\beta)} U_{C_L}^{(\beta)}$. From \cref{eq:loopABC_beta,eq:OperatorPhaseFactorRelation2,eq:usedinloopBC1}, we obtain
\begin{align}
U_{B_L}^{(\gamma)} U_{C_R}^{(\beta)} \ket{\hat{\psi}_\beta}
&= \frac{1}{\sqrt{2}}
\bigl(
e^{\ii (\hat{\varphi}^{ABC}_\beta + \phi_{B_L})} \ket{01}
+ \ket{10}
\bigr)_{S_{\beta_R} S_{\beta_L}}
\otimes
\ket{\zeta_\beta^{a}}_{J_{\beta}}.
\end{align}
Similarly, from \cref{eq:loopABC_gamma,eq:OperatorPhaseFactorRelation5,eq:usedinloopBC2}, we have
\begin{align}
U_{C_L}^{(\beta)} U_{B_R}^{(\gamma)} \ket{\hat{\psi}_\gamma}
&= \frac{1}{\sqrt{2}}
\bigl(
e^{\ii (\hat{\varphi}^{ABC}_\gamma + \phi_{B_R})} \ket{01}
+ \ket{10}
\bigr)_{S_{\gamma_R} S_{\gamma_L}}
\otimes
\ket{\zeta_\gamma^{a}}_{J_{\gamma}}.
\end{align}
Therefore, the total accumulated phase in the loop $BC$ is
\begin{equation}
\phi_{BC} := \hat{\varphi}^{ABC}_\beta + \hat{\varphi}^{ABC}_\gamma + \phi_{B_R} + \phi_{B_L} = 0.
\end{equation}

We are now ready to compute the conditional probability distribution in the loop $BC$. The conditional state in this loop, after the application of the local unitaries, reads
\begin{equation}
\ket{\tilde{\psi}}^{\mathrm{cond}}_{BC} = \frac{1}{\sqrt{2}}
\bigl(
e^{\ii \phi_{BC}} \ket{01,01} + \ket{10,10}
\bigr)_{S_{\beta_R} S_{\beta_L} S_{\gamma_R} S_{\gamma_L}}
\otimes \ket{\zeta_\beta^{a}}_{J_{\beta}} \otimes \ket{\zeta_\gamma^{a}}_{J_{\gamma}},
\end{equation}
where $\phi_{BC} = 0$.

Rewriting the system part in terms of the party labeling, we obtain
\begin{equation}
\label{eq:BCState}
\ket{\tilde{\psi}}^{\mathrm{cond}}_{BC} = \frac{1}{\sqrt{2}}
\bigl(
\ket{01,01} + \ket{10,10}
\bigr)_{S_{B_L} S_{C_R} S_{C_L} S_{B_R}}
\otimes \ket{\zeta_\beta^{a}}_{J_{\beta}} \otimes \ket{\zeta_\gamma^{a}}_{J_{\gamma}}.
\end{equation}

We now recall the conditional PVM elements of the measurement in this loop.
\begin{equation}
\begin{aligned}
\ket{\hat{m}^{s_i}_{i, x_i=1}}\bra{\hat{m}^{s_i}_{i, x_i=1}} &= \frac{\id + s_i \mathcal{X}_i}{2}, \\
\ket{\hat{m}^{s_i}_{i, x_i=0}}\bra{\hat{m}^{s_i}_{i, x_i=0}} &= \frac{\id + s_i \mathcal{Z}_i}{2}.
\end{aligned}
\end{equation}

Since the state in this loop is self-tested (see \cref{eq:BCState}), we can determine the corresponding self-tested measurements and compute the resulting distribution. According to \cref{thm:fullselftesting}, the measurement operators satisfy
\begin{equation}
\begin{aligned}
U_i \mathcal{Z}_i U_i^\dagger &= (Z)_{S_i} \otimes \id_{J_i}, \\
U_i \mathcal{X}_i U_i^\dagger &= (X)_{S_i} \otimes \id_{J_i},
\end{aligned}
\end{equation}
where
\begin{equation*}
\begin{aligned}
(Z)_{S_i} &= (\ketbra{01}{01} - \ketbra{10}{10})_{S_{i_R} S_{i_L}}, \\
(X)_{S_i} &= (\ketbra{01}{10} + \ketbra{10}{01})_{S_{i_R} S_{i_L}}.
\end{aligned}
\end{equation*}
Therefore, the self-tested POVM elements are
\begin{equation}
\begin{aligned}
\ket{\tilde{m}^{s_i}_{i, x_i=1}}\bra{\tilde{m}^{s_i}_{i, x_i=1}} &= \left( \frac{\id + s_i X_i}{2} \right)_{S_i} \otimes \id_{J_i}, \\
\ket{\tilde{m}^{s_i}_{i, x_i=0}}\bra{\tilde{m}^{s_i}_{i, x_i=0}} &= \left( \frac{\id + s_i Z_i}{2} \right)_{S_i} \otimes \id_{J_i}.
\end{aligned}
\end{equation}

We now consider the following conditional probability in the loop $BC$:
\[
P^{BC}(a_B, a_C \mid (a_B, a_C) \in \{1_+, 1_- \}^2 , x_B, x_C).
\]
For $x_B = x_C = 1$, this is given by
\begin{equation}
\begin{aligned}
P^{BC}(a_B, a_C \mid (a_B, a_C) \in \{1_+, 1_- \}^2 , x_B = x_C = 1)
&= \frac{1}{4} \bra{\psi}^{\mathrm{cond}}_{BC}
(\id + s_B \mathcal{X}_B)(\id + s_C \mathcal{X}_C)
\ket{\psi}^{\mathrm{cond}}_{BC} \\
&= \frac{1}{4} \bra{\psi}
(\id + s_B X_B)(\id + s_C X_C)
\ket{\psi},
\end{aligned}
\end{equation}
where
\[
\ket{\psi} = \frac{1}{\sqrt{2}}
\bigl(
\ket{01,01} + \ket{10,10}
\bigr)_{S_{B_L} S_{C_R} S_{C_L} S_{B_R}}.
\]
A direct computation yields
\begin{equation}
P^{BC}(a_B, a_C \mid (a_B, a_C) \in \{1_+, 1_- \}^2 , x_B = x_C = 1)
= \frac{1}{4} \bigl(1 + s_B s_C\bigr).
\end{equation}
On the other hand, from \cref{eq:Dist_2}, the fermionic strategy predicts
\begin{equation}
\begin{aligned}
P_f^{BC}(a_B, a_C \mid (a_B, a_C) \in \{1_+, 1_- \}^2 , x_B = x_C = 1)
&= \frac{1}{2} \bigl| u_{B,1}^{s_B} u_{C,1}^{s_C} - v_{B,1}^{s_B} v_{C,1}^{s_C} \bigr|^2 \\
&= \frac{1}{2} \left( \frac{1}{2} - \frac{s_B s_C}{2} \right)^2 \\
&= \frac{1}{4} \bigl(1 - s_B s_C\bigr),
\end{aligned}
\end{equation}
where we used $u_{i,1}^{s_i} = \frac{1}{\sqrt{2}}$ and $v_{i,1}^{s_i} = s_i \frac{1}{\sqrt{2}}$ for $i \in \{B,C\}$.
For each outcome \((a_B,a_C)\), the individual terms differ by
\begin{align}
&\left|
P_f^{BC}
\left(
a_B,a_C
\,\middle|\,
(a_B, a_C) \in \{1_+, 1_- \}^2 ,\,
x_B=x_C=1
\right)
\right. \nonumber \\
&\left.
\hspace{1.5cm}
-
P^{BC}
\left(
a_B,a_C
\,\middle|\,
(a_B, a_C) \in \{1_+, 1_- \}^2 ,\,
x_B=x_C=1
\right)
\right|
=
\frac{1}{2}.
\end{align}
Hence, for the total variation distance between the two conditional
distributions, we obtain
\begin{align}
d^{\rm cond}_{\rm TV}
\left(
P_f^{BC},P^{BC}
\right) = \dfrac{1}{2}&\sum_{\substack{a_B,\,a_C \\  x_B,x_C}}
\left|
P_f^{BC}
\left(
a_B,a_C
\,\middle|\,
(a_B, a_C) \in \{1_+, 1_- \}^2 ,\,
x_B , x_C
\right)
\right. \nonumber \\
&\left.
\hspace{1.5cm}
-
P^{BC}
\left(
a_B,a_C
\,\middle|\,
(a_B, a_C) \in \{1_+, 1_- \}^2 ,\,
x_B, x_C
\right)
\right|
=
1.
\end{align}

Therefore, the predicted distributions differ, contradicting the assumption that a local Q-QIT strategy reproduces all fermionic loop distributions. This proves the claim.

Note that the discrepancy between the distributions $P_f^{BC}$ and $P^{BC}$ cannot simply be resolved by relabeling the conditional outcomes $a_B,a_C$, as this would change the correlations in the remaining loops. In fact, this observation is consistent with the fact that naive simulation strategies where a local ``-1" phase is introduced by one of the parties are impossible. For a relabeling strategy of the outcomes $a_B,a_C$ to work, the parties would need to know in which loops $l$ the relabeling is performed and in which loops it is not. This is exactly the information that is not accessible to the parties, giving rise to a contradiction. 

Finally, we note that the coarse-grained distributions in the fermionic and qubit cases coincide for this loop. Hence, the difference between the full distributions is captured by
\begin{equation}
\begin{aligned}
d_{\rm TV}\!\left(P_f^{BC},P^{BC}\right)
&=
\frac{1}{2}
\sum_{\substack{a_B,\,a_C \\ x_B,x_C}}
\left|
P_f^{BC}\!\left(a_B,a_C \mid x_B,x_C\right)
-
P^{BC}\!\left(a_B,a_C \mid x_B,x_C\right)
\right| =
\frac{1}{2}.
\end{aligned}
\end{equation}

\section{Bosonic simulation using pre-shared entanglement}\label{appendix:BK-simulation}

While we prove that the thought experiment described above cannot be reproduced by a fully local bosonic strategy, we remark that it could be reproduced if we allowed the bosonic parties to initially share a fixed entangled global state. This highlights the crucial nature for our result of the assumption that the bosonic parties, just like the fermionic ones, begin the protocol with a locally preparable product state. A bosonic system with free access to a suitable source of globally entangled states could not in fact be operationally distinguished from a fermionic system.

For completeness, we discuss here a possible construction of how our thought experiment could be simulated by bosons with access to a pre-shared entangled state. The idea is to apply a fermion-to-qubit mapping that preserves locality of operators (though not of states). There exist many variants of mappings of this type, mostly inspired by the superfast encodings of Bravyi and Kitaev~\cite{bravyi_fermionic_2002}: our construction specifically is based on the \emph{generalised superfast encodings} of Ref.~\cite{setia_superfast_2019}.

As discussed in Supplemental Material~\ref{app:experimental-implementation}, in order to implement the thought experiment a fermionic system needs to be able to apply gates of the form $U=e^{-i\alpha H_{jk}}$ on pairs of modes $(j,k)$, with $H_{jk}=i(f_{j}^\dagger f_{k} -f_{k}^\dagger f_{j})$, and to perform measurements of the occupation number of single modes. Our goal here is therefore to show that both the Hamiltonian $H_{jk}$ and the occupation number operator $n_{k}=f^\dagger_{k} f_{k} $ can be implemented as local operators on a suitable qubit system, once an initial entangled state is shared. To do this we follow Refs.~\cite{bravyi_fermionic_2002, setia_superfast_2019} and introduce for every pair of modes $(j,k)$ the Hermitian operators
\begin{align}
    A_{jk}:=-i(f^\dagger_j+f_j)(f^\dagger_k+f_k)\,, \quad  \quad B_k:=1-2f_k^\dagger f_k \,. \label{eq:fermionic_AB}
\end{align}
After some algebra, one can show that the Hamiltonian $H_{jk}$ and the number operator $n_k$ can be expressed in terms of $A_{jk}$ and $B_k$ as 
\begin{align}
    H_{jk}=\frac{1}{2}A_{jk}(B_jB_k-1)\,, \quad  \quad n_k=\frac{1-B_k}{2} \,. \label{eq:H_and_n}
\end{align}
If follows that we can equivalently focus on defining a local qubit implementation for $A_{jk}$ and $B_k$, from which we will be then able to construct the one of $H_{jk}$ and $n_k$.

Let us now consider in more detail for which pairs of modes $(j,k)$ we need to implement these operators $A_{jk}$. Firstly, each party needs to apply a gate on the pair $(\theta_R,\theta_L)$ with $\theta\in\{\alpha,\beta,\gamma\}$ to prepare the initial state (Step 2. in Supplemental Material~\ref{app:experimental-implementation}).  Then, after the modes are rearranged by the referee, each party will hold their initial mode $\theta_R$ and receive a further mode $\tilde\theta_L$ from one of the other parties. Each party will need to apply a gate on this pair of modes $(\tilde\theta_L,\theta_R)$ that it is now holding (Step 4. in Supplemental Material~\ref{app:experimental-implementation}). For example, party $A$ keeps $\alpha_R$ and could receive either $\beta_L$ or $\gamma_L$ and would then need to perform a gate on $(\beta_L,\alpha_R)$ or $(\gamma_L,\alpha_R)$ (in the case that it is excluded from the loop, it will again perform a gate on $(\alpha_R,\alpha_L)$ as in the previous step). Each of these gates needs to be implemented locally by only acting on the two subsystems $(\tilde\theta_L,\theta_R)$ that the party is holding at the moment the gate is applied. In summary, all the pairs of modes that, at any point in any round of the experiment, may need to be acted jointly upon are represented by edges in the graph $G$ shown in Figure~\ref{BK-graph}.

\begin{figure}[ht]
    \centering
    \includegraphics[width=0.4\linewidth]{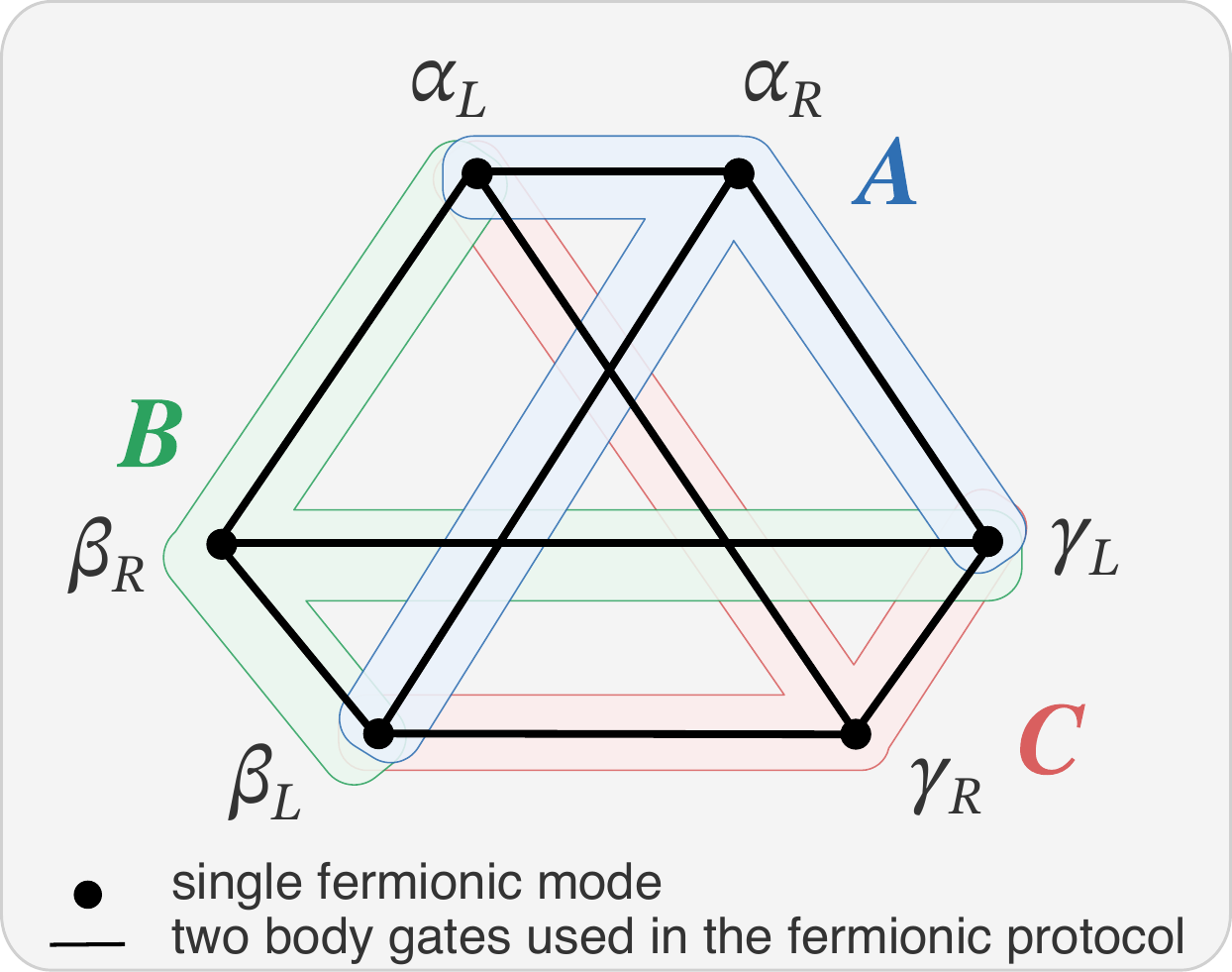}
    \caption{The black lines representat the graph $G$: each vertex represents a fermionic mode in the experiment, two vertices are connected if the two modes are jointly acted on by a local party at some point in some round of the experiment. The coloured shadings reflect which party acts on each such pair of modes.}
    \label{BK-graph}
\end{figure}

We are now ready to introduce a qubit construction that reproduces these local fermionic operators. For this, the bosonic parties will represent each fermionic mode by a system of two qubits. We will label $\alpha_R$, $\alpha_R'$ the qubits representing the fermionic mode $\alpha_R$, and analogously for the remaining modes. We will further denote by $X_k$, $Y_k$, $Z_k$ the standard Pauli operators acting on qubit $k$. We now introduce a set of qubit operators $\hat{A}_{jk}$ and $\hat{B}_{k}$ for every edge $(j,k)$ and vertex $k$ in the graph $G$. These are intended to be qubit representations of the corresponding fermionic operators $A_{jk}$ and $B_{k}$ from Eq.~\eqref{eq:fermionic_AB} and are defined as follows:
\begin{align}
    \hat{A}_{\alpha_R,\alpha_L}&=X_{\alpha_R}X_{\alpha_L}           &\qquad  \hat B_{\alpha_R}&=Z_{\alpha_R} Z_{\alpha_R'},\nonumber\\
    \hat{A}_{\beta_R,\beta_L}&=X_{\beta_R}X_{\beta_L}           &\qquad \hat B_{\alpha_L}&=Z_{\alpha_L} Z_{\alpha_L'},\nonumber\\
    \hat{A}_{\gamma_R,\gamma_L}&=X_{\gamma_R}X_{\gamma_L}           &\qquad  \hat B_{\beta_R}&=Z_{\beta_R} Z_{\beta_R'},\nonumber\\
    \hat{A}_{\alpha_L,\beta_R}&=X_{\alpha_L'}Z_{\alpha_L} X_{\beta_R'}Z_{\beta_R} &\qquad  \hat B_{\beta_L}&=Z_{\beta_L} Z_{\beta_L'},\nonumber\\
    \hat{A}_{\beta_L,\gamma_R}&=X_{\beta_L'}Z_{\beta_L} X_{\gamma_R'}Z_{\gamma_R} &\qquad  \hat B_{\gamma_R}&=Z_{\gamma_R} Z_{\gamma_R'},\nonumber\\
    \hat{A}_{\gamma_L,\alpha_R}&=X_{\gamma_L'}Z_{\gamma_L} X_{\alpha_R'}Z_{\alpha_R} &\qquad  \hat B_{\gamma_L}&=Z_{\gamma_L} Z_{\gamma_L'},\nonumber\\
    \hat{A}_{\alpha_L,\gamma_R}&=Y_{\alpha_L'}Z_{\alpha_L} Y_{\gamma_R'}Z_{\gamma_R}, \nonumber\\
    \hat{A}_{\beta_L,\alpha_R}&=Y_{\beta_L'}Z_{\beta_L} Y_{\alpha_R'}Z_{\alpha_R}, \nonumber\\
    \hat{A}_{\gamma_L,\beta_R}&=Y_{\gamma_L'}Z_{\gamma_L} Y_{\beta_R'}Z_{\beta_R}. \label{eq:encoded_AB}
\end{align}
Firstly, notice that each of these operators only acts on the qubits pertaining to the subsystems that the corresponding fermionic operator is supposed to act on. For example, when the party $A$ needs to implement the operator $\hat{A}_{\beta_L,\alpha_R}$ on its subsystem $\alpha_R$ and on the subsystem $\beta_L$ it has received from $B$, then it can do it just by acting on the qubits representing precisely the subsystems $\alpha_R$ and $\beta_L$ which it holds. In other words, these operators satisfy the locality conditions of the experiment. 
One subtlety here is that the parties need to know from whom they received the subsystem, in order to know how to act. For instance, party $A$ needs to apply $\hat{A}_{\beta_L,\alpha_R}$ when it receives a system from $B$ and $\hat{A}_{\gamma_L,\alpha_R}$ when it receives a system from $C$: as $\hat{A}_{\beta_L,\alpha_R}$ and $\hat{A}_{\gamma_L,\alpha_R}$ have a different form it needs to know in which case it finds itself. However, it is not hard to convey this information without violating the protocol's assumptions: each party could for example append a further register of qubits to their transferred subsystem which encodes the identity of the sender. The receiving parties would need to apply the corresponding $\hat{A}_{jk}$ conditioned on the value of this extra register.

Second, following Ref.~\cite{setia_superfast_2019}, we can observe that these operators $\hat{A}_{jk}$ and $\hat{B}_{k}$ obey exactly the same mutual (anti-)commutation relations as the corresponding fermionic operators $A_{jk}$ and $B_{k}$. However, for $\hat{A}_{jk}$ and $\hat{B}_{k}$ to exactly reproduce the fermionic operator algebra they also need to fulfill some further conditions. Indeed, the fermionic operators $A_{jk}$ are not all algebraically independent and the operators $\hat{A}_{jk}$ should satisfy the same mutual dependencies. These dependencies can be reduced to the following conditions~\cite{bravyi_fermionic_2002, setia_superfast_2019}:
\begin{equation}
    i^sA_{k_1,k_2} A_{k_2,k_3}\cdots A_{k_s,k_1} = \id,
\end{equation}
for all cycles in the graph G (\textit{i.e.} for all closed sequences of edges $(k_1,k_2), (k_2,k_3), \dots (k_s,k_1)$). It is clear that the operators defined above in~\eqref{eq:encoded_AB} do not satisfy these relations exactly (and it would not be possible to achieve this for the given graph, see Ref.~\cite{Guaita2025localityofqubit}). However, in the spirit of the superfast encodings, we can define a subspace $\mathcal{C}$ of states in the global qubit Hilbert space for which we have $i^sA_{k_1,k_2} A_{k_2,k_3}\cdots A_{k_s,k_1}\ket{\psi} =\ket{\psi}$ for all $\ket{\psi}\in\mathcal{C}$ and all cycles in the graph $G$. That is, within this subspace the operators $\hat{A}_{jk}$ and $\hat{B}_{k}$ act exactly like the fermionic operator algebra. The space $\mathcal{C}$ represents the encoded fermionic states. 

The final step that we need to complete the construction is to find the state within $\mathcal{C}$ that corresponds to the initial state of the fermionic protocol. According to Step 1. in Supplemental Material~\ref{app:experimental-implementation} this is the fermionic state $\fket{010101}_{\alpha_R\alpha_L\beta_R\beta_L\gamma_R\gamma_L}$. The corresponding qubit state $\ket{\psi_{\rm init}}\in\mathcal{C}$ is identified by two sets of conditions. First, it needs to satisfy the defining conditions of the subspace $\mathcal{C}$, namely $i^sA_{k_1,k_2} A_{k_2,k_3}\cdots A_{k_s,k_1}\ket{\psi_{\rm init}} =\ket{\psi_{\rm init}}$ for all cycles in $G$. Second, it needs to satisfy the defining conditions of the Fock state $\fket{010101}_{\alpha_R\alpha_L\beta_R\beta_L\gamma_R\gamma_L}$, that is $n_{\alpha_R}\ket{\psi_{\rm init}}=n_{\beta_R}\ket{\psi_{\rm init}}=n_{\gamma_R}\ket{\psi_{\rm init}}=0$ and $n_{\alpha_L}\ket{\psi_{\rm init}}=n_{\beta_L}\ket{\psi_{\rm init}}=n_{\gamma_L}\ket{\psi_{\rm init}}=\ket{\psi_{\rm init}}$. Noting that there are four independent cycles in the graph $G$ and applying relation~\eqref{eq:H_and_n}, these conditions ultimately are
\begin{align}
    \hat B_{\alpha_R}\ket{\psi_{\rm init}}&=\ket{\psi_{\rm init}},   &\quad \hat{A}_{\alpha_R,\alpha_L}\hat{A}_{\alpha_L,\beta_R}\hat{A}_{\beta_R,\beta_L}\hat{A}_{\beta_L,\alpha_R}\ket{\psi_{\rm init}}&=\ket{\psi_{\rm init}}, \nonumber \\
    \hat B_{\beta_R}\ket{\psi_{\rm init}}&=\ket{\psi_{\rm init}},   &\quad\hat{A}_{\beta_R,\beta_L}\hat{A}_{\beta_L,\gamma_R}\hat{A}_{\gamma_R,\gamma_L}\hat{A}_{\gamma_L,\beta_R}\ket{\psi_{\rm init}}&=\ket{\psi_{\rm init}},\nonumber\\
    \hat B_{\gamma_R}\ket{\psi_{\rm init}}&=\ket{\psi_{\rm init}}, &\quad \hat{A}_{\gamma_R,\gamma_L}\hat{A}_{\gamma_L,\alpha_R}\hat{A}_{\alpha_R,\alpha_L}\hat{A}_{\alpha_L,\gamma_R}\ket{\psi_{\rm init}}&=\ket{\psi_{\rm init}},\nonumber\\
    -\hat B_{\alpha_L}\ket{\psi_{\rm init}}&=\ket{\psi_{\rm init}}, &\quad  -\hat{A}_{\alpha_R,\alpha_L}\hat{A}_{\alpha_L,\beta_R}  \hat{A}_{\beta_R,\beta_L}\hat{A}_{\beta_L,\gamma_R}\hat{A}_{\gamma_R,\gamma_L}\hat{A}_{\gamma_L,\alpha_R}\ket{\psi_{\rm init}}&=\ket{\psi_{\rm init}}, \nonumber\\
    -\hat B_{\beta_L}\ket{\psi_{\rm init}}&=\ket{\psi_{\rm init}},  &\quad &\nonumber\\
     -\hat B_{\gamma_L}\ket{\psi_{\rm init}}&=\ket{\psi_{\rm init}}.  &\quad &
\end{align}
Substituting the expressions for $\hat{A}_{jk}$ and $\hat{B}_{k}$ from~\eqref{eq:encoded_AB}, we finally have that $\ket{\psi_{\rm init}}$ is defined by the following ten independent conditions:
\begin{align}
    Z_{\alpha_R}Z_{\alpha_R'}\ket{\psi_{\rm init}}&=\ket{\psi_{\rm init}},   &\quad -Y_{\alpha_R}Y_{\alpha_R'}Y_{\alpha_L}X_{\alpha_L'}Y_{\beta_R}X_{\beta_R'}Y_{\beta_L}Y_{\beta_L'}\ket{\psi_{\rm init}}&=\ket{\psi_{\rm init}}, \nonumber \\
    Z_{\beta_R}Z_{\beta_R'}\ket{\psi_{\rm init}}&=\ket{\psi_{\rm init}},   &\quad -Y_{\beta_R}Y_{\beta_R'}Y_{\beta_L}X_{\beta_L'}Y_{\gamma_R}X_{\gamma_R'}Y_{\gamma_L}Y_{\gamma_L'}\ket{\psi_{\rm init}}&=\ket{\psi_{\rm init}},\nonumber\\
    Z_{\gamma_R}Z_{\gamma_R'}\ket{\psi_{\rm init}}&=\ket{\psi_{\rm init}}, &\quad -Y_{\gamma_R}Y_{\gamma_R'}Y_{\gamma_L}X_{\gamma_L'}Y_{\alpha_R}X_{\alpha_R'}Y_{\alpha_L}Y_{\alpha_L'}\ket{\psi_{\rm init}}&=\ket{\psi_{\rm init}},\nonumber\\
    -Z_{\alpha_L}Z_{\alpha_L'}\ket{\psi_{\rm init}}&=\ket{\psi_{\rm init}}, &\quad  Y_{\alpha_R}X_{\alpha_R'}Y_{\alpha_L}X_{\alpha_L'}Y_{\beta_R}X_{\beta_R'}Y_{\beta_L}X_{\beta_L'}Y_{\gamma_R}X_{\gamma_R'}Y_{\gamma_L}X_{\gamma_L'}\ket{\psi_{\rm init}}&=\ket{\psi_{\rm init}}, \nonumber\\
    -Z_{\beta_L}Z_{\beta_L'}\ket{\psi_{\rm init}}&=\ket{\psi_{\rm init}},  &\quad &\nonumber\\
     -Z_{\gamma_L}Z_{\gamma_L'}\ket{\psi_{\rm init}}&=\ket{\psi_{\rm init}}.  &\quad &
\end{align}

Recalling that the total qubit system is composed of 12 qubits, these conditions define a four-dimensional stablilizer code space. Any state from this subspace is a valid choice for $\ket{\psi_{\rm init}}$. Sharing this state initially between the parties allows them to implement the rest of the protocol locally by using the operators defined above. Note that any such valid $\ket{\psi_{\rm init}}$ must contain entanglement between the parties, to avoid a contradiction with the main result of this paper (see also the analysis in Ref.~\cite{Guaita2025localityofqubit}).

\section{\texorpdfstring{Generalization to $N$}{Generalization to N}}
\label{appendix:NPartyGeneralization}
In this section, we explain how the three-party game described in
Supplemental Material~\ref{appendix:ExperimentAndMainResult} can be generalized to \(N\geq 3\)
parties. The construction is obtained by considering loops involving all
\(N\) parties, and also loops involving \(N-1\) parties.

Let the parties be denoted by \(A_1,\dots,A_N\). Each party \(A_i\) locally
prepares a bipartite state and sends one half of it to the referee, as in
Supplemental Material~\ref{appendix:ExperimentAndMainResult}. The referee then redistributes the received subsystems according to a chosen loop configuration. More precisely,
the referee considers all directed loops containing all \(N\) parties, with all
possible cyclic orderings, and all directed loops containing \(N-1\) parties,
again with all possible cyclic orderings. In the latter case, the party that is
not included in the loop receives back the subsystem that it originally sent to
the referee.

After the redistribution, the rest of the protocol is the same for all loop
configurations. Each party
\(A_i\) receives an input bit \(x_i\) and produces an output $a_i \in \{0, 1_-, 1_+, 2 \}$.

The fermionic strategy is the direct \(N\)-party analogue of the strategy
described in Supplemental Material~\ref{appendix:ExperimentAndMainResult}. Each party prepares the
same fermionic single-particle state as in~\cref{eq:FermionicTrustedState}.
After the redistribution of the states by the referee, the parties then perform the measurement as in
\cref{eq:FermionicProjectors}. We choose one
distinguished party, say \(A_1\), to use the measurement coefficients in
\cref{eq:meas_coeff_forA}, while all other parties use the coefficients in
\cref{eq:meas_coeff_forother}. This is the same choice used in the fixed-loop
self-testing statement of \cref{thm:fullselftesting}.

We now argue that the corresponding family of fermionic distributions cannot be
simultaneously reproduced by a local Q-QIT strategy, that is, by a strategy
based on distinguishable particles, bosons, or qubits. The claim is that no such
strategy can reproduce the exact fermionic distributions for all \(N\)-party
loops and all \((N-1)\)-party loops described above.

The proof follows the same logic as the three-party contradiction in
Supplemental Material~\ref{appendix:contradiction}. Suppose, for contradiction, that there exists a
local Q-QIT strategy reproducing all the fermionic loop distributions. For each
fixed loop configuration involving all \(N\) parties, the self-testing theorem
of \cref{thm:fullselftesting} implies that the corresponding states and
measurements are equivalent, up to local unitaries, to the trusted qubit
strategy described in the theorem. In particular, one obtains self-testing
relations of the form given in \cref{eq:FullSelf-testedState}, analogous to
those used in Supplemental Material~\ref{appendix:contradiction}, now for all \(N\)-party
loops.

The same self-testing theorem also applies to all \((N-1)\)-party loops that
include the special party \(A_1\). We then use the locality constraint,
which is crucial. The states prepared by the parties are fixed before the
referee chooses the loop configuration, and therefore cannot depend on that
configuration. Similarly, a party may adapt its measurement only to the local
subsystem it receives, not to the global loop chosen by the referee. Hence,
whenever the same locally prepared state is received in two different loop
configurations, the corresponding local unitary appearing in the self-tested
description must be the same.

Applying the fixed-loop self-testing relations to all \(N\)-party loops and to
the \((N-1)\)-party loops that contain \(A_1\), one obtains a collection of unitariy realtion and 
phase relations analogous to those derived in
Supplemental Material~\ref{appendix:contradiction}. These are the analogues of the relations
\crefrange{eq:loopABC_alpha}{eq:loopAB_beta},
appearing in the three-party proof. Together with locality, these relations
determine how the local unitaries must act also in the \((N-1)\)-party loops
that do not contain \(A_1\). In other words, the loops that include \(A_1\)
fix the effective action of the strategy on the remaining \((N-1)\)-party
loops.

The resulting prediction for the remaining loop is incompatible with the
fermionic one. More explicitly, consider the loop
\(\ell=(A_2,\dots,A_N)\), and take all inputs to be equal to \(1\).
The Q-QIT strategy then predicts a conditional distribution
\[
P^{A_2\cdots A_N}
\left(
a_{A_2},\dots,a_{A_N}
\,\middle|\,
(a_{A_2} \cdots a_{A_N}) \in \{1_+, 1_- \}^{N-1},\,
x_{A_2}=\cdots=x_{A_N}=1
\right),
\]
which differs from the fermionic distribution
\[
P_f^{A_2\cdots A_N}
\left(
a_{A_2},\dots,a_{A_N}
\,\middle|\,
(a_{A_2} \cdots a_{A_N}) \in \{1_+, 1_- \}^{N-1},\,
x_{A_2}=\cdots=x_{A_N}=1
\right).
\]
For each outcome \((a_{A_2},\dots,a_{A_N}) \in \{ 1_\pm\}^{N-1}\) , the absolute difference is
\begin{align}
&\left|
P_f^{A_2\cdots A_N}
\left(
a_{A_2},\dots,a_{A_N}
\,\middle|\,
(a_{A_2} \cdots a_{A_N}) \in \{1_+, 1_- \}^{N-1},\,
x_{A_2}=\cdots=x_{A_N}=1
\right)
\right. \nonumber \\
&\left.
\hspace{1.5cm}
-
P^{A_2\cdots A_N}
\left(
a_{A_2},\dots,a_{A_N}
\,\middle|\,
(a_{A_2} \cdots a_{A_N}) \in \{1_+, 1_- \}^{N-1},\,
x_{A_2}=\cdots=x_{A_N}=1
\right)
\right|
=
\frac{1}{2^{N-2}} .
\end{align}
Summing over all \(2^{N-1}\) possible values of
\((a_{A_2},\dots,a_{A_N})\), we obtain
\begin{align}
&\sum_{\substack{a_{A_i},\,x_{A_i}\\ i=2,\dots,N}}
\left|
P_f^{A_2\cdots A_N}
\left(
a_{A_2},\dots,a_{A_N}
\,\middle|\,
(a_{A_2} \cdots a_{A_N}) \in \{1_+, 1_- \}^{N-1},\,
x_{A_2}, \cdots, x_{A_N}
\right)
\right. \nonumber \\
&\left.
\hspace{1.5cm}
-
P^{A_2\cdots A_N}
\left(
a_{A_2},\dots,a_{A_N}
\,\middle|\,
(a_{A_2} \cdots a_{A_N}) \in \{1_+, 1_- \}^{N-1},\,
x_{A_2}, \cdots, x_{A_N}
\right)
\right|
=
2 .
\end{align}
Equivalently, the total variation distance between the two conditional
distributions is
\begin{equation}
d_{\rm TV}^{\rm cond}
\left(
P_f^{A_2\cdots A_N},P^{A_2\cdots A_N}
\right)
=
1.
\end{equation}
Therefore, a local Q-QIT strategy cannot reproduce all the fermionic loop
distributions simultaneously.

The distance between the full distributions is
\begin{equation}
d_{\rm TV}
\left(
P_f^{A_2\cdots A_N},P^{A_2\cdots A_N}
\right)
=
\frac{1}{2^{N-2}},
\end{equation}
which decays exponentially with \(N\). Nevertheless, this mismatch occurs for \((N-2)!\) loops. Hence, the cumulative discrepancy over all such loops scales as
\[
\frac{(N-2)!}{2^{N-2}},
\]
and therefore increases with \(N\).

\section{First quantized description of the fermionic protocol}\label{appendix:FirstQuantization}
Here we describe the state prepration phase of the fermionic trusted protocol in the first quantization picture.
Although equivalent to the second-quantized description used in the main text
and previous appendices, the first-quantized formulation is less transparent
from the perspective of locally preparable states and local operations, since both states and measurements
must be antisymmetrized. Once antisymmetrised, a local operation is represented as a mathematical operator acting over all particles, making the notion of locality harder to read in the mathematical formalism of first quantization. Second quantization solves that problem by going to a mode picture, in which the described systems are modes which may or may not contain particles. Operations are acting over these modes, which makes the notion of locality more intuitive to read in the mathematical formalism of second quantization.

Let \(S_N\) denote the permutation group on \(N\) particles. For each
\(\sigma \in S_N\), let \(P_\sigma\) be the corresponding permutation operator on
the \(N\)-particle tensor-product space, defined by
\begin{equation}
    P_\sigma
    \bigl(
        \ket{\psi_1}_1 \otimes \cdots \otimes \ket{\psi_N}_N
    \bigr)
    =
    \ket{\psi_{\sigma(1)}}_1 \otimes \cdots
    \otimes \ket{\psi_{\sigma(N)}}_N .
\end{equation}
We denote by \(\mathcal{A}_N\) the fermionic antisymmetrization map,
\begin{equation}
    \mathcal{A}_N
    =
    \frac{1}{\sqrt{N!}}
    \sum_{\sigma \in S_N}
    \epsilon(\sigma) P_\sigma,
\end{equation}
where \(\epsilon(\sigma)\) is the signature of the permutation \(\sigma\). Thus, for
any \(N\)-particle state \(\ket{\psi}\),
\begin{equation}
    \mathcal{A}_N \ket{\psi}
    =
    \frac{1}{\sqrt{N!}}
    \sum_{\sigma \in S_N}
    \epsilon(\sigma) P_\sigma \ket{\psi}.\label{eq:antisymmetrised-state}
\end{equation}

In this formalism the notion of operator locality can be defined in the following way. Let us consider single particle operators\footnotemark{}, \textit{i.e.} acting on a single particle at a time. For distinguishable particles, an operator acting just on the $i$-th particle can be written as
\begin{equation}
    \id_{1}  \otimes \cdots \otimes O_{i} \otimes \cdots \otimes \id_{N}. \label{eq:single-particle-op-distinguishable}
\end{equation}
This operator is \emph{local} if it only modifies the single particle states of particle $i$ in a spatially local way. More precisely, consider a basis $\{\ket{\alpha}_i\}_\alpha$ of the single particle Hilbert space where each state $\ket{\alpha}_i$ is localised in a specific spatial region (\textit{e.g.} the states of a particle localised at the position of one of the game's local parties). Then for a local operator we require that $\braket{\alpha}{O_i|\alpha'}=0$ unless $\ket{\alpha}$ and $\ket{\alpha'}$ are localised in the same spatial region.

\footnotetext{One should in general also consider many-body interactions which act on more than one particle simultaneously. However for our specific experiment only single particle operators are ever necessary, so for simplicity we limit our discussion to those. }

In the case of indistinguishable particles, an operator of the form~\eqref{eq:single-particle-op-distinguishable} cannot exist as it acts individually on the $i$-th particle, meaning it should be able to distinguish it from the others. Notice also that the operator~\eqref{eq:single-particle-op-distinguishable} does not map an antisymmetrised state~\eqref{eq:antisymmetrised-state} to an antisymmetrized state. A single particle operator for indistinguishable particles must indeed act uniformly across all particles, \textit{i.e.} it must have the form
\begin{align}
    \sum_{i=1}^N  \underset{\mbox{$i$-th particle}}{\id  \otimes \cdots \otimes \underset{\uparrow}{O} \otimes \cdots \otimes \id },
    \label{eq:single-particle-op-indistinguishable}
\end{align}
where the same operator $O$ acts on all particles. The operator is in particular \emph{local} if the operator $O$ is local in the sense defined above. Note that an operator of this form in particular commutes with the antisymmetrisation  operator $\mathcal{A}_N$, thereby preserving the antisymmetrisation of states.

\subsection{Particle-picture beam-splitter action}

For a $50{:}50$ beam splitter, the amplitudes at the input and output ports are
related by a $2\times 2$ unitary transformation whose entries all have magnitude
$1/\sqrt{2}$. With a symmetric choice of phases, this transformation is
\begin{equation}
U_{\rm BS}^{(i)}
=
\frac{1}{\sqrt{2}}
\begin{pmatrix}
1 & i\\
i & 1
\end{pmatrix}.
\end{equation}
Thus, denoting the input ports amplitude by $a_{\rm in}^\dagger$ and $b_{\rm in}^\dagger$  respectively,
\begin{equation}
a^\dagger_{\rm in}
\mapsto
\frac{1}{\sqrt{2}}
\left(
a^\dagger_{\rm out}
+i b^\dagger_{\rm out}
\right),
\end{equation}
and
\begin{equation}
b^\dagger_{\rm in}
\mapsto
\frac{1}{\sqrt{2}}
\left(
i a^\dagger_{\rm out}
+
b^\dagger_{\rm out}
\right).
\end{equation}

Equivalently, one may use the asymmetric, or Hadamard, beam-splitter convention,
defined by
\begin{equation}
U_{\rm BS}^{(H)}
=
\frac{1}{\sqrt{2}}
\begin{pmatrix}
1 & 1\\
1 & -1
\end{pmatrix}.
\end{equation}
In what follows, we use the Hadamard convention for convenience.

In the particle picture, a single-particle state $\ket{I_i}$ entering a
$50{:}50$ beam splitter is transformed as 
\begin{equation}
\ket{I_i}
\mapsto
\frac{1}{\sqrt{2}}
\left(
\ket{K_i}
+
\ket{T_{i+1}}
\right),
\end{equation}
Here, \(I_i\) denotes the initial single-particle state before the beam splitter at party \(i\). The beam splitter maps this state into a superposition of two output modes: \(K_i\), corresponding to the component kept at party \(i\), and \(T_{i+1}\), corresponding to the component transmitted to party \(i+1\). These states are different single-particle states that live in the same Hilbert space as $I_i$. They can be corresponding to different states for instance due to different spatial locations. See \cref{fig:SuperPositionState}(a).

When we need to keep track of which particle occupies a given single-particle state, we write the particle label explicitly as a subscript on the ket. For example, if particle \(i\) is initially at party \(A\), then the beam-splitter transformation is
\begin{equation}
\label{eq:BSFirstQuant}
\ket{I_A}_{i}
\mapsto
\frac{1}{\sqrt{2}}
\left(
\ket{K_A}_{i}
+
\ket{T_B}_{i}
\right).
\end{equation}
This means that particle \(i\), initially located at \(A\), is transformed into a coherent superposition of being kept at \(A\) and being transmitted to \(B\).

\subsection{\texorpdfstring{Three-party loop \(A\to B\to C\to A\)}{Three-party loop A → B → C → A}}
Consider three space-like separated parties \(A,B,C\), each holding one
identical fermion labeled by 1, 2, 3 respectively. We denote the single-particle states by
\(\ket{I_A}_1, \ket{I_B}_2 , \ket{I_C}_3\). The
antisymmetrized initial state is 
\begin{align}
\ket{\Psi_{\rm in}^{ABC}}
&= \mathcal{A}_3
\left(
\ket{I_A}_1 \ket{I_B}_2 \ket{I_C}_3
\right) \nonumber\\
&=
\frac{1}{\sqrt{6}}
\Big[
\ket{I_A}_1\ket{I_B}_2\ket{I_C}_3
+
\ket{I_B}_1\ket{I_C}_2\ket{I_A}_3
+
\ket{I_C}_1\ket{I_A}_2\ket{I_B}_3
\nonumber\\
&\hspace{2.4cm}
-
\ket{I_A}_1\ket{I_C}_2\ket{I_B}_3
-
\ket{I_C}_1\ket{I_B}_2\ket{I_A}_3
-
\ket{I_B}_1\ket{I_A}_2\ket{I_C}_3
\Big].
\end{align}

We consider the three-party directed loop
\begin{equation}
A\to B,\qquad B\to C,\qquad C\to A .
\end{equation}
Thus the beam-splitter transformations are
\begin{align}
\ket{I_A}_{i}
&\mapsto
\frac{1}{\sqrt{2}}
\left(
\ket{K_A}_{i}
+
\ket{T_B}_{i}
\right), \\
\ket{I_B}_{j}
&\mapsto
\frac{1}{\sqrt{2}}
\left(
\ket{K_B}_{j}
+
\ket{T_C}_{j}
\right), \\
\ket{I_C}_{k}
&\mapsto
\frac{1}{\sqrt{2}}
\left(
\ket{K_C}_{k}
+
\ket{T_A}_{k}
\right).
\end{align}
Note that the action of each beam splitter is implicitly symmetrized with respect to the particles; that is, it takes the form given in \cref{eq:single-particle-op-indistinguishable}.

Applying these transformations to the antisymmetrized input state gives
\begin{align}
\ket{\Psi_{\rm out}^{ABC}}
=
\frac{1}{\sqrt{6}}\frac{1}{2\sqrt{2}}
\Big[
&
\left(\ket{K_A}_{1}+\ket{T_B}_{1}\right)
\left(\ket{K_B}_{2}+\ket{T_C}_{2}\right)
\left(\ket{K_C}_{3}+\ket{T_A}_{3}\right)
\nonumber\\
&+
\left(\ket{K_B}_{1}+\ket{T_C}_{1}\right)
\left(\ket{K_C}_{2}+\ket{T_A}_{2}\right)
\left(\ket{K_A}_{3}+\ket{T_B}_{3}\right)
\nonumber\\
&+
\left(\ket{K_C}_{1}+\ket{T_A}_{1}\right)
\left(\ket{K_A}_{2}+\ket{T_B}_{2}\right)
\left(\ket{K_B}_{3}+\ket{T_C}_{3}\right)
\nonumber\\
&-
\left(\ket{K_A}_{1}+\ket{T_B}_{1}\right)
\left(\ket{K_C}_{2}+\ket{T_A}_{2}\right)
\left(\ket{K_B}_{3}+\ket{T_C}_{3}\right)
\nonumber\\
&-
\left(\ket{K_C}_{1}+\ket{T_A}_{1}\right)
\left(\ket{K_B}_{2}+\ket{T_C}_{2}\right)
\left(\ket{K_A}_{3}+\ket{T_B}_{3}\right)
\nonumber\\
&-
\left(\ket{K_B}_{1}+\ket{T_C}_{1}\right)
\left(\ket{K_A}_{2}+\ket{T_B}_{2}\right)
\left(\ket{K_C}_{3}+\ket{T_A}_{3}\right)
\Big].
\end{align}
Equivalently, since the beam splitters act independently at the single-particle
level, the output state can be written compactly as the antisymmetrization of
the beam-splitter-transformed product state:
\begin{equation}
\ket{\Psi_{\rm out}^{ABC}}
=
\mathcal{A}_3
\left[
\frac{1}{2\sqrt{2}}
\left(\ket{K_A}_{1}+\ket{T_B}_{1}\right)
\left(\ket{K_B}_{2}+\ket{T_C}_{2}\right)
\left(\ket{K_C}_{3}+\ket{T_A}_{3}\right)
\right]
\end{equation}

We now postselect on the event that each party has exactly one fermion. This
condition selects the terms in which either all particles are kept or all
particles are transmitted. After normalization, the postselected state is
\begin{align}
\ket{\Psi_{\rm post}^{ABC}}
=
\frac{1}{\sqrt{12}}
\Big[
&
\ket{K_A}_{1}\ket{K_B}_{2}\ket{K_C}_{3}
+
\ket{T_B}_{1}\ket{T_C}_{2}\ket{T_A}_{3}
\nonumber\\
&+
\ket{K_B}_{1}\ket{K_C}_{2}\ket{K_A}_{3}
+
\ket{T_C}_{1}\ket{T_A}_{2}\ket{T_B}_{3}
\nonumber\\
&+
\ket{K_C}_{1}\ket{K_A}_{2}\ket{K_B}_{3}
+
\ket{T_A}_{1}\ket{T_B}_{2}\ket{T_C}_{3}
\nonumber\\
&-
\ket{K_A}_{1}\ket{K_C}_{2}\ket{K_B}_{3}
-
\ket{T_B}_{1}\ket{T_A}_{2}\ket{T_C}_{3}
\nonumber\\
&-
\ket{K_C}_{1}\ket{K_B}_{2}\ket{K_A}_{3}
-
\ket{T_A}_{1}\ket{T_C}_{2}\ket{T_B}_{3}
\nonumber\\
&-
\ket{K_B}_{1}\ket{K_A}_{2}\ket{K_C}_{3}
-
\ket{T_C}_{1}\ket{T_B}_{2}\ket{T_A}_{3}
\Big].
\end{align}

Equivalently, this state can be written as a coherent superposition of the
fully kept and fully transmitted antisymmetric components,
\begin{equation}
\ket{\Psi_{\rm post}^{ABC}}
=
\frac{1}{\sqrt{2}}
\left(
\mathcal{A}_3
\left[
\ket{K_A}_{1}\ket{K_B}_{2}\ket{K_C}_{3}
\right]
+
\mathcal{A}_3
\left[
\ket{T_B}_{1}\ket{T_C}_{2}\ket{T_A}_{3}
\right]
\right).
\end{equation}
Defining
\begin{align}
\ket{\Psi_K^{ABC}}
:=
\mathcal{A}_3
\left[
\ket{K_A}_{1}\ket{K_B}_{2}\ket{K_C}_{3}
\right],
\end{align}
and
\begin{align}
\ket{\Psi_T^{ABC}}
:=
\mathcal{A}_3
\left[
\ket{T_A}_{1}\ket{T_B}_{2}\ket{T_C}_{3}
\right],
\end{align}
we obtain
\begin{equation}
\ket{\Psi_{\rm post}^{ABC}}
=
\frac{1}{\sqrt{2}}
\left(
\ket{\Psi_K^{ABC}}
+
\ket{\Psi_T^{ABC}}
\right).
\end{equation}

\subsection{\texorpdfstring{Two-party loop \(A\leftrightarrow B\), with \(C\) excluded}{Two-party loop A ↔ B, with C excluded}}

We now consider the case in which only \(A\) and \(B\) are connected in a loop,
\begin{equation}
A\to B,\qquad B\to A,
\end{equation}
while \(C\) is excluded from the loop. The input state is again the
antisymmetrized three-fermion state
\begin{align}
\ket{\Psi_{\rm in}^{AB|C}}
&= \mathcal{A}_3
\left(
\ket{I_A}_1 \ket{I_B}_2 \ket{I_C}_3
\right) \nonumber\\
&=
\frac{1}{\sqrt{6}}
\Big[
\ket{I_A}_1\ket{I_B}_2\ket{I_C}_3
+
\ket{I_B}_1\ket{I_C}_2\ket{I_A}_3
+
\ket{I_C}_1\ket{I_A}_2\ket{I_B}_3
\nonumber\\
&\hspace{2.4cm}
-
\ket{I_A}_1\ket{I_C}_2\ket{I_B}_3
-
\ket{I_C}_1\ket{I_B}_2\ket{I_A}_3
-
\ket{I_B}_1\ket{I_A}_2\ket{I_C}_3
\Big].
\end{align}

For a generic product term
\begin{equation}
\ket{I_A}_{i}\ket{I_B}_{j}\ket{I_C}_{k},
\end{equation}
the beam-splitter transformations are now
\begin{align}
\ket{I_A}_{i}
&\mapsto
\frac{1}{\sqrt{2}}
\left(
\ket{K_A}_{i}
+
\ket{T_B}_{i}
\right), \\
\ket{I_B}_{j}
&\mapsto
\frac{1}{\sqrt{2}}
\left(
\ket{K_B}_{j}
+
\ket{T_A}_{j}
\right), \\
\ket{I_C}_{k}
&\mapsto
\frac{1}{\sqrt{2}}
\left(
\ket{K_C}_{k}
+
\ket{T_C}_{k}
\right).
\end{align}
Note that here the \(C\)-particle is not transmitted to another party; but it is given back to itself.

Applying these transformations to the antisymmetrized input state gives
\begin{align}
\ket{\Psi_{\rm out}^{AB|C}}
=
\frac{1}{\sqrt{6}}\frac{1}{2\sqrt{2}}
\Big[
&
\left(\ket{K_A}_{1}+\ket{T_B}_{1}\right)
\left(\ket{K_B}_{2}+\ket{T_A}_{2}\right)
\left(\ket{K_C}_{3}+\ket{T_C}_{3}\right)
\nonumber\\
&+
\left(\ket{K_B}_{1}+\ket{T_A}_{1}\right)
\left(\ket{K_C}_{2}+\ket{T_C}_{2}\right)
\left(\ket{K_A}_{3}+\ket{T_B}_{3}\right)
\nonumber\\
&+
\left(\ket{K_C}_{1}+\ket{T_C}_{1}\right)
\left(\ket{K_A}_{2}+\ket{T_B}_{2}\right)
\left(\ket{K_B}_{3}+\ket{T_A}_{3}\right)
\nonumber\\
&-
\left(\ket{K_A}_{1}+\ket{T_B}_{1}\right)
\left(\ket{K_C}_{2}+\ket{T_C}_{2}\right)
\left(\ket{K_B}_{3}+\ket{T_A}_{3}\right)
\nonumber\\
&-
\left(\ket{K_C}_{1}+\ket{T_C}_{1}\right)
\left(\ket{K_B}_{2}+\ket{T_A}_{2}\right)
\left(\ket{K_A}_{3}+\ket{T_B}_{3}\right)
\nonumber\\
&-
\left(\ket{K_B}_{1}+\ket{T_A}_{1}\right)
\left(\ket{K_A}_{2}+\ket{T_B}_{2}\right)
\left(\ket{K_C}_{3}+\ket{T_C}_{3}\right)
\Big].
\end{align}
Equivalently, since the beam splitters act independently at the single-particle
level, the output state can be written compactly as the antisymmetrization of
the beam-splitter-transformed product state:
\begin{equation}
\ket{\Psi_{\rm out}^{AB|C}}
=
\mathcal{A}_3
\left[
\frac{1}{2\sqrt{2}}
\left(\ket{K_A}_{1}+\ket{T_B}_{1}\right)
\left(\ket{K_B}_{2}+\ket{T_A}_{2}\right)
\left(\ket{K_C}_{3}+\ket{T_C}_{3}\right)
\right].
\end{equation}

We now postselect on the event that each party has exactly one fermion. This
condition selects the terms in which the particles at \(A\) and \(B\) are either
both kept or both transmitted, while the particle at \(C\) remains in one of its
two local states. After normalisation
\begin{align}
\ket{\Psi_{\rm post}^{AB|C}}
=
\frac{1}{\sqrt{24}}
\Big[
&
\ket{K_A}_{1}\ket{K_B}_{2}
\left(
\ket{K_C}_{3}+\ket{T_C}_{3}
\right)
+
\ket{T_B}_{1}\ket{T_A}_{2}
\left(
\ket{K_C}_{3}+\ket{T_C}_{3}
\right)
\nonumber\\
&+
\ket{K_B}_{1}
\left(
\ket{K_C}_{2}+\ket{T_C}_{2}
\right)
\ket{K_A}_{3}
+
\ket{T_A}_{1}
\left(
\ket{K_C}_{2}+\ket{T_C}_{2}
\right)
\ket{T_B}_{3}
\nonumber\\
&+
\left(
\ket{K_C}_{1}+\ket{T_C}_{1}
\right)
\ket{K_A}_{2}\ket{K_B}_{3}
+
\left(
\ket{K_C}_{1}+\ket{T_C}_{1}
\right)
\ket{T_B}_{2}\ket{T_A}_{3}
\nonumber\\
&-
\ket{K_A}_{1}
\left(
\ket{K_C}_{2}+\ket{T_C}_{2}
\right)
\ket{K_B}_{3}
-
\ket{T_B}_{1}
\left(
\ket{K_C}_{2}+\ket{T_C}_{2}
\right)
\ket{T_A}_{3}
\nonumber\\
&-
\left(
\ket{K_C}_{1}+\ket{T_C}_{1}
\right)
\ket{K_B}_{2}\ket{K_A}_{3}
-
\left(
\ket{K_C}_{1}+\ket{T_C}_{1}
\right)
\ket{T_A}_{2}\ket{T_B}_{3}
\nonumber\\
&-
\ket{K_B}_{1}\ket{K_A}_{2}
\left(
\ket{K_C}_{3}+\ket{T_C}_{3}
\right)
-
\ket{T_A}_{1}\ket{T_B}_{2}
\left(
\ket{K_C}_{3}+\ket{T_C}_{3}
\right)
\Big].
\end{align}
The postselected state can be written
as
\begin{equation}
\ket{\Psi_{\rm post}^{AB|C}}
=
\frac{1}{\sqrt{2}}
\left(
\mathcal{A}_3
\left[
\ket{K_A}_{1}\ket{K_B}_{2}
\frac{
\ket{K_C}_{3}+\ket{T_C}_{3}
}{\sqrt{2}}
\right]
-
\mathcal{A}_3
\left[
\ket{T_A}_{1}\ket{T_B}_{2}
\frac{
\ket{K_C}_{3}+\ket{T_C}_{3}
}{\sqrt{2}}
\right]
\right).
\end{equation}
The relative minus sign appears because, in the transmitted component, the
particles associated with \(A\) and \(B\) are exchanged.

Equivalently, define
\begin{align}
\ket{\Psi_K^{AB|C}}
:=
\mathcal{A}_3
\left[
\ket{K_A}_{1}\ket{K_B}_{2}
\frac{
\ket{K_C}_{3}+\ket{T_C}_{3}
}{\sqrt{2}}
\right],
\end{align}
and
\begin{align}
\ket{\Psi_T^{AB|C}}
:=
\mathcal{A}_3
\left[
\ket{T_A}_{1}\ket{T_B}_{2}
\frac{
\ket{K_C}_{3}+\ket{T_C}_{3}
}{\sqrt{2}}
\right].
\end{align}
Then
\begin{equation}
\ket{\Psi_{\rm post}^{AB|C}}
=
\frac{1}{\sqrt{2}}
\left(
\ket{\Psi_K^{AB|C}}
-
\ket{\Psi_T^{AB|C}}
\right).
\end{equation}

The relative sign is the key feature. The three-party loop produces a plus sign between the all-kept and all-transmitted components, because it involves two fermionic exchanges. By contrast, the two-party loop produces a minus sign between the \(A,B\)-kept and \(A,B\)-transmitted components, due to a single fermionic exchange. This sensitivity to the parity of the loop is the origin of the fermionic advantage, which can be revealed at the level of observable correlations by applying suitable measurements to the states above.
Note that the same computation can be done with bosons, replacing the the fermionic antisymmetrization map $\mathcal{A}_N$ by the bosonic symmetrization map $\mathcal{S}_N$:
\[
    \mathcal{S}_N
    =
    \frac{1}{\sqrt{N!}}
    \sum_{\sigma \in S_N}
    P_\sigma.
\]
In that case, the derivation would be essentially the same, but with no minus sign appearing in the $AB$ loop case.

\FloatBarrier
\numberwithin{equation}{subsection}
\section{Proof of Lemmas}
\subsection{GHZ-state self-testing}
\label{appendix:GHZ_selftesting}

In this Supplemental Material we recall the stabilizer relations that follow from maximal
violation of the two Bell inequalities used in the GHZ self-testing
construction, following Ref.~\cite{baccari2020scalable}. We state the two cases
separately, since the sign of the \(N\)-body stabilizer differs for the two Bell
inequalities.

\subsubsection{\texorpdfstring{\cref{lem:FirstBell}}{First Bell lemma}}
\label{appendix:FirstBell}

If Bell inequality \Romannum{1} is maximally violated, namely
\begin{equation}
\label{eq:Bell_N_first}
\left\langle
B
\equiv
(N-1)(M_0^{(1)}+M_1^{(1)})M_1^{(2)}\cdots M_1^{(N)}
+
\sum_{i=2}^{N}
(M_0^{(1)}-M_1^{(1)})M_0^{(i)}
\right\rangle
=
\beta_Q,
\end{equation}
where
\[
\beta_C=2(N-1),
\qquad
\beta_Q=2(N-1)\sqrt{2},
\]
then \(\mathcal{X}_1\cdots \mathcal{X}_N\) is a stabilizer of the
conditional state:
\begin{equation}
\mathcal{X}_1\cdots \mathcal{X}_N
\ket{\psi}^{\rm cond}
=
\ket{\psi}^{\rm cond}.
\end{equation}
Moreover, for every \(i\in\{1,\dots,N\}\),
\begin{equation}
\begin{aligned}
\mathcal{X}_i^2\ket{\psi}^{\rm cond}
=
\mathcal{Z}_i^2\ket{\psi}^{\rm cond}
&=
\ket{\psi}^{\rm cond},\\
\{\mathcal{X}_i,\mathcal{Z}_i\}\ket{\psi}^{\rm cond}
&=
0.
\end{aligned}
\end{equation}
Here
\begin{equation}
\begin{aligned}
\mathcal{X}_{1}
&=
\frac{M^{(1)}_0+M^{(1)}_1}{\sqrt{2}},
&
\mathcal{Z}_{1}
&=
\frac{M^{(1)}_0-M^{(1)}_1}{\sqrt{2}},
\\
\mathcal{X}_{i}
&=
M^{(i)}_1,
&
\mathcal{Z}_{i}
&=
M^{(i)}_0,
\qquad i\in\{2,\dots,N\}.
\end{aligned}
\end{equation}

\begin{proof}
Maximal violation of Bell inequality \Romannum{1} gives
\begin{equation}
\left\langle
B
\equiv
(N-1)(M_0^{(1)}+M_1^{(1)})M_1^{(2)}\cdots M_1^{(N)}
+
\sum_{i=2}^{N}
(M_0^{(1)}-M_1^{(1)})M_0^{(i)}
\right\rangle
=
\beta_Q .
\end{equation}
Equivalently, the following sum-of-squares decomposition is saturated:
\begin{equation}
\left\langle
\beta_Q\id-B
\right\rangle
=
\frac{1}{\sqrt{2}}
\left\langle
(N-1)(\id-S_1)^2
+
\sum_{i=2}^{N}
(\id-S_i)^2
\right\rangle
=
0,
\end{equation}
where
\begin{equation}
\label{eq:stabilizers_first}
\begin{aligned}
S_1
&=
\mathcal{X}_1\cdots \mathcal{X}_N,\\
S_i
&=
\mathcal{Z}_1\mathcal{Z}_i,
\qquad i\in\{2,\dots,N\}.
\end{aligned}
\end{equation}
Since each term in the sum of squares is positive, maximal violation implies
that each square vanishes on the conditional state. Therefore,
\begin{equation}
S_i\ket{\psi}^{\rm cond}
=
\ket{\psi}^{\rm cond},
\qquad
i\in\{1,\dots,N\}.
\end{equation}
Thus the operators \(S_i\) are stabilizers of the state attaining maximal
violation of Bell inequality \Romannum{1}.

We first show that all \(\mathcal{X}_i\) and \(\mathcal{Z}_i\) square to the
identity on \(\ket{\psi}^{\rm cond}\). For parties \(i\in\{2,\dots,N\}\), this
is immediate from the construction, since \(M_0^{(i)}\) and \(M_1^{(i)}\) are
Hermitian unitary observables. It remains to prove the same statement for
party \(1\).

From \cref{eq:stabilizers_first}, for every \(i\in\{2,\dots,N\}\),
\begin{equation}
\mathcal{Z}_1\mathcal{Z}_i\ket{\psi}^{\rm cond}
=
\ket{\psi}^{\rm cond}.
\end{equation}
Since \(\mathcal{Z}_i^2=\id\) for \(i\geq 2\), this implies
\begin{equation}
\mathcal{Z}_1\ket{\psi}^{\rm cond}
=
\mathcal{Z}_i\ket{\psi}^{\rm cond}.
\end{equation}
Consequently,
\begin{equation}
\mathcal{Z}_1^2\ket{\psi}^{\rm cond}
=
\mathcal{Z}_1\mathcal{Z}_i\ket{\psi}^{\rm cond}
=
S_i\ket{\psi}^{\rm cond}
=
\ket{\psi}^{\rm cond}.
\end{equation}

We now prove the corresponding relation for \(\mathcal{X}_1\). Since
\[
S_1=\mathcal{X}_1\mathcal{X}_2\cdots\mathcal{X}_N
\]
and \(\mathcal{X}_i^2=\id\) for \(i\geq 2\), we have
\begin{equation}
\mathcal{X}_1^2\ket{\psi}^{\rm cond}
=
S_1^2\ket{\psi}^{\rm cond}.
\end{equation}
Using \(S_1\ket{\psi}^{\rm cond}=\ket{\psi}^{\rm cond}\), it follows that
\begin{equation}
\mathcal{X}_1^2\ket{\psi}^{\rm cond}
=
S_1^2\ket{\psi}^{\rm cond}
=
S_1\ket{\psi}^{\rm cond}
=
\ket{\psi}^{\rm cond}.
\end{equation}

It remains to prove the anticommutation relations. For party \(1\), the
anticommutation follows directly from the definitions:
\[
\mathcal{X}_1
=
\frac{M_0^{(1)}+M_1^{(1)}}{\sqrt{2}},
\qquad
\mathcal{Z}_1
=
\frac{M_0^{(1)}-M_1^{(1)}}{\sqrt{2}},
\]
together with \((M_0^{(1)})^2=(M_1^{(1)})^2=\id\).

We now prove the relation for the remaining parties. It is enough to present the
argument for party \(2\); the same reasoning applies to any
\(i\in\{2,\dots,N\}\). Since
\[
S_1=\mathcal{X}_1\mathcal{X}_2\cdots\mathcal{X}_N,
\]
and since the observables on different parties commute, we can write, on the
relevant conditional support,
\begin{equation}
\mathcal{X}_2\ket{\psi}^{\rm cond}
=
\mathcal{X}_1\mathcal{X}_3\cdots\mathcal{X}_N
S_1\ket{\psi}^{\rm cond}.
\end{equation}
Similarly,
\begin{equation}
\mathcal{X}_2
\left(
\mathcal{Z}_2\ket{\psi}^{\rm cond}
\right)
=
\mathcal{X}_1\mathcal{X}_3\cdots\mathcal{X}_N
S_1
\left(
\mathcal{Z}_2\ket{\psi}^{\rm cond}
\right).
\end{equation}
Therefore,
\begin{equation}
\begin{aligned}
\{\mathcal{X}_2,\mathcal{Z}_2\}
\ket{\psi}^{\rm cond}
&=
\mathcal{X}_1\mathcal{X}_3\cdots\mathcal{X}_N
\{S_1,\mathcal{Z}_2\}
\ket{\psi}^{\rm cond}.
\end{aligned}
\end{equation}
Using the stabilizer relation
\[
\mathcal{Z}_1\mathcal{Z}_2\ket{\psi}^{\rm cond}
=
\ket{\psi}^{\rm cond},
\]
we may replace \(\mathcal{Z}_2\ket{\psi}^{\rm cond}\) by
\(\mathcal{Z}_1\ket{\psi}^{\rm cond}\) inside the above expression. Hence,
\begin{equation}
\begin{aligned}
\{\mathcal{X}_2,\mathcal{Z}_2\}
\ket{\psi}^{\rm cond}
&=
\mathcal{X}_1\mathcal{X}_3\cdots\mathcal{X}_N
\{S_1,\mathcal{Z}_1\}
\ket{\psi}^{\rm cond}\\
&=
\mathcal{X}_1\mathcal{X}_2
\{\mathcal{X}_1,\mathcal{Z}_1\}
\ket{\psi}^{\rm cond}\\
&=0.
\end{aligned}
\end{equation}
The final equality follows from
\(\{\mathcal{X}_1,\mathcal{Z}_1\}=0\). The same argument gives
\begin{equation}
\{\mathcal{X}_i,\mathcal{Z}_i\}\ket{\psi}^{\rm cond}
=
0,
\qquad
i\in\{2,\dots,N\}.
\end{equation}
This proves the claim.
\end{proof}

\subsubsection{\texorpdfstring{\cref{lem:SecondBell}}{Second Bell lemma}}\label{appendix:SecondBell}

If Bell inequality \Romannum{2} is maximally violated, namely
\begin{equation}
\label{eq:Bell_N_second}
\left\langle
B
\equiv
-(N-1)(M_0^{(1)}+M_1^{(1)})M_1^{(2)}\cdots M_1^{(N)}
+
\sum_{i=2}^{N}
(M_0^{(1)}-M_1^{(1)})M_0^{(i)}
\right\rangle
=
\beta_Q,
\end{equation}
where
\[
\beta_C=2(N-1),
\qquad
\beta_Q=2(N-1)\sqrt{2},
\]
then \(-\mathcal{X}_1\cdots\mathcal{X}_N\) is a stabilizer of the conditional
state:
\begin{equation}
-
\mathcal{X}_1\cdots\mathcal{X}_N
\ket{\psi}^{\rm cond}
=
\ket{\psi}^{\rm cond}.
\end{equation}
Moreover, for every \(i\in\{1,\dots,N\}\),
\begin{equation}
\begin{aligned}
\mathcal{X}_i^2\ket{\psi}^{\rm cond}
=
\mathcal{Z}_i^2\ket{\psi}^{\rm cond}
&=
\ket{\psi}^{\rm cond},\\
\{\mathcal{X}_i,\mathcal{Z}_i\}\ket{\psi}^{\rm cond}
&=
0.
\end{aligned}
\end{equation}
Here
\begin{equation}
\begin{aligned}
\mathcal{X}_{1}
&=
\frac{M^{(1)}_0+M^{(1)}_1}{\sqrt{2}},
&
\mathcal{Z}_{1}
&=
\frac{M^{(1)}_0-M^{(1)}_1}{\sqrt{2}},
\\
\mathcal{X}_{i}
&=
M^{(i)}_1,
&
\mathcal{Z}_{i}
&=
M^{(i)}_0,
\qquad i\in\{2,\dots,N\}.
\end{aligned}
\end{equation}

\begin{proof}
Maximal violation of Bell inequality \Romannum{2} gives
\begin{equation}
\left\langle
B
\equiv
-(N-1)(M_0^{(1)}+M_1^{(1)})M_1^{(2)}\cdots M_1^{(N)}
+
\sum_{i=2}^{N}
(M_0^{(1)}-M_1^{(1)})M_0^{(i)}
\right\rangle
=
\beta_Q .
\end{equation}
Equivalently, the following sum-of-squares decomposition is saturated:
\begin{equation}
\left\langle
\beta_Q\id-B
\right\rangle
=
\frac{1}{\sqrt{2}}
\left\langle
(N-1)(\id-S_1)^2
+
\sum_{i=2}^{N}
(\id-S_i)^2
\right\rangle
=
0,
\end{equation}
where
\begin{equation}
\label{eq:stabilizers_second}
\begin{aligned}
S_1
&=
-\mathcal{X}_1\cdots\mathcal{X}_N,\\
S_i
&=
\mathcal{Z}_1\mathcal{Z}_i,
\qquad i\in\{2,\dots,N\}.
\end{aligned}
\end{equation}
As before, since the sum of positive terms has vanishing expectation value, each
term vanishes on the conditional state. Therefore,
\begin{equation}
S_i\ket{\psi}^{\rm cond}
=
\ket{\psi}^{\rm cond},
\qquad
i\in\{1,\dots,N\}.
\end{equation}
Thus the operators \(S_i\) are stabilizers of the state attaining maximal
violation of Bell inequality \Romannum{2}.

We first prove that all \(\mathcal{X}_i\) and \(\mathcal{Z}_i\) square to the
identity on \(\ket{\psi}^{\rm cond}\). For parties \(i\in\{2,\dots,N\}\), this
again follows from the fact that \(M_0^{(i)}\) and \(M_1^{(i)}\) are Hermitian
unitary observables.

For \(\mathcal{Z}_1\), we use the stabilizers
\[
S_i=\mathcal{Z}_1\mathcal{Z}_i,
\qquad i\in\{2,\dots,N\}.
\]
From \cref{eq:stabilizers_second}, we have
\begin{equation}
\mathcal{Z}_1\mathcal{Z}_i\ket{\psi}^{\rm cond}
=
\ket{\psi}^{\rm cond}.
\end{equation}
Since \(\mathcal{Z}_i^2=\id\) for \(i\geq 2\), this implies
\begin{equation}
\mathcal{Z}_1\ket{\psi}^{\rm cond}
=
\mathcal{Z}_i\ket{\psi}^{\rm cond}.
\end{equation}
Therefore,
\begin{equation}
\mathcal{Z}_1^2\ket{\psi}^{\rm cond}
=
\mathcal{Z}_1\mathcal{Z}_i\ket{\psi}^{\rm cond}
=
S_i\ket{\psi}^{\rm cond}
=
\ket{\psi}^{\rm cond}.
\end{equation}

We now prove the same statement for \(\mathcal{X}_1\). Since
\[
S_1
=
-\mathcal{X}_1\mathcal{X}_2\cdots\mathcal{X}_N,
\]
and \(\mathcal{X}_i^2=\id\) for \(i\geq 2\), we have
\begin{equation}
\mathcal{X}_1^2\ket{\psi}^{\rm cond}
=
S_1^2\ket{\psi}^{\rm cond}.
\end{equation}
Using \(S_1\ket{\psi}^{\rm cond}=\ket{\psi}^{\rm cond}\), we obtain
\begin{equation}
\mathcal{X}_1^2\ket{\psi}^{\rm cond}
=
S_1^2\ket{\psi}^{\rm cond}
=
S_1\ket{\psi}^{\rm cond}
=
\ket{\psi}^{\rm cond}.
\end{equation}

It remains to prove the anticommutation relations. For party \(1\), the
anticommutation follows directly from the definitions of
\(\mathcal{X}_1\) and \(\mathcal{Z}_1\), together with the fact that
\(M_0^{(1)}\) and \(M_1^{(1)}\) are Hermitian unitaries.

We now prove the relation for party \(2\); the same argument applies to all
\(i\in\{2,\dots,N\}\). Because now
\[
S_1
=
-\mathcal{X}_1\mathcal{X}_2\cdots\mathcal{X}_N,
\]
we have, on the relevant conditional support,
\begin{equation}
\mathcal{X}_2\ket{\psi}^{\rm cond}
=
-
\mathcal{X}_1\mathcal{X}_3\cdots\mathcal{X}_N
S_1\ket{\psi}^{\rm cond},
\end{equation}
and similarly,
\begin{equation}
\mathcal{X}_2
\left(
\mathcal{Z}_2\ket{\psi}^{\rm cond}
\right)
=
-
\mathcal{X}_1\mathcal{X}_3\cdots\mathcal{X}_N
S_1
\left(
\mathcal{Z}_2\ket{\psi}^{\rm cond}
\right).
\end{equation}
Therefore,
\begin{equation}
\begin{aligned}
\{\mathcal{X}_2,\mathcal{Z}_2\}
\ket{\psi}^{\rm cond}
&=
-
\mathcal{X}_1\mathcal{X}_3\cdots\mathcal{X}_N
\{S_1,\mathcal{Z}_2\}
\ket{\psi}^{\rm cond}.
\end{aligned}
\end{equation}
Using the stabilizer relation
\[
\mathcal{Z}_1\mathcal{Z}_2\ket{\psi}^{\rm cond}
=
\ket{\psi}^{\rm cond},
\]
we may replace \(\mathcal{Z}_2\ket{\psi}^{\rm cond}\) by
\(\mathcal{Z}_1\ket{\psi}^{\rm cond}\), giving
\begin{equation}
\begin{aligned}
\{\mathcal{X}_2,\mathcal{Z}_2\}
\ket{\psi}^{\rm cond}
&=
-
\mathcal{X}_1\mathcal{X}_3\cdots\mathcal{X}_N
\{S_1,\mathcal{Z}_1\}
\ket{\psi}^{\rm cond}\\
&=
\mathcal{X}_1\mathcal{X}_2
\{\mathcal{X}_1,\mathcal{Z}_1\}
\ket{\psi}^{\rm cond}\\
&=0.
\end{aligned}
\end{equation}
The last equality again follows from
\(\{\mathcal{X}_1,\mathcal{Z}_1\}=0\). Hence
\begin{equation}
\{\mathcal{X}_i,\mathcal{Z}_i\}\ket{\psi}^{\rm cond}
=
0,
\qquad
i\in\{2,\dots,N\}.
\end{equation}
This proves the claim.
\end{proof}

\subsection{\texorpdfstring{\cref{lem:Zsep_SysJunk}}{System-junk separation lemma}}
\label{appendix:Zsep}

We prove that, on the conditional subspace \(H_i^{\rm cond}\), the operator
\(\mathcal{Z}_i\) acts only on the two-dimensional system part and
trivially on the junk degrees of freedom. More precisely, for every
\(i\in\{1,\dots,N\}\),
\begin{equation}
\begin{aligned}
\mathcal Z_i\big|_{H_i^{\mathrm{cond}}} = Z_{a_i}\otimes \id_{J_i}\big|_{H_i^{\mathrm{cond}}},
\end{aligned}
\end{equation}
where
\[
Z= \left(\ket{01}\bra{01}- \ket{10}\bra{10}\right)_{S_{i_R} S_{i_L}} .
\]

\begin{proof}
The token-counting self-testing result of Ref.~\cite{sekatski2023partial}
implies that the coarse grained measurement of party \(i\) is, up to local
isometries,
\begin{equation}
\begin{aligned}
\Pi_i^{0}
&=
\ketbra{00}{00}_{S_{i_R}S_{i_L}}\otimes \id_{J_i},\\
\Pi_i^{1}
&=
\left(
\ketbra{01}{01}
+
\ketbra{10}{10}
\right)_{S_{i_R}S_{i_L}}
\otimes \id_{J_i},\\
\Pi_i^{2}
&=
\ketbra{11}{11}_{S_{i_R}S_{i_L}}\otimes \id_{J_i}.
\end{aligned}
\end{equation}
Moreover, the conditional state, given that all parties obtain an output $a_i \in \{1_\pm \}$, has
the form
\begin{equation}
\ket{\psi}^{\rm cond}
=
\frac{1}{\sqrt{2}}
\left(
\ket{\psi_a}^{\rm cond}
+
\ket{\psi_c}^{\rm cond}
\right),
\end{equation}
where
\begin{align}
\ket{\psi_a}^{\rm cond}
&:=
\ket{01,\dots,01}_{S_{1_R}S_{1_L}\cdots S_{N_R}S_{N_L}}
\otimes
\ket{\zeta_1^{a}}_{J_{1_L}J_{2_R}}
\otimes \cdots \otimes
\ket{\zeta_N^{a}}_{J_{N_L}J_{1_R}},
\\
\ket{\psi_c}^{\rm cond}
&:=
\ket{10,\dots,10}_{S_{1_R}S_{1_L}\cdots S_{N_R}S_{N_L}}
\otimes
\ket{\zeta_1^{c}}_{J_{1_L}J_{2_R}}
\otimes \cdots \otimes
\ket{\zeta_N^{c}}_{J_{N_L}J_{1_R}} .
\end{align}

We use the probability terms associated with events in which one party
obtains an output in $\{1_\pm \}$, while its immediate neighbors obtain either zero or two. These terms constrain the action of the operators \(\mathcal{Z}_i\) on the two components of the conditional state. The argument is slightly different for \(i\neq 1\) and \(i=1\): in the former case, \(\mathcal{Z}_i\) is a binary observable, so \(\id\pm\mathcal{Z}_i\geq 0\); in the latter, \(\mathcal{Z}_1\) is inferred from the conditional observables of party 1 and requires a separate step.

\paragraph{Case \(i\neq 1\).}
Let \(i\in\{2,\dots,N\}\). From the probability terms of the trusted strategy,
we have
\begin{equation}
P(a_i=1_+ \mid a_{i-1}=2,a_i=1_\pm,a_{i+1}=0,x_i=0)=1 .
\end{equation}
Equivalently,
\begin{align}
1
&=
\frac{
P(a_i=1_+,a_{i-1}=2,a_{i+1}=0 \mid x_i=0)
}{
P(a_{i-1}=2,a_i=1_\pm,a_{i+1}=0 \mid x_i=0)
}
\nonumber\\
&=
\frac{
\bra{\psi}
\Pi^2_{i-1}
\frac{\id+\mathcal{Z}_i}{2}
\Pi^0_{i+1}
\ket{\psi}
}{
P(a_{i-1}=2,a_i=1_\pm,a_{i+1}=0 \mid x_i=0)
}
\nonumber\\
&=
\bra{\psi_{i,a}}^{\rm cond}
\frac{\id+\mathcal{Z}_i}{2}
\ket{\psi_{i,a}}^{\rm cond}.
\end{align}
Here \(\ket{\psi_{i,a}}^{\rm cond}\) denotes the normalized state obtained after
conditioning on the event \(a_{i-1}=2,a_i=1_\pm,a_{i+1}=0\), i.e.
\begin{equation}
\ket{\psi_{i,a}}^{\rm cond}:= \frac{\Pi^2_{i-1} \Pi^0_{i+1} \ket{\psi}}
{\sqrt{P(a_{i-1}=2,a_{i+1}=0)}}.
\end{equation}
On the local
conditional space of party \(i\), this state has the same reduced support as
the component \(\ket{\psi_a}^{\rm cond}\). Hence,
\begin{equation}
\bra{\psi_a}^{\rm cond}
\frac{\id+\mathcal{Z}_i}{2}
\ket{\psi_a}^{\rm cond}
=
1 .
\end{equation}
Therefore,
\begin{equation}
\bra{\psi_a}^{\rm cond}
\mathcal{Z}_i
\ket{\psi_a}^{\rm cond}
=
1 .
\end{equation}

Similarly, the conditional probability
\begin{equation}
P(a_i=1_+ \mid a_{i-1}=0,a_i=1_\pm,a_{i+1}=2,x_i=0)=0
\end{equation}
gives
\begin{align}
0
&=
\frac{
P(a_i=1_+,a_{i-1}=0, a_{i+1}=2 \mid x_i=0)
}{
P(a_{i-1}=0,a_i=1_\pm,a_{i+1}=2 \mid x_i=0)
}
\nonumber\\
&=
\frac{
\bra{\psi}
\Pi^0_{i-1}
\frac{\id+\mathcal{Z}_i}{2}
\Pi^2_{i+1}
\ket{\psi}
}{
P(a_{i-1}=0,a_i=1_\pm,a_{i+1}=2 \mid x_i=0)
}
\nonumber\\
&=
\bra{\psi_{i,c}}^{\rm cond}
\frac{\id+\mathcal{Z}_i}{2}
\ket{\psi_{i,c}}^{\rm cond},
\end{align}
where 
\begin{equation}
\ket{\psi_{i,c}}^{\rm cond}:= \frac{\Pi^0_{i-1} \Pi^2_{i+1} \ket{\psi}}
{\sqrt{P(a_{i-1}=0,a_{i+1}=2)}}.
\end{equation}
Again, from the local perspective of party \(i\),
\(\ket{\psi_{i,c}}^{\rm cond}\) has the same reduced support as
\(\ket{\psi_c}^{\rm cond}\). Thus,
\begin{equation}
\bra{\psi_c}^{\rm cond}
\frac{\id+\mathcal{Z}_i}{2}
\ket{\psi_c}^{\rm cond}
=
0,
\end{equation}
or equivalently,
\begin{equation}
\bra{\psi_c}^{\rm cond}
\mathcal{Z}_i
\ket{\psi_c}^{\rm cond}
=
-1 .
\end{equation}

We have therefore obtained
\begin{align}
\bra{\psi_a}^{\rm cond}
\mathcal{Z}_i
\ket{\psi_a}^{\rm cond}
&=1,\\
\bra{\psi_c}^{\rm cond}
\mathcal{Z}_i
\ket{\psi_c}^{\rm cond}
&=-1.
\end{align}
This implies
\begin{align}
\bra{\psi_a}^{\rm cond}
(\id-\mathcal{Z}_i)
\ket{\psi_a}^{\rm cond}
&=0,\\
\bra{\psi_c}^{\rm cond}
(\id+\mathcal{Z}_i)
\ket{\psi_c}^{\rm cond}
&=0.
\end{align}
Since \(\mathcal{Z}_i\) is Hermitian and
\( \id \pm \mathcal{Z}_i \geq 0\), we can write
\begin{equation}
\begin{aligned}
&\left(\id-\mathcal{Z}_i\right)^{1/2}
\ket{\psi_{a}}^{\rm cond}=0,\\
&\left(\id+\mathcal{Z}_i\right)^{1/2}
\ket{\psi_{c}}^{\rm cond}=0 .
\end{aligned}    
\end{equation}
Multiplying by \(\left(\id-\mathcal{Z}_i\right)^{1/2}\) and
\(\left(\id+\mathcal{Z}_i\right)^{1/2}\), respectively, we obtain
\begin{equation}
\begin{aligned}
&\mathcal{Z}_i \ket{\psi_{a}}^{\rm cond}
=
\ket{\psi_{a}}^{\rm cond},\\
&\mathcal{Z}_i \ket{\psi_{c}}^{\rm cond}
=
-\ket{\psi_{c}}^{\rm cond}.
\end{aligned}
\end{equation}

Thus, \(\mathcal{Z}_i\) has eigenvalue \(+1\) on the component whose system part is
\(\ket{01}_{S_{i_R}S_{i_L}}\), and eigenvalue \(-1\) on the component whose system
part is \(\ket{10}_{S_{i_R}S_{i_L}}\). To express this action on the local
conditional support, define the junk supports
\begin{align}
J^a_{i_R}
&:=
\operatorname{supp}\!\left(
\operatorname{Tr}_{J_{(i-1)_L}}
\ketbra{\zeta^a_{i-1}}{\zeta^a_{i-1}}
\right),\\
J^a_{i_L}
&:=
\operatorname{supp}\!\left(
\operatorname{Tr}_{J_{(i+1)_R}}
\ketbra{\zeta^a_i}{\zeta^a_i}
\right),\\
J^c_{i_R}
&:=
\operatorname{supp}\!\left(
\operatorname{Tr}_{J_{(i-1)_L}}
\ketbra{\zeta^c_{i-1}}{\zeta^c_{i-1}}
\right),\\
J^c_{i_L}
&:=
\operatorname{supp}\!\left(
\operatorname{Tr}_{J_{(i+1)_R}}
\ketbra{\zeta^c_i}{\zeta^c_i}
\right).
\end{align}
On the relevant local conditional support
\begin{equation}
\left(
\ket{01}_{S_{i_R}S_{i_L}}
\otimes J^a_{i_R}\otimes J^a_{i_L}
\right)
\oplus
\left(
\ket{10}_{S_{i_R}S_{i_L}}
\otimes J^c_{i_R}\otimes J^c_{i_L}
\right),
\end{equation}
we therefore have
\begin{equation}
\mathcal{Z}_i
=
Z_{S_i}\otimes \id_{J_i},
\end{equation}
where
\begin{equation}
Z_{S_i}
=
\left(
\ketbra{01}{01}
-
\ketbra{10}{10}
\right)_{S_{i_R}S_{i_L}} .
\end{equation}

\paragraph{Case \(i=1\).}
It remains to treat party \(i=1\). In this case, the previous positivity
argument does not apply directly, because the operator \(\mathcal{Z}_1\) is
defined through the conditional observables in
\cref{eq:observableToXZ}. In particular, \(\id\pm\mathcal{Z}_1\) need not be
positive.

For \(x_1=0\), the relevant second-round observable is
\[
\frac{\mathcal{X}_1+\mathcal{Z}_1}{\sqrt{2}}.
\]
Using the same type of conditional probability terms as above, we obtain
\begin{align}
{u_{1,0}^{+}}^2
&=
P(a_1=1_+ \mid a_1=1_\pm,a_2=0,a_N=2,x_1=0)
\nonumber\\
&=
\frac{
P(a_1=1_+,a_2=0,a_N=2 \mid x_1=0)
}{
P(a_1=1_\pm,a_2=0,a_N=2 \mid x_1=0)
}
\nonumber\\
&=
\frac{
\bra{\psi}
\frac{
\id+\frac{\mathcal{X}_1+\mathcal{Z}_1}{\sqrt{2}}
}{2}
\Pi^0_{2}
\Pi^2_{N}
\ket{\psi}
}{
P(a_1=1,a_2=0,a_N=2 \mid x_1=0)
}
\nonumber\\
&=
\bra{\psi_{1,a}}^{\rm cond}
\left(
\frac{\id}{2}
+
\frac{\mathcal{X}_1+\mathcal{Z}_1}{2\sqrt{2}}
\right)
\ket{\psi_{1,a}}^{\rm cond}
\nonumber\\
&=
\bra{\psi_a}^{\rm cond}
\left(
\frac{\id}{2}
+
\frac{\mathcal{X}_1+\mathcal{Z}_1}{2\sqrt{2}}
\right)
\ket{\psi_a}^{\rm cond},
\end{align}

\begin{align}
{v_{1,0}^{+}}^2
&=
P(a_1=1_+ \mid a_1=1_\pm,a_2=2,a_N=0,x_1=0)
\nonumber\\
&=
\frac{
P(a_1=1_+,a_2=2,a_N=0 \mid x_1=0)
}{
P(a_1=1_\pm,a_2=2,a_N=0 \mid x_1=0)
}
\nonumber\\
&=
\frac{
\bra{\psi}
\frac{
\id+\frac{\mathcal{X}_1+\mathcal{Z}_1}{\sqrt{2}}
}{2}
\Pi^2_{2}
\Pi^0_{N}
\ket{\psi}
}{
P(a_1=1_\pm,a_2=2,a_N=0 \mid x_1=0)
}
\nonumber\\
&=
\bra{\psi_{1,c}}^{\rm cond}
\left(
\frac{\id}{2}
+
\frac{\mathcal{X}_1+\mathcal{Z}_1}{2\sqrt{2}}
\right)
\ket{\psi_{1,c}}^{\rm cond}
\nonumber\\
&=
\bra{\psi_c}^{\rm cond}
\left(
\frac{\id}{2}
+
\frac{\mathcal{X}_1+\mathcal{Z}_1}{2\sqrt{2}}
\right)
\ket{\psi_c}^{\rm cond}.
\end{align}
Here \(\ket{\psi_{1,a}}^{\rm cond}\) and
\(\ket{\psi_{1,c}}^{\rm cond}\) denote the normalized states obtained after the
corresponding conditioning events. As in the case \(i\neq 1\), these states have
the same reduced support on party \(1\) as
\(\ket{\psi_a}^{\rm cond}\) and \(\ket{\psi_c}^{\rm cond}\), respectively.

Similarly, for \(x_1=1\), the relevant conditional observable is
\[
\frac{\mathcal{X}_1-\mathcal{Z}_1}{\sqrt{2}},
\]
and hence
\begin{align}
{u_{1,1}^{+}}^2
&=
\bra{\psi_a}^{\rm cond}
\left(
\frac{\id}{2}
+
\frac{\mathcal{X}_1-\mathcal{Z}_1}{2\sqrt{2}}
\right)
\ket{\psi_a}^{\rm cond},
\\
{v_{1,1}^{+}}^2
&=
\bra{\psi_c}^{\rm cond}
\left(
\frac{\id}{2}
+
\frac{\mathcal{X}_1-\mathcal{Z}_1}{2\sqrt{2}}
\right)
\ket{\psi_c}^{\rm cond}.
\end{align}
Substituting the coefficient values from~\cref{eq:meas_coeff_forA} gives
\begin{align}
\bra{\psi_a}^{\rm cond}
\mathcal{Z}_1
\ket{\psi_a}^{\rm cond}
&=1,\\
\bra{\psi_c}^{\rm cond}
\mathcal{Z}_1
\ket{\psi_c}^{\rm cond}
&=-1.
\end{align}

Although \(\id\pm\mathcal{Z}_1\) need not be positive on the full Hilbert
space, \cref{lem:FirstBell} or \cref{lem:SecondBell} implies that
\[
\mathcal{Z}_1^2\ket{\psi}^{\rm cond}
=
\ket{\psi}^{\rm cond}.
\]
Moreover, \(\mathcal{Z}_1\) acts locally on party \(1\) and therefore acts
trivially on the system parts of the states of the other parties. Since
\[
\ket{\psi}^{\rm cond}
=
\frac{1}{\sqrt{2}}
\left(
\ket{\psi_a}^{\rm cond}
+
\ket{\psi_c}^{\rm cond}
\right),
\]
and the two components are orthogonal already on the system parts of the other parties, the above relation must hold separately on each component. Hence,
\begin{equation}
\begin{aligned}
\mathcal{Z}_1^2\ket{\psi_a}^{\rm cond}
&=
\ket{\psi_a}^{\rm cond},\\
\mathcal{Z}_1^2\ket{\psi_c}^{\rm cond}
&=
\ket{\psi_c}^{\rm cond}.
\end{aligned}
\end{equation}
Therefore,
\begin{align}
\bra{\psi_a}^{\rm cond}
(\id-\mathcal{Z}_1)^2
\ket{\psi_a}^{\rm cond}
&=
2\bra{\psi_a}^{\rm cond}
(\id-\mathcal{Z}_1)
\ket{\psi_a}^{\rm cond}
=
0,\\
\bra{\psi_c}^{\rm cond}
(\id+\mathcal{Z}_1)^2
\ket{\psi_c}^{\rm cond}
&=
2\bra{\psi_c}^{\rm cond}
(\id+\mathcal{Z}_1)
\ket{\psi_c}^{\rm cond}
=
0.
\end{align}
It follows that
\begin{align}
(\id-\mathcal{Z}_1)\ket{\psi_a}^{\rm cond}
&=0,\\
(\id+\mathcal{Z}_1)\ket{\psi_c}^{\rm cond}
&=0,
\end{align}
and hence
\begin{align}
\mathcal{Z}_1\ket{\psi_a}^{\rm cond}
&=
\ket{\psi_a}^{\rm cond},\\
\mathcal{Z}_1\ket{\psi_c}^{\rm cond}
&=
-\ket{\psi_c}^{\rm cond}.
\end{align}

The same identification of the relevant local conditional support then yields
\begin{equation}
\mathcal{Z}_1
=
Z_{S_1}\otimes \id_{J_1}
\qquad
\text{on } H_1^{\rm cond},
\end{equation}
where
\begin{equation}
Z_{S_1}
=
\left(
\ketbra{01}{01}
-
\ketbra{10}{10}
\right)_{S_{1_R}S_{1_L}} .
\end{equation}

Combining the two cases, we conclude that for every
\(i\in\{1,\dots,N\}\),
\begin{equation}
\mathcal{Z}_i
=
Z_{a_i}\otimes \id_{J_i}
\qquad
\text{on } H_i^{\rm cond}.
\end{equation}
This proves the claim.
\end{proof}
\subsection{\texorpdfstring{\cref{lem:Xsep_SysJunk}}{X-system-junk separation lemma}}
\label{appendix:Xsep_SJ_3}

We prove that the constraints
\begin{equation}
\begin{aligned}
&\mathcal{X}_i,\mathcal{Z}_i:
H_i^{\rm cond}
\rightarrow
H_i^{\rm cond}\oplus {H_i}^{\rm cond}_{\perp},\\
&\mathcal{X}_i^{2}=\mathcal{Z}_i^{2}=\id,\\
&\{\mathcal{X}_i,\mathcal{Z}_i\}=0,\\
&\mathcal{Z}_i=Z_{S_i}\otimes \id_{J_i},
\end{aligned}
\end{equation}
imply that, on the conditional subspace \(H_i^{\rm cond}\), the operator
\(\mathcal{X}_i\) has the form
\begin{equation}
\label{eq:Xsep_appendix}
\begin{aligned}
\mathcal{X}_i &: H_i^{\rm cond}
\rightarrow
H_i^{\rm cond}\oplus {H_i}^{\rm cond}_{\perp},\\
\mathcal{X}_i
&=
\frac{X+\ii Y}{2}\otimes T_i
+
\frac{X-\ii Y}{2}\otimes T_i^\dagger,
\end{aligned}
\end{equation}
where
\begin{equation}
X
=
\left(
\ket{10}\bra{01}
+
\ket{01}\bra{10}
\right)_{S_{i_R}S_{i_L}},
\qquad
Y
=
\left(
\ii\ket{10}\bra{01}
-
\ii\ket{01}\bra{10}
\right)_{S_{i_R}S_{i_L}},
\end{equation}
and \(T_i\) is a unitary map
\begin{equation}
T_i:
J^c_{i_L}\otimes J^c_{i_R}
\longrightarrow
J^a_{i_L}\otimes J^a_{i_R}.
\end{equation}

\begin{proof}
By \cref{lem:X_HtoH}, the operator \(\mathcal{X}_i\) maps
\(H_i^{\rm cond}\) into itself. Nevertheless, within the conditional subspace,
it may act nontrivially on both system part and the junk part.
We decompose \(H_i^{\rm cond}\) according to the system basis
\(\{\ket{01},\ket{10}\}_{S_{i_R}S_{i_L}}\). Since \(\mathcal{X}_i\) is
Hermitian, it can be written in block form as
\begin{equation}
\mathcal{X}_i =
\begin{pNiceMatrix}[first-row,last-col]
\ket{01} & \ket{10} & \\
R_i & T_i & \bra{01} \\
T_i^{\dagger} & P_i & \bra{10}
\end{pNiceMatrix}\;,
\end{equation}
where the blocks act on the corresponding junk spaces. More explicitly,
\begin{align}
R_i &: J^a_{i_L}\otimes J^a_{i_R}
\rightarrow
J^a_{i_L}\otimes J^a_{i_R},
\label{eq:R_i_domain}\\
T_i &: J^c_{i_L}\otimes J^c_{i_R}
\rightarrow
J^a_{i_L}\otimes J^a_{i_R},
\label{eq:T_i_domain}\\
P_i &: J^c_{i_L}\otimes J^c_{i_R}
\rightarrow
J^c_{i_L}\otimes J^c_{i_R}.
\label{eq:P_i_domain}
\end{align}

Using
\[
Z_{S_i}
=
\left(
\ketbra{01}{01}
-
\ketbra{10}{10}
\right)_{S_{i_R}S_{i_L}},
\]
we have, in the same block decomposition,
\begin{align}
\mathcal{Z}_i \mathcal{X}_i
&=
\left(Z_{S_i}\otimes \id_{J_i}\right)
\begin{pmatrix}
R_i & T_i \\
T_i^{\dagger} & P_i
\end{pmatrix}
=
\begin{pmatrix}
R_i & T_i \\
-T_i^{\dagger} & -P_i
\end{pmatrix}, \\
\mathcal{X}_i \mathcal{Z}_i
&=
\begin{pmatrix}
R_i & T_i \\
T_i^{\dagger} & P_i
\end{pmatrix}
\left(Z_{S_i}\otimes \id_{J_i}\right)
=
\begin{pmatrix}
R_i & -T_i \\
T_i^{\dagger} & -P_i
\end{pmatrix}.
\end{align}
Therefore,
\begin{equation}
\{\mathcal{X}_i,\mathcal{Z}_i\}
=
\begin{pmatrix}
2R_i & 0 \\
0 & -2P_i
\end{pmatrix}.
\end{equation}
The anticommutation relation
\(\{\mathcal{X}_i,\mathcal{Z}_i\}=0\) implies
\begin{equation}
R_i=0,
\qquad
P_i=0.
\end{equation}
Hence, on \(H_i^{\rm cond}\),
\begin{equation}
\mathcal{X}_i
=
\begin{pmatrix}
0 & T_i \\
T_i^{\dagger} & 0
\end{pmatrix}.
\end{equation}

We now impose \(\mathcal{X}_i^2=\id\). This gives
\begin{align}
\mathcal{X}_i^2
&=
\begin{pmatrix}
0 & T_i \\
T_i^{\dagger} & 0
\end{pmatrix}
\begin{pmatrix}
0 & T_i \\
T_i^{\dagger} & 0
\end{pmatrix}
=
\begin{pmatrix}
T_iT_i^{\dagger} & 0 \\
0 & T_i^{\dagger}T_i
\end{pmatrix}.
\end{align}
Thus,
\begin{equation}
T_iT_i^{\dagger}
=
T_i^{\dagger}T_i
=
\id,
\end{equation}
so \(T_i\) is unitary from
\(J^c_{i_L}\otimes J^c_{i_R}\) to
\(J^a_{i_L}\otimes J^a_{i_R}\).

Finally, observe that, in the basis
\(\{\ket{01},\ket{10}\}_{S_{i_R}S_{i_L}}\),
\begin{equation}
\frac{X+\ii Y}{2}
=
\ket{01}\bra{10},
\qquad
\frac{X-\ii Y}{2}
=
\ket{10}\bra{01}.
\end{equation}
Therefore,
\begin{equation}
\mathcal{X}_i
=
\frac{X+\ii Y}{2}\otimes T_i
+
\frac{X-\ii Y}{2}\otimes T_i^\dagger,
\end{equation}
with \(T_i\) unitary. This proves the claim.
\end{proof}
\subsection{\texorpdfstring{\cref{lem:JunkRLSep}}{Junk-RL separation lemma}}
\label{appendix:Tsep_RL_3}

Here, we generalize~\cite[Lemmas~1o and~1e]{sekatski2023partial} to an infinite-dimensional setting, which contributes to the proof of \cref{lem:JunkRLSep}.
This generalization allows us to state our main self-testing result, \cref{thm:fullselftesting}, without any dimensionality assumptions.

Let $|\zeta\rangle_{JJ'}$ be a non-zero vector in a bipartite Hilbert space $\mathcal H_{JJ'}=\mathcal H_J\otimes \mathcal H_{J'}$ where both $\cH_J, \cH_{J'}$ are separable. 
We can then consider the Schmidt decomposition
\begin{equation*}
|\zeta\rangle_{JJ'} = \sum_{i\in \mathbb N} \lambda_{i} |i\rangle_{J} |i\rangle_{J'}.
\end{equation*}
Here, without loss of generality, we assume that $\lambda_{i} > 0$ for all $i$, and $\{\ket i_J:\, i\in \mathbb N\}$ and $\{\ket i_{J'}:\, i\in \mathbb N\}$ are orthonormal bases in $\cH_J$ and $\cH_{J'}$, respectively. 

For any two bounded operators $X, Y: \mathcal{H}_{J} \rightarrow \mathcal{H}_{J}$, we define their inner product relative to the state $|\zeta\rangle_{JJ'}$ by
\begin{align*}
\langle X, Y \rangle_{J} &:= \langle\zeta| (X^\dagger Y \otimes \id_{J'}) |\zeta\rangle= \sum_{i} \lambda_{i}^2 \langle i| X^\dagger Y |i\rangle,
\end{align*}
with the induced norm given by
\begin{equation*}
\|X\|_{J}^2 = \langle X, X \rangle_{J} = \sum_{i} \lambda_{i}^2 \|X|i\rangle\|^2.
\end{equation*}
It is straightforward to verify that this is indeed a valid inner product and $\|X\|_{J}^2 < +\infty$ if $X$ is bounded. Based on this inner product, we define $\mathcal O_J$ to be the Hilbert space of the completion of the space of bounded operators under this norm:
\begin{equation*}
\mathcal{O}_{J} := \overline{\left\{ X \;\middle|\; X : \mathcal{H}_{J} \rightarrow \mathcal{H}_{J} \text{ bounded  } \right\}}.
\end{equation*}
We emphasize that by definition, $\mathcal{O}_{J}$ is a Hilbert space, which possibly contains unbounded operators. 
Moreover, it is separable since $\mathcal{H}_{J}$ is separable.

We can similarly define the Hilbert space $\mathcal{O}_{J'}$ of operators acting on $\cH_{J'}$. Then, as will be shown below, $\mathcal O_{J'}$ is isomorphic to $\mathcal O_J$. 

Let $Y: \mathcal{H}_{J'} \rightarrow \mathcal{H}_{J'}$ be a bounded operator. We define the linear map $\Gamma(Y) = \tilde{Y}: \mathcal{H}_{J} \rightarrow \mathcal{H}_{J}$ by considering its action on basis vectors given by 
\begin{equation*}
 \tilde{Y} | i \rangle = \frac{1}{\lambda_{i}} \sum_j \lambda_{j}\langle j | Y^\dagger | i \rangle^* \ket j_J= \frac{1}{\lambda_{i}}\sum_j \lambda_{j} \langle i | Y | j \rangle \ket j_J.
\end{equation*}
Evaluating the norm of $\tilde{Y}$ we have
\begin{align*}
\|\tilde{Y}\|_{J}^2 &= \sum_{i} \lambda_{i}^2 \|\tilde{Y}|i\rangle\|^2 \\
&= \sum_{i} \lambda_{i}^2 \left( \frac{1}{\lambda_{i}^2} \sum_{j} \lambda_{j}^2 |\langle i|Y|j\rangle|^2 \right) \\
&= \sum_{i,j} \lambda_{j}^2 |\langle i|Y|j\rangle|^2 \\
&= \sum_{j} \lambda_{j}^2 \|Y|j\rangle\|^2 = \langle\zeta| \id_J \otimes Y^\dagger Y |\zeta\rangle = \|Y\|_{J'}^2.
\end{align*}
Moreover, this mapping satisfies 
\begin{align*}
(\tilde{Y} \otimes \id) |\zeta\rangle_{JJ'} &= \sum_{i} \lambda_{i} (\tilde{Y}|i\rangle) \otimes |i\rangle_{J'} \\
&= \sum_{i} \sum_{j} \lambda_{j} \langle j|Y|i\rangle |j\rangle_{J} \otimes |i\rangle_{J'} \\
&= \sum_{j} \lambda_{j} |j\rangle_{J} \otimes (Y|j\rangle_{J'}) \\
&= (\id \otimes Y) |\zeta\rangle_{JJ'},
\end{align*}
for bounded $Y$. Extending $\Gamma$ by continuity to the whole $\mathcal O_{J'}$, then $\Gamma$ is an isometry from $\mathcal{O}_{J'}$ to $\mathcal{O}_{J}$ since an isometry on a dense subspace extends to the completion. 

The next step is to understand the behavior of these Hilbert spaces of operators under tensor product. 
Let $|\zeta_1\rangle_{J_1J'_1}$ and $|\zeta_{2}\rangle_{J_2J'_2}$ be two bipartite states. Then, we can consider the Hilbert spaces $\mathcal{O}_{J_1}, \mathcal{O}_{J'_1}, \mathcal{O}_{J_2}$ and $\mathcal{O}_{J'_2}$ as above. 
We can also consider the tensor product state
\begin{equation*}
|\psi\rangle_{KK'} = |\zeta_1\rangle_{J_1J'_1} \otimes |\zeta_{2}\rangle_{J_2J'_2},
\end{equation*}
with the bipartition $K = J_1J_2$ and $K' = J'_1J'_2$. Then, we can construct the Hilbert spaces $\mathcal O_K, \mathcal O_{K'}$, which satisfy
\begin{equation*}
\mathcal O_K=\mathcal{O}_{J_1J_2} = \mathcal{O}_{J_1} \otimes \mathcal{O}_{J_2}, \qquad \mathcal O_{K'}=\mathcal{O}_{J'_1J'_2} = \mathcal{O}_{J'_1} \otimes \mathcal{O}_{J'_2},
\end{equation*}
with the canonical Hilbert-space tensor product.
In particular, for $X\in \mathcal O_{J_1}$ and $Y\in \mathcal O_{J_2}$ we have 
\begin{equation*}
\|X \otimes Y\|_{K}^2 = \|X\|_{J_1}^2 \cdot \|Y\|_{J_2}^2.
\end{equation*}
More importantly, the isometry $\Gamma_K:\mathcal O_{J'_1J'_2}\to \mathcal O_{J_1J_2}$ is easily verified to be of the form
$$\Gamma_K = \Gamma_{J_1}\otimes \Gamma_{J_2}.$$
Furthermore, for any bounded operator $X_{J'_1J'_2}\in \mathcal O_{J'_1J'_2}$, we have
\begin{equation}\label{eq:X-zeta-12-isometry}
(\id_{J_1J_2} \otimes X) |\zeta_1\rangle |\zeta_{2}\rangle = (\tilde{X} \otimes \id_{J'_1J'_2}) |\zeta_{1}\rangle |\zeta_{2}\rangle,
\end{equation}
where $\tilde{X}=\Gamma_K(X) \in \mathcal O_{J_1J_2}$.

\begin{lemma}\label{lem:ProductStructureLemmaInfDim}
Suppose that $\ket{\zeta_1}_{J_1J'_1}, \dots, \ket{\zeta_n}_{J_nJ'_n}$ be $n$ bipartie quantum states and we have
\begin{equation*}
\left(\bigotimes_{i:\text{ odd}}U_{i}\right) |\zeta_{1}\rangle \otimes \dots \otimes |\zeta_{n}\rangle = \left(\bigotimes_{i:\text{ even}}U_{i}\right) |\zeta_{1}\rangle \otimes \cdots \otimes |\zeta_{n}\rangle,
\end{equation*}
where $U_{i}$, for any $1\leq i\leq n$, is a unitary operator acting on $\mathcal{H}_{J_{i}'} \otimes \mathcal{H}_{J_{i+1}}$ with indices taken modulo $n$.
Then, each $U_i$ is of product form, i.e., product of unitaries acting on $\mathcal{H}_{J_{i}'}$ and $\mathcal{H}_{J_{i+1}}$.
\end{lemma}

\begin{proof}
The first step is to apply the isometry $\Gamma$ constructed above on $U_i$'s with even indices $i$ to obtain $\tilde U_i$. Then, using~\cref{eq:X-zeta-12-isometry}, the identity in the assumption can be written as 
\begin{equation}\label{eq:U-odd-even-zeta}
(U_{1} \otimes U_{3} \otimes \cdots) |\zeta_{1}\rangle \otimes \cdots \otimes |\zeta_{n}\rangle = (\tilde{U}_{2} \otimes \tilde{U}_{4} \otimes \cdots) |\zeta_{1}\rangle \otimes \cdots \otimes |\zeta_{n}\rangle.
\end{equation}
We note that $\tilde{U}_{2i}$ acts on $\mathcal{H}_{J_{2i}} \otimes \mathcal{H}_{J_{2i+1}'}$. Moreover, in the above equation, the operators $U_i, \tilde U_i$ act nontrivially only on subsystems $J'_i$ with odd $i$ and $J_i$ with even $i$. This motivates a relabeling of the underlying subsystems to 
\begin{equation*}
S_{i} = \begin{cases}
J_{i}' & \text{if } i \text{ is odd} \\
J_{i} & \text{if } i \text{ is even}
\end{cases}
\end{equation*}
Under this convention, the odd unitaries $U_{2k+1}$ act on $\mathcal{H}_{S_{2k+1}} \otimes \mathcal{H}_{S_{2k+2}}$, while the transformed even unitaries $\tilde{U}_{2k}$ act on $\mathcal{H}_{S_{2k}} \otimes \mathcal{H}_{S_{2k+1}}$. Moreover, since full rank Schmidt vectors are separating, \cref{eq:U-odd-even-zeta} can equivalently be written as
$$U_{1} \otimes U_{3} \otimes \cdots = \tilde{U}_{2} \otimes \tilde{U}_{4} \otimes \cdots.$$
We emphasize that  $U_{2k+1}$ belongs to $\mathcal{O}_{S_{2k+1}, S_{2k+2}}$ and $\tilde U_{2k}$ belongs to $\mathcal{O}_{S_{2k}, S_{2k+1}}$. Moreover, the above identity should be understood inside the tensor product Hilbert space $\mathcal O_{S_1}\otimes \cdots\otimes \mathcal O_{S_n}$. This observation allows us to use some basic tools from linear algebra such as expanding a vector in an orthonormal basis.

It is more convenient to split the rest of the proof into two cases of even of odd $n$.
In the following, product $U_{1} U_{3} \dots U_{2m-1}$ is to be understood as tensor products, for simplicity.

\subsubsection*{Case 1: Even number of subsystems ($n = 2m$)}
In this case, our identity takes the form
\begin{equation*}
U_{1} U_{3} \dots U_{2m-1} = \tilde{U}_{2} \tilde{U}_{4} \dots \tilde{U}_{2m}.
\end{equation*}
Let $\{X_{S_{k}, i} : i\in \mathbb N\}$ be an orthonormal basis for the Hilbert space $\mathcal{O}_{S_{k}}$ thanks to its separability. We can decompose the odd unitaries along this basis as
\begin{equation*}
U_{2k-1} = \sum_{i_{k}} Y_{S_{2k-1}, i_{k}} \otimes X_{S_{2k}, i_{k}},
\end{equation*}
for some $Y_{S_{2k-1}, i_{k}} \in \mathcal{O}_{S_{2k-1}}$.
Similarly, for the transformed even unitaries we can write
\begin{equation*}
\tilde{U}_{2k} = \sum_{j_{k}} X_{S_{2k}, j_{k}} \otimes Z_{S_{2k+1}, j_{k}},
\end{equation*}
where $Z_{S_{2k+1}, j_{k}} \in \mathcal{O}_{S_{2k+1}}$.
Substituting these expansions into the previous equation yields
\begin{align*}
\sum_{i_{1}, \dots, i_{m}} Y_{S_{1}, i_{1}} \otimes X_{S_{2}, i_{1}} \otimes Y_{S_{3}, i_{2}} \dots \otimes Y_{S_{2m-1}, i_{m}} \otimes X_{S_{2m}, i_{m}} \nonumber \\
\qquad = \sum_{j_{1}, \dots, j_{m}} X_{S_{1}, j_{1}} \otimes Z_{S_{2}, j_{1}} \otimes \dots \otimes X_{S_{2m}, j_{m}} \otimes Z_{S_{2m+1}, j_{m}}.
\end{align*}
Since these expressions lie within the well-defined tensor product Hilbert space $\mathcal{O}_{S_{1}} \otimes \dots \otimes \mathcal{O}_{S_{2m}}$, matching the independent basis components implies that
\begin{equation*}
Y_{S_{1}, i_{1}} \otimes Y_{S_{3}, i_{2}} \dots \otimes Y_{S_{2m-1}, i_{m}} = Z_{S_{3}, j_{1}} \otimes Z_{S_{5}, j_{2}} \dots \otimes Z_{S_{2m+1}, j_{m}},
\end{equation*}
for all indices $i_1, \dots, i_m, j_1, \dots, j_m$. Since both sides are of tensor product form, this means that for any $k$ and all indices $i, j$
the operators $Y_{S_{2k-1}, i}$ and $Z_{S_{2k-1}, j}$ are proportional. Therefore, for any $k$, and any $i, i'$ the operators  $Y_{S_{2k-1}, i}, Y_{S_{2k-1}, i'}$ are proportional. Using this in the expansion $U_{2k-1} = \sum_{i_{k}} Y_{S_{2k-1}, i_{k}} \otimes X_{S_{2k}, i_{k}}$, we realize that $U_{2k-1}$ is a product operator as desired. 
Indeed, this shows that $U_{2k-1} =  Y'_{2k-1}\otimes X'_{2k}$ for operators $Y'_{2k-1}\in \mathcal O_{S_{2k-1}}$ and $X'_{2k}\in \mathcal O_{S_{2k}}$. Next, using the fact that $U_{2k-1}$ is bounded and unitary, it is readily verified that $Y'_{2k-1}, X'_{2k}$ are also bounded and can be taken to be unitaries. By a similar argument, $\tilde U_{2k}$ and hence $U_{2k}$ are product operators as desired. 

\subsubsection*{Case 2: Odd Number of Subsystems ($n = 2m+1$)}
When the number of systems is odd, we reduce the problem to the even case by merging two subsystems. To this end, we define a combined system and a grouped unitary as
\begin{equation*}
S_{1}' = S_{1} S_{2m+1}, \qquad U_{1}' = U_{1} U_{2m+1}
\end{equation*}
The starting identity then becomes
\begin{align*}
\tilde{U}_{2} \cdots \tilde{U}_{2m} = U_{1} U_{3} \cdots U_{2m+1} 
= (U_{1} U_{2m+1}) U_{3} \dots U_{2m-1} = U_{1}' U_{3} \dots U_{2m-1},
\end{align*}
which takes us back to the even case. Therefore, applying the result from the previous case, it follows that $\tilde{U}_{2}, \dots, \tilde{U}_{2m}$ and $U_{1}', U_{3}, \dots, U_{2m-1}$ are all product operators. In particular, for the merged operator we have
\begin{equation*}
U_{1} U_{2m+1} = U_{1}' = V_{S_{1}S_{2m+1}} \otimes W_{S_{2}},
\end{equation*}
for some unitaries $V_{S_{1}S_{2m+1}}, W_{S_{2}}$. 
This reveals that 
\begin{equation*}
U_{1} = (U_{2m+1}^\dagger V_{S_{1}S_{2m+1}}) \otimes W_{S_{2}}, 
\end{equation*}
is a product operator as $U_1$ acts only on $S_1$ and $S_2$. By symmetry, $U_{2m+1}$ is also a product operator, completing the proof for all cases.

\end{proof}

\end{document}